\documentclass[11pt,draftclsnofoot,peerreviewca,letterpaper,onecolumn]{IEEEtran}
\usepackage{cite}

\usepackage{graphicx,color,epsfig,rotating,subfigure}
\usepackage{amsfonts,amsmath,amssymb}
\usepackage{algorithm,algorithmic}
\usepackage{subfigure}

\long\def\comment#1{}

\newfont{\bbb}{msbm10 scaled 700}

\newfont{\bb}{msbm10 scaled 1100}
\newcommand{\CC}{\mbox{\bb C}}
\newcommand{\PP}{\mbox{\bb P}}

\newcommand{\ZZ}{\mbox{\bb Z}}
\newcommand{\FF}{\mbox{\bb F}}

\newcommand{\EE}{\mbox{\bb E}}

\newcommand{\av}{{\bf a}}
\newcommand{\bv}{{\bf b}}
\newcommand{\cv}{{\bf c}}
\newcommand{\dv}{{\bf d}}

\newcommand{\hv}{{\bf h}}

\newcommand{\mv}{{\bf m}}

\newcommand{\rv}{{\bf r}}
\newcommand{\sv}{{\bf s}}
\newcommand{\tv}{{\bf t}}
\newcommand{\uv}{{\bf u}}
\newcommand{\wv}{{\bf w}}
\newcommand{\vv}{{\bf v}}
\newcommand{\xv}{{\bf x}}
\newcommand{\yv}{{\bf y}}
\newcommand{\zv}{{\bf z}}
\newcommand{\zerov}{{\bf 0}}
\newcommand{\onev}{{\bf 1}}

\newcommand{\Am}{{\bf A}}
\newcommand{\Bm}{{\bf B}}
\newcommand{\Cm}{{\bf C}}
\newcommand{\Dm}{{\bf D}}
\newcommand{\Em}{{\bf E}}
\newcommand{\Fm}{{\bf F}}
\newcommand{\Gm}{{\bf G}}
\newcommand{\Hm}{{\bf H}}
\newcommand{\Id}{{\bf I}}
\newcommand{\Jm}{{\bf J}}

\newcommand{\Lm}{{\bf L}}
\newcommand{\Mm}{{\bf M}}

\newcommand{\Qm}{{\bf Q}}

\newcommand{\Sm}{{\bf S}}
\newcommand{\Tm}{{\bf T}}
\newcommand{\Um}{{\bf U}}
\newcommand{\Wm}{{\bf W}}
\newcommand{\Vm}{{\bf V}}
\newcommand{\Xm}{{\bf X}}
\newcommand{\Ym}{{\bf Y}}
\newcommand{\Zm}{{\bf Z}}

\newcommand{\Ac}{{\cal A}}

\newcommand{\Cc}{{\cal C}}
\newcommand{\Dc}{{\cal D}}

\newcommand{\Ic}{{\cal I}}

\newcommand{\Lc}{{\cal L}}
\newcommand{\Mc}{{\cal M}}
\newcommand{\Nc}{{\cal N}}

\newcommand{\Pc}{{\cal P}}

\newcommand{\Sc}{{\cal S}}

\newcommand{\Wc}{{\cal W}}
\newcommand{\Vc}{{\cal V}}

\newcommand{\alphav}{\hbox{\boldmath$\alpha$}}

\newcommand{\lambdav}{\hbox{\boldmath$\lambda$}}

\newcommand{\zetav}{\hbox{\boldmath$\zeta$}}

\newcommand{\diag}{{\hbox{diag}}}
\renewcommand{\det}{{\hbox{det}}}
\newcommand{\trace}{{\hbox{tr}}}

\newcommand{\SNR}{{\sf SNR}}
\newcommand{\INR}{{\sf INR}}

\newcommand{\eqdef}{\stackrel{\Delta}{=}}

\newcommand{\herm}{{\sf H}}

\newcommand{\transp}{{\sf T}}

\newcommand{\BLUE}{\color[rgb]{0,0,0.90}}

\newtheorem{theorem}{Theorem}
\newtheorem{lemma}{Lemma}
\newtheorem{corollary}{Corollary}

\newtheorem{proof}{Proof}
\newtheorem{remark}{Remark}

\newtheorem{example}{Example}
\newcommand{\argmin}{\operatornamewithlimits{argmin}}

\begin{document}

\sloppy

\title{Structured Lattice Codes for Some Two-User Gaussian Networks with Cognition,
Coordination and Two Hops}

\author{\authorblockN{Song-Nam~Hong,~\IEEEmembership{Student Member,~IEEE,}
        and~Giuseppe~Caire,~\IEEEmembership{Fellow,~IEEE}}
\authorblockA{Department of Electrical Engineering, University of Southern California, Los Angeles, CA, USA}
\authorblockA{(e-mail: \{songnamh, caire\}$@$usc.edu)}}

\maketitle


\newpage

\begin{abstract}
We study a number of two-user interference networks with multiple-antenna transmitters/receivers (MIMO),
transmitter side information in the form of linear combinations (over an appropriate finite-field) of the information messages, and
two-hop relaying. We start with  a Cognitive Interference Channel (CIC) where one of the transmitters (non-cognitive)
has knowledge of a rank-1 linear combination of the two information messages, while the other transmitter (cognitive) has access to
a rank-2 linear combination of the same messages.  This is referred to as the Network-Coded CIC,
since such linear combination may be the result of some random linear network coding scheme implemented in the backbone
wired network. For such channel we develop an achievable region based on a few novel concepts: Precoded Compute and Forward (PCoF)
with Channel Integer Alignment (CIA),
combined with standard Dirty-Paper Coding. We also develop a capacity region outer bound and find the sum symmetric Generalized Degrees of Freedom (GDoF) of the
Network-Coded CIC. Through the GDoF characterization, we show that knowing ``mixed data'' (linear combinations of the information messages)
provides an {\em unbounded} spectral efficiency gain over the classical CIC counterpart, if the ratio (in dB) of signal-to-noise (SNR) to interference-to-noise (INR)
is larger than certain threshold.  Then, we consider a Gaussian relay network having the two-user MIMO IC as the main building block.
We use PCoF with CIA to convert the MIMO IC into a {\em deterministic} finite-field IC. Then, we use a linear precoding scheme over the finite-field
to eliminate interference in the finite-field domain. Using this unified approach, we derive the {\em symmetric} sum rate of the
two-user MIMO IC with coordination, cognition, and two-hops. We also provide finite-SNR results (not just degrees of freedom) which show that the proposed coding schemes are competitive against state-of-the-art interference avoidance based on orthogonal access, for standard randomly generated
Rayleigh fading channels.
\end{abstract}


\clearpage

\section{Introduction}\label{sec:intro}

Interference is one of the fundamental aspects of wireless communication networks.
Although the full characterization of the Interference Channel (IC) capacity is elusive, much progress has been made in recent years.
The capacity region of the two-user Gaussian IC was characterized within 1 bit, by using superposition coding with an appropriate power allocation
of the private and common message codewords, and by providing a new upper bounding technique (known as, Genie-aided bound)  \cite{Etkin}. Degrees of Freedom (DoF) results are obtained under the assumption of full channel knowledge for the two-user multiple input multiple output (MIMO) IC with arbitrary number of antennas at each node \cite{Jafar07}, for the $K$ user IC with time-varying or frequency-selective channels \cite{Cadambe-K}, for the $K$-user IC with constant channel coefficients \cite{Motahari}, and for the $K$-user IC with multiple antennas \cite{Gou-K}. Also, Generalized DoF (GDoF) results are found for the two-user MIMO IC \cite{Karmakar} and for {\em symmetric} $K$ user IC \cite{Jafar-G}. We refer the reader to the \cite{Jafar-B} for the further results of various interference networks.

In many practical communication systems, transmitters or receivers are not isolated.
For example, in cellular systems the base stations are connected via a wired backhaul network through which information messages and some form of channel state information or coordination can be shared \cite{Flanagan,Lin,Marict,Song-ISIT,Song-IT}.
In wired networks, routing is generally optimal only for the single-source single-destination case \cite{Ford}.
In the more general case of multiple sources and multiple destinations (multi-source multicasting), linear network coding is known to achieve
the min-cut max-flow bound \cite{Ahlswede}. In practice, random linear network coding is of particular interest for its simplicity.
In this case, intermediate nodes forward linear combinations of the incoming messages by randomly and independently choosing the coefficients
from an appropriate finite-field \cite{Ho}. Going back to the cellular systems case, if random linear network coding is
used in the backhaul network, the base stations obtain linear combination of the messages instead of individual messages.
If the backhaul link serving a given base station has capacity large enough, the rank (per unit time) of such linear combinations is equal to the number of
independent messages (per unit time), so that the base station knows all the messages. In contrast, if the backhaul link is a capacity bottleneck,
the rank (per unit time) is less than the number of independent messages (per unit time). In this case, the base station has access to ``mixed data'',
i.e., rank-deficient linear combinations of the messages. We refer to this model as the Network-Coded Cognitive IC (CIC).

An example of Network-Coded CIC is shown in Fig.~\ref{model3}, including a cellular BS and a home BS (e.g.,  a femtocell access point).
The cellular BS is connected to the data router, which generates both messages,
via a high capacity link supporting rate $2R_{0}$.
The home BS is connected to the same data router via lower capacity link supporting only rate $R_{0}$.
In this case, the data router sends two information messages to the cellular BS (equivalently, a rank-2 linear combination thereof)
and a rank-1 linear combination of the messages to the home BS.
In the case of routing, this linear combination has coefficients 0 and 1, reducing to the classical CIC, which has been extensively investigated in the literature
\cite{Maric,Wu,Jovicic}. In particular, the Gaussian CIC capacity region was approximately characterized within one bit in \cite{Rini}.
If general network coding is used instead of routing, the rank-1 linear combination has generally non-zero coefficients and therefore contains mixed data.
Notice that in the model of Fig.~\ref{model3} mixed data can be provided without violating the backhaul capacity constraint of $R_0$.
At this point, a natural question arises: {\em Does mixed data at the ``non-cognitive'' transmitter provide a capacity
increase ``for free'' (i.e., without any cost in terms of backhaul rate),
for the Network-Coded CIC over the  conventional CIC?}

\begin{figure}
\centerline{\includegraphics[width=14cm]{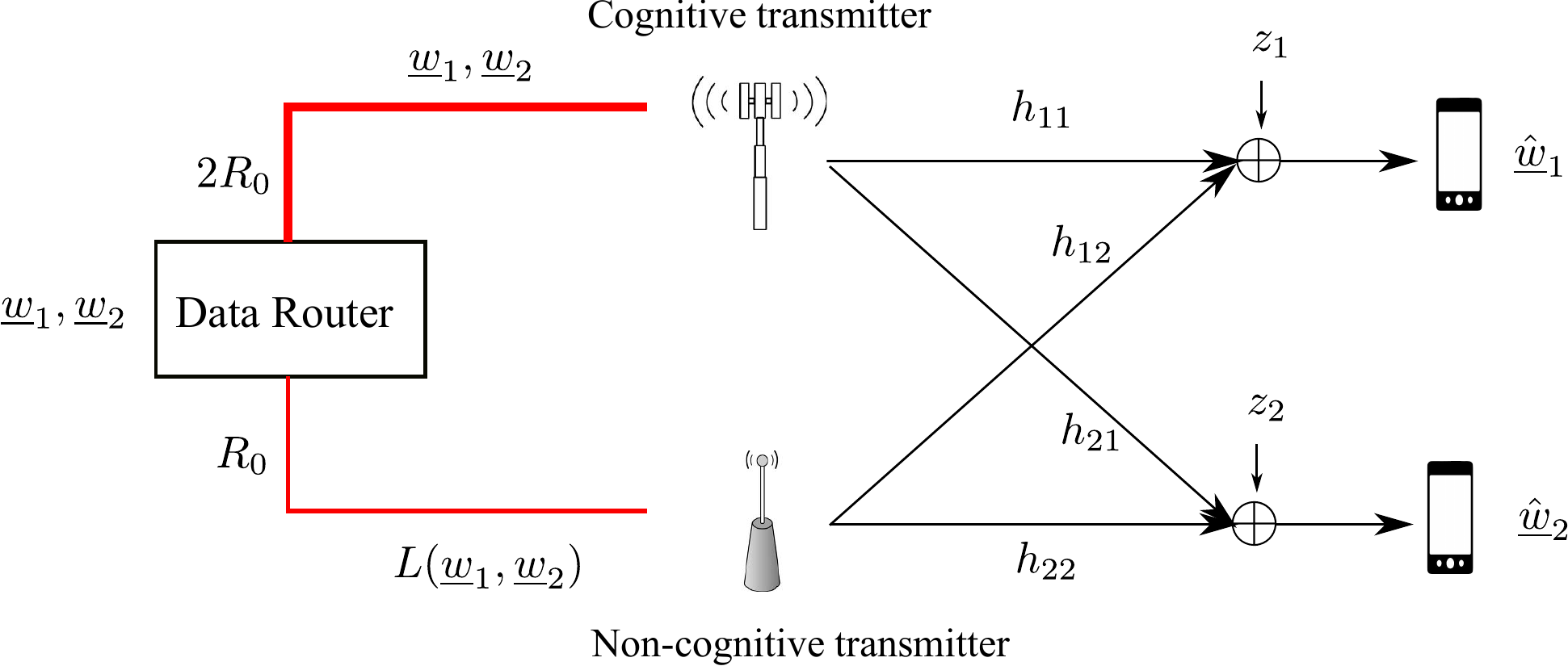}}
\caption{In the classical CIC, the data router sends the one of information messages to the non-cognitive transmitter (i.e., $L(\underline{\wv}_{1},\underline{\wv}_{2}) = \underline{\wv}_{2}$). In the Network-Coded CIC, the data router forwards ``mixed data" to the non-cognitive transmitter
(i.e., $L(\underline{\wv}_{1},\underline{\wv}_{2})=\underline{\wv}_{1} \oplus \underline{\wv}_{2}$).}
\label{model3}
\end{figure}

\begin{figure}
\centerline{\includegraphics[width=16cm]{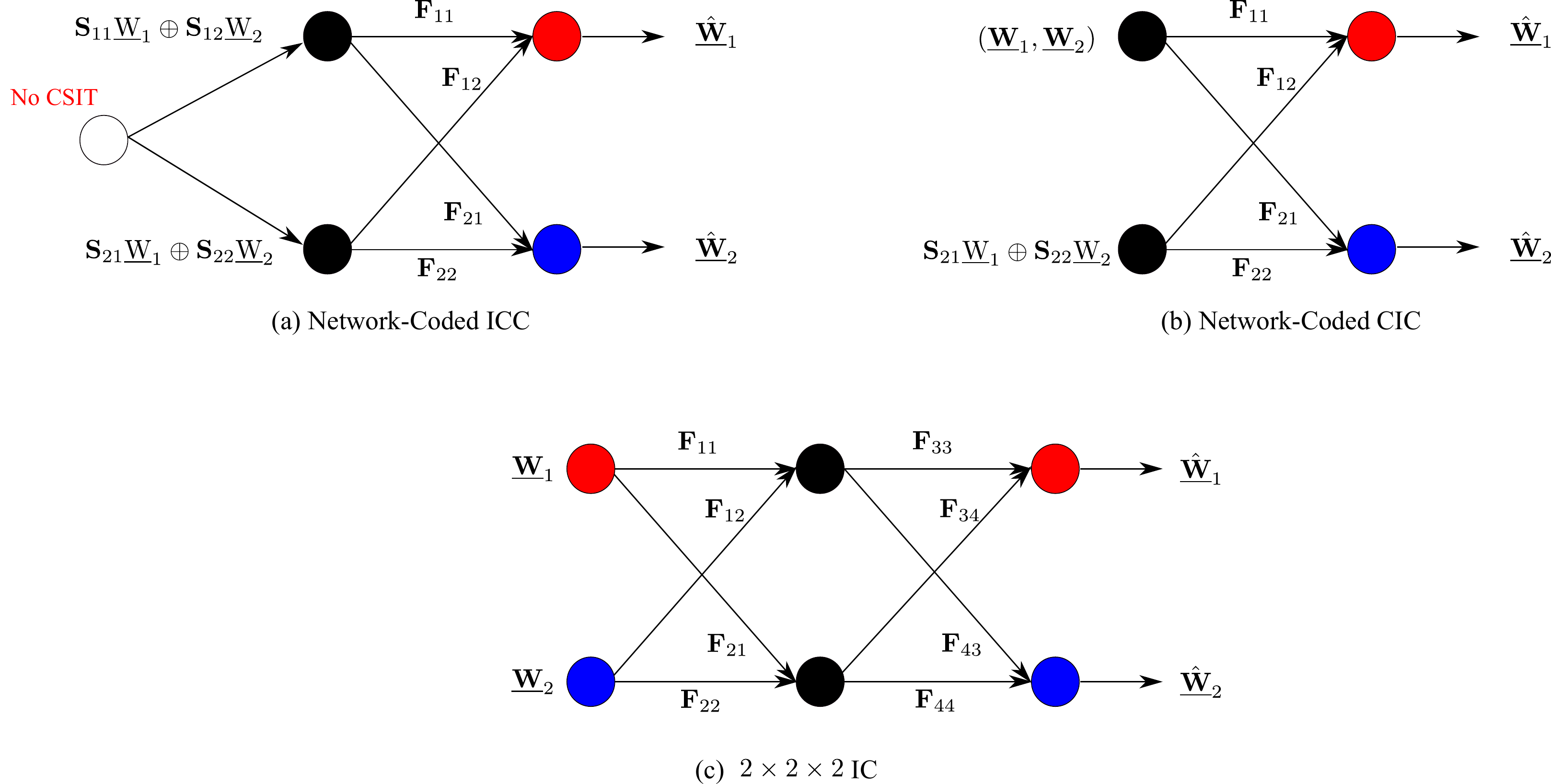}}
\caption{Two-User Gaussian networks with coordination, cognition, and two hops. }
\label{system-model}
\end{figure}

Proceeding along this line, we observe that a basic level of multicell cooperation
named {\em interference coordination} has been investigated and it is currently considered in industry for its practical aspects.
A simple two-user model for interference coordination consists of a data router, two $M$-antenna base stations,
and two $M$-antenna user receivers. The wired backhaul links from the router to the base stations have the same capacity equal to $R_0$.
The data router has no knowledge on channel state information (CSI), due to the separation between
physical layer and network layer.
However, the base stations have full CSI of both the direct and interfering links, obtained from the users through feedback channels.
Such CSI knowledge allows the base stations to coordinate in their sharing strategies such as power allocation and beamforming directions~\cite{Gesbert}.
From an information theoretic viewpoint,  this model is a two-user MIMO IC.
Assuming that network coding is used in the backhaul,  the data router can deliver linear combinations of the messages instead of individual messages,
at the same cost in terms of backhaul capacity constraint. Hence, the model becomes a two-user MIMO IC with mixed data at both transmitters (see Fig.~\ref{system-model} (a)), and shall be referred to as
the Network-Coded Interference Coordination Channel (ICC). In this paper, we address the following question:
 {\em Does mixed data at the transmitters provide a capacity increase ``for free" for the Network-Coded ICC over conventional interference coordination?}

Finally, building on the insight gained in the above Network-Coded cognitive networks,
we study the $2 \times 2 \times 2$ MIMO IC, as shown in Fig.~\ref{system-model} (c), consisting of two transmitters (sources),
two relays, and two receivers (destinations), where nodes have $M$ multiple antennas.
This model is non-cognitive but has some commonality with the previous models in the sense that it consists of two
cascaded two-user MIMO ICs where the relay can have access to mixed messages if proper alignment and coding over the finite-fields is used
in the first hop.  The $2 \times 2 \times 2$ Gaussian IC has received much attention recently, being one of the fundamental building blocks to characterize
two-flow networks \cite{Shomorony}. One natural approach is to consider this model as a cascade of two ICs.
In \cite{Simeone}, the authors apply the Han-Kobayashi scheme \cite{Han} for the first hop to split each message into private and common parts.
Relays can cooperate using the shared information (i.e., common messages) for the second hop, in order to enhance the data rates.
This approach is known to be highly suboptimal at high SNR, since two-user IC can only achieve 1 DoF.
In \cite{Cadambe} it was shown that $\frac{4}{3}$ DoF is achievable by viewing each hop as an X-channel.
This is accomplished using the {\em interference alignment} scheme for each hop.
More recently, the optimal DoF was obtained in \cite{Gou} using a new scheme called {\em aligned interference neutralization},
which appropriately combines interference alignment and interference neutralization.
Also, the $K\times K\times K$ Gaussian IC was recently studied in \cite{Shomorony-K}, where it is shown that the DoFs
cut-set upper bound  (equal to $K$) can be effectively achieved using {\em aligned network diagonalization}.

\subsection{Contributions}

\subsubsection{Network-Coded CIC:  single antenna case}

We characterize the capacity region of a finite-field Network-Coded CIC
using distributed zero-forcing precoding.
We notice that this region is equivalent to that of  a finite-field  vector broadcast channel. This shows that in this case partial cooperation
yields the same performance of full cooperation, as long as the non-cognitive
transmitter knows the mixed message rather than its own individual message only.
Thus, we conclude that mixed data at the non-cognitive transmitter {\em can increase capacity}.
It is worthwhile noticing that the finite-field model is itself meaningful in practical wireless communication systems,
by the observation that the main bottleneck of a digital receiver is the Analog to Digital Conversion (ADC), which is costly, power-hungry, and does not scale
with Moore's law. Rather, the number of bit per second produced by ADC is roughly a constant that depends on the power consumption \cite{Walden,Singh}.
Therefore, it makes sense to consider the ADC as part of the channel. This, together with the algebraic structure induced by lattice coding,
produces a finite-field model as shown by the authors in \cite{Song-IT,Song-ITW}. Motivated by this first successful result, we present a novel scheme
nicknamed {\em Precoded Compute-and-Forward} (PCoF) for Gaussian Network-Coded CIC.
CoF makes use of nested lattice codes, such that each receiver can reliably decode a linear combination with integer
coefficient of the interfering codewords \cite{Nazer}. Thanks to the fact that lattice are modules over the ring of integers,
this linear combination translates directly into a linear combination of the information messages defined over a suitable finite-field.
For brevity, we refer to this fact as ``lattice linearity'' in the following.
Finally, the interference in the finite-field domain is completely eliminated by distributed zero forcing
precoding (over finite-field). This scheme can be thought of as a distributed approach of Reverse CoF (RCoF), proposed by the authors
in \cite{Song-ISIT,Song-IT} for the downlink of distributed antenna systems.

\begin{figure}
\centerline{\includegraphics[width=14cm]{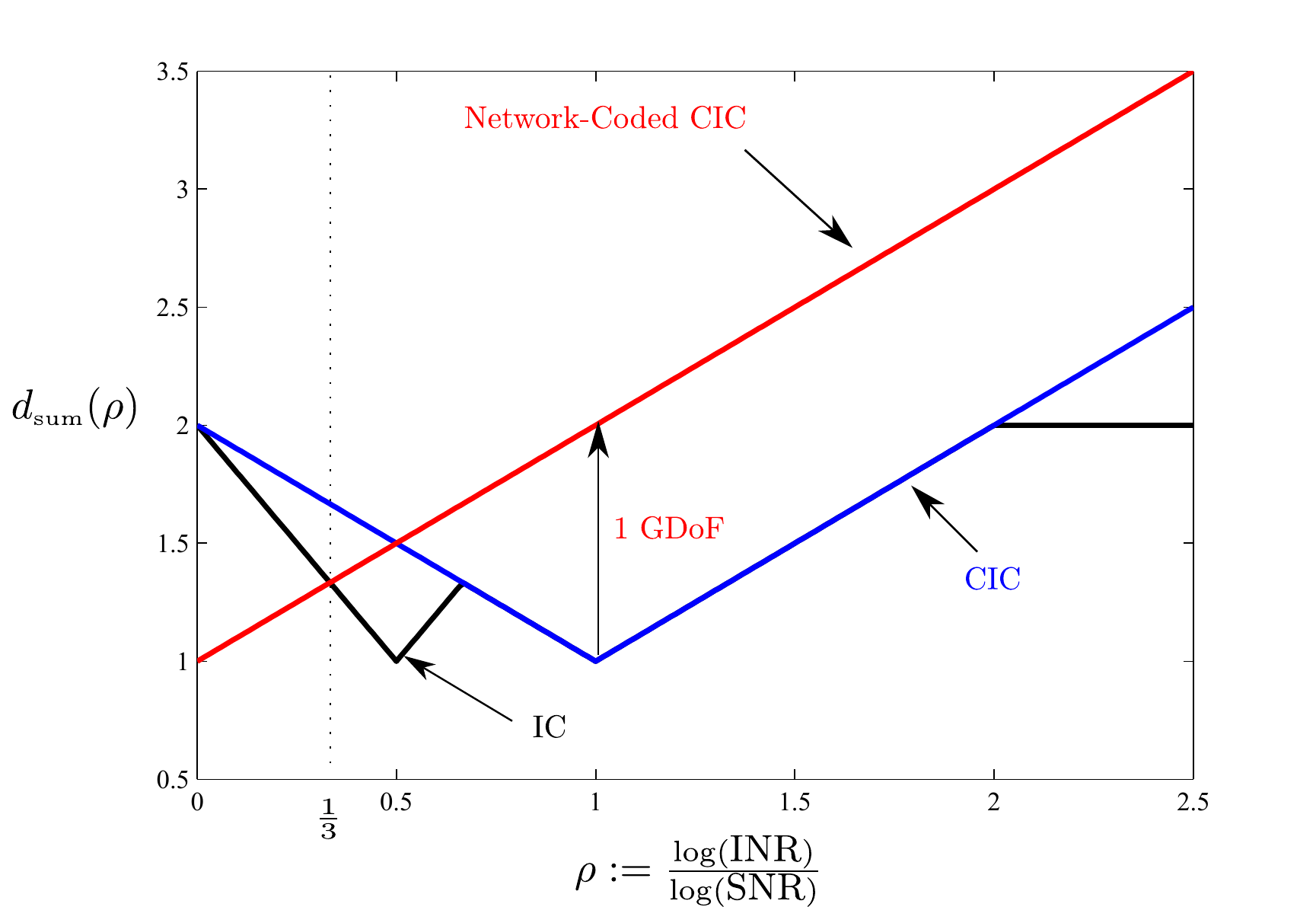}}
\caption{The generalized degrees-of-freedom (GDoF) of the two-user Gaussian Network-Coded CIC. For the interference regimes with $\rho \geq 1/2$, the gap between the Network-Coded CIC and CIC becomes arbitrarily large as SNR and INR goes to infinity.
This shows that mixed-data at the non-cognitive transmitter can provide the unbounded capacity gain at high SNR.}
\label{GDoF}
\end{figure}

Another novel contribution of this work is the characterization of the sum GDoF  (see
\cite{Etkin}  and definition in (\ref{GDoF-sum}))
of the Gaussian Network-Coded CIC. This is obtained by combining an improved achievability result
that makes use of both Dirty-Paper Coding (DPC) and PCoF, with a new outer bound on the sum rate.
As a consequence of the GDoF analysis, we show that mixed data can provide an {\em unbounded}
capacity gain with respect to conventional CICs.
As shown in Fig.~\ref{GDoF}, the sum GDoF  of the Network-Coded CIC is larger than
the sum GDoF of the standard IC when
$\rho = \frac{\log \INR}{\log \SNR}$, the ratio of the interference power over the direct link power (expressed in dB),
is larger than $1/3$, and it is larger than the  sum GDoF of the standard CIC when $\rho > 1/2$. In contrast,
for $\rho < 1/2$ it is better not to mix the data on the side information link
to the non-cognitive transmitter (i.e., use routing in the backhaul link).
It is also interesting to notice that for $\rho=1$ the use of mixed data provides the same
two degrees of freedom of full cooperation (two-users vector broadcast channel), mimicking the result of the finite-field case.

\subsubsection{Network-Coded MIMO IC}

As anticipated before, we have considered three communication models having the two-user MIMO IC as a building block:
Network-Coded ICC (representative of a cellular system downlink with interference coordination),
Network-Coded CIC with MIMO generalization, and $2\times 2\times 2$ IC  \cite{Gou,Shomorony}. In all these models,
we assume that all nodes have $M$ transmit/receive antennas.
Our coding scheme is based on the extension of the PCoF idea to the MIMO case.
This scheme consists of two phases: 1) Using the CoF framework in order to  transform the two-user MIMO IC into a {\em deterministic}
finite-field IC; 2) Using linear precoding on the finite-field domain in order to eliminate interference.
The main performance bottleneck of CoF consists of the non-integer penalty,
which ultimately limits the performance of CoF at high SNR \cite{Niesen1}. To overcome this bottleneck, we employ
{\em Channel Integer Alignment}  (CIA) in order to create an ``aligned" channel matrix for which exact integer
forcing is possible. We derive achievable {\em symmetric} sum rate results for all three channel modes
and prove that PCoF with CIA can achieve sum DoF equal to $2M-1$ in all cases. In particular, for the Network-Coded CIC, we prove that the optimal $2M$ sum DoF is achieved by appropriately combining DPC and PCoF as in the scalar case.  Beyond the DoF results, we further employ the lattice codes algebraic structure in order to
obtain good performance at finite SNRs. We use the {\em integer-forcing receiver} (IFR) approach of \cite{Zhan} and
{\em integer-forcing beamforming} (IFB), proposed by the authors in \cite{Song-ISIT,Song-IT},
in order to minimize the power penalty at the transmitters. We provide numerical results showing that
PCoF with CIA outperforms time-sharing even at reasonably moderate SNR,
with increasing performance gain as SNR increases.  Notice that PCoF with CIA cannot achieve the optimal DoF equal to $2M$. However, we would like to emphasize that unlike a DoF-optimal scheme based on either rational dimension framework or symbol extension framework, the achievable of the proposed scheme is computable at any finite SNR. Also, it can provide a satisfactory performance at ``practical" values of SNR and at manageable complexity but the DoF-optimal scheme cannot not guarantee a good performance at finite SNRs.

\subsection{Organization}

This paper is organized as follows.
In Section~\ref{sec:pre}, we summarize some definitions on lattices and lattice coding
and review CoF.
In Section~\ref{sec:CIC-FF}, we characterize the capacity region of finite-field Network-Coded CIC and present PCoF,
as a natural extension of finite-field scheme, for the Gaussian Network-Coded CIC.
Further, we derive an achievable rate region of the Gaussian Network-Coded CIC, by appropriately combining
PCoF and DPC, and characterize the sum GDoF.
In Section~\ref{sec:twoMIMO}, we characterize an achievable {\em symmetric} sum rate for the two-user MIMO ICs under investigation,
and derive the DoF of these channels. The sum rates are improved in Section~\ref{sec:SIC} using successive cancellation
with respect to CoF. In Section~\ref{sec:finite}, we optimize the (symmetric) sum rate
and provide some numerical results, showing good performance
at intermediate and ``practical'' values of SNR. Some concluding remarks are provided in Section~\ref{sec:con}.

\section{Preliminaries}\label{sec:pre}

In this section we provide some basic definitions and background results that will be extensively used in the sequel.

{\it Notation:} We use boldface capital letters $\Xm$ for matrices and boldface small letters $\xv$ for column vectors. In addition, we use ``underline"  to denote matrices whose horizontal dimension (column index) denotes ``time" and vertical dimension (row index) runs across the antennas. That is, the $K\times n$ matrix
\begin{equation}
\underline{\Xm} = \left[
  \begin{array}{c}
    \underline{\xv}_{1} \\
    \vdots \\
    \underline{\xv}_{K}
  \end{array}
\right]
\end{equation} contains, arranged by rows, the row vectors $\underline{\xv}_{k} \in \CC^{1 \times n}$ for $k=1,\ldots, K$. $\Xm^{\dag}$ denote the Hermitian transpose of the matrix $\Xm$. Also, $\trace{(\Xm)}$ and $\det(\Xm)$ denote the trace and the determinant of the square matrix $\Xm$. Also, $\Id_{M}$ denotes the $M \times M$ identity matrix.

\subsection{Nested Lattice Codes}\label{subsec:NLC}

Let $\ZZ[j]$ be the ring of Gaussian integers and $p$ be a prime. Let $\oplus$ denote the addition over $\FF_{q}$ with $q=p^2$, and let $g: \FF_{q} \rightarrow \CC$ be the natural mapping of $\FF_{q}$ onto $\{a+jb: a,b \in \ZZ_{p}\} \subset \CC$.
We recall the nested lattice code construction given in \cite{Nazer}.
Let $\Lambda = \{ \underline{\lambdav} = \underline{\zv} \Tm : \underline{\zv} \in \ZZ^n[j]\}$ be a lattice in $\CC^n$,
with full-rank generator matrix $\Tm \in \CC^{n \times n}$.
Let $\Cc  = \{ \underline{\cv} = \underline{\wv} \Gm : \wv \in \FF_{q}^r \}$ denote a linear code over $\FF_{q}$ with block length $n$ and dimension $r$, with generator matrix $\Gm$. The lattice $\Lambda_1$ is defined through ``construction A'' (see \cite{Erez2004} and references therein) as
\begin{equation} \label{construction-A}
\Lambda_1 = p^{-1} g(\Cc) \Tm + \Lambda,
\end{equation}
where $g(\Cc)$ is the image of $\Cc$ under the mapping $g$ (applied component-wise).
It follows that $\Lambda \subseteq \Lambda_1 \subseteq p^{-1} \Lambda$ is a chain of nested lattices, such that
$|\Lambda_1/\Lambda| = p^{2r}$ and $|p^{-1} \Lambda/\Lambda_1| = p^{2(n - r)}$.

For a lattice $\Lambda$ and $\underline{\rv} \in \CC^n$,
we define the lattice quantizer $Q_{\Lambda}(\underline{\rv}) = \argmin_{\tiny{\underline{\lambdav}} \in \Lambda}\|\underline{\rv} - \underline{\lambdav} \|^2$, the Voronoi region $\Vc_\Lambda = \{\underline{\rv} \in \CC^{n}: Q_{\Lambda}(\underline{\rv}) = \underline{\zerov}\}$,
the modulo reduction $[\underline{\rv}] \mod \Lambda = \underline{\rv} - Q_{\Lambda}(\underline{\rv})$,
and the per-component second moment
$\sigma_\Lambda^2 = \frac{1}{n \mbox{\small Vol}(\Vc)} \int_{\Vc} \| \underline{\rv} \|^2 d\underline{\rv}=\SNR$.

Given $\Lambda$ and $\Lambda_1$ above, we define the lattice code $\Lc = \Lambda_{1} \cap \Vc_\Lambda$ with rate
$R = \frac{1}{n} \log |\Lc| = \frac{r}{n}\log{q}$.
The set $p^{-1} g(\Cc)\Tm$ is a {\em system of representatives} of the cosets of $\Lambda$ in $\Lambda_1$.
This induces a natural labeling $f : \FF^r_{q} \rightarrow \Lc$ of the codewords of $\Lc$ by the information messages $\underline{\wv} \in \FF_q^r$
defined by $f(\underline{\wv}) = p^{-1} g(\underline{\wv} \Gm)\Tm \mod \Lambda$.

\subsection{Compute-and-Forward and Integer-Forcing}\label{subsec:CoF}

We recall here the CoF scheme of \cite{Nazer} applied to a particular case of Gaussian MIMO channel
with joint processing of the receiver antennas and  independent lattice coding at each transmit antenna.
Our reference model is given by
\begin{equation}  \label{2user-GMAC}
\underline{\Ym} = \Hm\Cm \underline{\Xm} + \underline{\Zm}
\end{equation}
where $\Hm \in \CC^{M \times M}$ is a full-rank channel matrix, $\Cm \in \ZZ[j]^{M \times S}$,
$\underline{\Xm} \in \CC^{S \times n}$, and where $S$ denotes the number of independent and independently encoded information streams (messages) sent by a virtual
``super-user'' collecting all channel inputs. Here, $\underline{\Zm}$ contains i.i.d. Gaussian noise samples
$\sim \Cc\Nc(0,1)$.
For $k = 1, \ldots, S$, each $k$-th independent message $\underline{\wv}_{k} \in \FF_{q}^{r}$ is encoded input the codeword
$\underline{\tv}_{k} = f(\underline{\wv}_{k})$ of the same lattice code $\Lc$ of rate $R$ and mapped to the channel input sequence
\begin{equation}
\underline{\xv}_{k}= [\underline{\tv}_{k} + \underline{\dv}_{k}] \mod \Lambda,
\end{equation}
where the {\em dithering sequences} $\{\underline{\dv}_{k}\}$ are mutually independent,
uniformly distributed over $\Vc_{\Lambda}$, and known to the receiver.\footnote{The dithering sequences in this paper have
these properties, and this fact will be understood in the sequel even though not explicitly stated.}
The encoded sequences $\{\underline{\xv}_{k}\}$ are arranged by rows into the transmit signal matrix $\underline{\Xm}$.
We also define the matrix $\underline{\Tm}$ of dimensions $S \times n$ containing
$\{\underline{\tv}_{k}\}$ arranged by rows,  and the dithering matrix $\underline{\Dm}$ with rows $\{\underline{\dv}_{k}\}$.

Channel matrices in the form $\Hm\Cm$ as in (\ref{2user-GMAC}) will appear several times in this paper as a consequence of
{\em channel  integer alignment}, explicitly designed such that $[\Cm \underline{\Tm}] \mod \Lambda$ has lattice codewords arranged by rows.\footnote{The modulo
$\Lambda$ reduction applied to matrices is intended row by row.}

The decoder's goal is to recover $L \leq M$ integer linear combinations of the $S$ lattice codewords, given by the rows $\underline{\sv}_\ell$
of the matrix $\underline{\Sm} = [\Bm^{\herm} \Cm \underline{\Tm}] \mod \Lambda$,
for some integer matrix $\Bm \in \ZZ[j]^{M \times L}$. Letting $\bv_\ell$ denote the $\ell$-th column of $\Bm$, the receiver computes
\begin{eqnarray}
\hat{\underline{\yv}}_\ell &=& \left[\alphav_\ell^{\herm} \underline{\Ym} - \bv_\ell^{\herm}\Cm\underline{\Dm}\right] \mod \Lambda\nonumber\\
&=&[\bv_\ell^{\herm}\Cm\underline{\Tm} + \alpha_\ell^{\herm}(\Hm\Cm\underline{\Xm} + \underline{\Zm})-\bv_\ell^{\herm}\Cm(\underline{\Tm}+\underline{\Dm})] \mod \Lambda\nonumber\\
&=& \left[\underline{\sv}_\ell + \underline{\zv}_{\mbox{\tiny{eff}}}(\Hm\Cm,\bv_\ell,\alphav_\ell) \right] \mod \Lambda\label{eq:cof}
\end{eqnarray}
where $\alphav_\ell \in \CC^{M\times 1}$ and  $\underline{\zv}_{\mbox{\tiny{eff}}}(\Hm\Cm,\bv_\ell,\alphav_\ell)$ is the $\ell$-th effective noise sequence,
distributed as: \footnote{This follows from the fact that $\underline{\Xm}$ has the same distribution of $\underline{\Dm}$.}
\begin{equation}
\underbrace{\left (\alphav_\ell^{\herm}\Hm - \bv_\ell^{\herm} \right )\Cm \underline{\tilde{\Dm}}}_{\mbox{non-integer penalty}}  \;\;\;\ +
\underbrace{\alphav_\ell^{\herm} \underline{\Zm}}_{\mbox{Gaussian noise}}
\end{equation} where each row of $\tilde{\Dm}$ is drawn independently according to a uniform distribution over $\Vc_{\Lambda}$.
Choosing $\alphav_\ell^{\herm}=\bv_\ell^{\herm}\Hm^{-1}$, the variance of the effective noise is given by
\begin{equation} \label{exact-IFR}
\sigma^{2}_{\mbox{\tiny{eff}},\ell} = \|(\Hm^{-1})^{\herm}\bv_\ell \|^2.
\end{equation}
This choice is referred to in \cite{Zhan} as the {\em exact} Integer Forcing Receiver (IFR).
In this way, the non-integer penalty of CoF is completely eliminated.
More in general, the decoding performance can be improved especially at low SNR by  minimizing the effective noise variance with respect to $\alphav_\ell$ for given
$\bv_\ell$ \cite{Zhan}. This yields
\begin{eqnarray}
\sigma^{2}_{\mbox{\tiny{eff}},\ell}
&=& \bv_\ell^{\herm}\Cm(\SNR^{-1}\Id+\Cm^{\herm}\Hm^{\herm}\Hm\Cm)^{-1}\Cm^{\herm}\bv_\ell.
\end{eqnarray}
Since $\bv_\ell$ and $\Cm$ are integer-valued, $\underline{\sv}_\ell = [\bv_\ell^{\herm}\Cm\underline{\Tm}] \mod \Lambda$ is a codeword of $\Lc$.
From \cite{Nazer}, we know that by applying lattice decoding to $\hat{\underline{\yv}}_\ell$ given in (\ref{eq:cof}) there exist sequences of lattice codes
$\Lc$ of rate $R$ and increasing block length $n$ such that $\underline{\sv}_\ell$ can be decoded successfully with arbitrarily high probability
as $n \rightarrow \infty$, provided that\footnote{We define $\log^{+}(x)\triangleq \max\{\log(x),0\}$.}
\begin{equation}
R <  \log^{+}\left (\frac{\SNR}{\sigma^2_{\mbox{\tiny{eff}},\ell}} \right ),  \label{eq:cofrate}
\end{equation}
where the expression in the right-hand side of (\ref{eq:cofrate}) is the {\em computation rate} for the modulo-$\Lambda$ additive noise
channel (\ref{eq:cof}) with given SNR and effective noise variance.
All the $L$ linear combinations can be reliably decoded if
\begin{equation}
R \leq \min_{\ell} \left \{ \log^+\left (\frac{\SNR}{\sigma^2_{\mbox{\tiny{eff}},\ell}} \right ) \right \}. \label{cofrate}
\end{equation}
The requirement that all inputs are encoded with the same lattice code with rate constrained by the worst computation rate over the desired
linear combinations (columns of $\Bm$) can be relaxed by allowing the inputs to be encoded at different rates, using a family of nested lattice codes,
and using {\em successive cancellation} with respect to CoF, according to the scheme proposed and analyzed in \cite{Ordentlich}.
The application of this idea to the networks treated in this paper is examined in Section \ref{sec:SIC}.

Using the linearity of lattice encoding,\footnote{{\em Lattice encoding linearity} refers to the
isomorphism between $\FF_q^r$ and $\Lc$ induced by the natural labeling $f$ defined before.}
the corresponding $L$ linear combinations over $\FF_{q}$ for the messages are given by
\begin{eqnarray}
\underline{\Um} &=& g^{-1}([\Bm^{\herm}] \mod p\ZZ[j])g^{-1}([\Cm] \mod p\ZZ[j])\underline{\Wm}\nonumber\\
&=& [\Bm^{\herm}]_{q}[\Cm]_{q}\underline{\Wm}, \label{cofeq}
\end{eqnarray}
where we use the notation $[\Bm^{\herm}]_{q}  \triangleq g^{-1}([\Bm^{\herm}] \mod p\ZZ[j])$ and $[\Cm]_{q}\triangleq g^{-1}([\Cm] \mod p\ZZ[j])$. Throughout the paper, we use the notation $[\Mm]_{q} \triangleq g^{-1}([\Mm] \mod p\ZZ[j])$ for any integer matrix $\Mm$.

\section{Network-Coded Cognitive Interference Channel}\label{sec:CIC-FF}

A two-user Gaussian Network-Coded CIC consists of a Gaussian interference channel where
transmitter 1 (the cognitive transmitter) knows both user 1 and user 2 information messages (or, equivalently, two
independent linear combinations thereof) and transmitter 2 (the non-cognitive transmitter) only knows only one linear combination of the messages.
Without loss of generality, we assume that transmitter 1 knows
($\underline{\wv}_{1},\underline{\wv}_{2})$, and transmitter 2 has
$\underline{\wv}_{1} \oplus \underline{\wv}_{2}$, where $\underline{\wv}_{k} \in \FF_{q}^{r}$ denotes
the information message desired at receiver $k$, at rate $R_k$ bit/symbol, for $k=1,2$.
We assume that if $R_{1}  \neq R_{2}$ then the lowest rate message is zero-padded
such that both messages have a common length,  given by $r = \max\{nR_{1},nR_{2}\}$, where $n$ denotes the
coding block length.  A block of $n$ channel uses of the discrete-time complex baseband two-user IC is described by
\begin{eqnarray}
\underline{\yv}_{1} &=& h_{11} \underline{\xv}_{1} + h_{12} \underline{\xv}_{2} + \underline{\zv}_{1}\label{eq:channel1}\\
\underline{\yv}_{2} &=& h_{21} \underline{\xv}_{1} + h_{22} \underline{\xv}_{2} + \underline{\zv}_{2}, \label{eq:channel2}
\end{eqnarray}
where $\underline{\zv}_{k} \in \CC^{n \times 1}$ contains i.i.d. Gaussian noise samples $\sim \Cc\Nc(0,1)$
and $h_{ij} \in \CC$ denotes the channel coefficients, assumed to be constant over the whole block of length $n$ and known to all nodes.
Also, we have a common per-user power constraint, given by $\frac{1}{n} \EE[\|\underline{\xv}_{k}\|^2]\leq \SNR$, for $k=1,2$.
Each receiver $k$ observes the channel output $\underline{\yv}_{k}$ and produces an estimate $\hat{\underline{\wv}}_{k}$
of the desired message $\underline{\wv}_{k}$.
A rate pair $(R_{1},R_{2})$ is achievable if there exists a family of codes satisfying the power constraint,
such that the average decoding error probability satisfies
$\lim_{n \rightarrow \infty}\PP(\hat{\underline{\wv}}_{k} \neq \underline{\wv}_{k}) = 0$, for both $k=1,2$.

\subsection{Capacity Region for finite-field Network-Coded CIC}

\begin{figure}[t]
\centerline{\includegraphics[width=14cm]{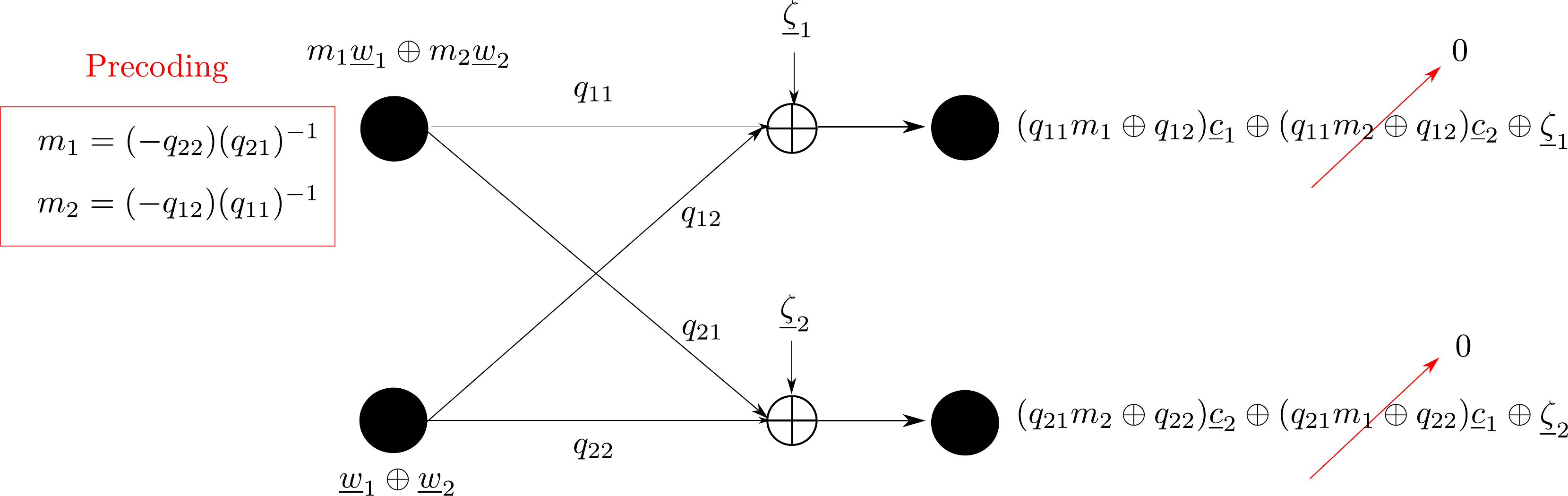}}
\caption{Distributed zero-forcing precoding for finite-field Network-Coded CIC.
Differently from RLNC, the cognitive transmitter carefully chooses the coefficients of linear
combination according to the channel coefficients $q_{ij}$'s.}
\label{FF-Precoding}
\end{figure}

In order to build an intuition for Gaussian channel, we consider the corresponding finite-field model and show that
{\em distributed zero-forcing precoding} achieves the capacity of finite-field Network-Coded CIC.
A block of $n$ channel uses of the discrete-time finite-field IC is described by
\begin{equation}
\left[
  \begin{array}{c}
    \underline{\yv}_{1} \\
    \underline{\yv}_{2} \\
  \end{array}
\right] =\Qm
 \left[
  \begin{array}{c}
    \underline{\xv}_{1} \\
    \underline{\xv}_{2} \\
  \end{array}
\right] \oplus \left[
  \begin{array}{c}
    \underline{\zetav}_{1} \\
    \underline{\zetav}_{2} \\
  \end{array}
\right]
\end{equation} where $\zetav_{k}=(\zeta_{k,1},\ldots,\zeta_{k,n}) \in \FF_{q}^{n}$ contains i.i.d. additive noise samples $\sim \prod_{\ell=1}^{n}P_{\zeta}(\zeta_{k,\ell})$, and $q_{ij} \in \FF_{q}$
is the $(i,j)$-th element of $\Qm$, denoting the channel coefficients from transmitter $j$ to receiver $i$,
assumed to be constant over the whole block of length $n$ and known to all nodes. Also, it is assumed that $H(\zeta_{k}) < \log{q}$ for $k=1,2$, in order to ensure a non-zero channel capacity for each receiver.

\begin{theorem}\label{thm:FF}
If $\det(\Qm) \neq 0$ and $q_{11}, q_{21} \neq 0$, the capacity region of the finite-field Network-Coded CIC is the set of all
rate pairs $(R_{1},R_{2})$ such that
\begin{eqnarray}
R_{k} &\leq& \log{q} - H(\zeta_{k})\mbox{ for }k=1,2.\label{eq:th1}
\end{eqnarray}
\end{theorem}

\begin{IEEEproof} We first derive a simple upper bound by assuming full transmitter cooperation. In this case, this model reduces to the finite-field vector broadcast channel. A trivial upper-bound on the broadcast capacity region is given by \cite{Kim}:
\begin{equation}
R_{k} \leq \max_{P_{X_{1},X_{2}}}I(X_{1},X_{2};Y_{k}) \mbox{ for } k=1,2\label{eq:upper}
\end{equation}
Due to the additive noise nature of the channel, we have $I(X_{1},X_{2};Y_{k}) = H(Y_{k}) - H(\zeta_{k})$.
Furthermore, $H(Y_{k}) \leq \log{q}$ and this upper bound is achieved by letting $(X_{1},X_{2}) \sim \mbox{Uniform}$
over $\FF_{q}^{2}$. This bound coincides with (\ref{eq:th1}).

Next, we derive an achievable rate using distributed zero-forcing precoding technique. Without loss of generality, it is assumed that $H(\zeta_{1}) \leq H(\zeta_{2})$. We use two nested linear codes $\Cc_{2} \subseteq \Cc_{1}$ where $\Cc_{k}$ has rate $R_{k} =  \frac{r_{k}}{n}\log{q}$. Let $\underline{\wv}_{1}$ and $\underline{\wv}_{2}$ be the zero-padded information messages to common length $r_{1}$.  Also, let $\Gm$ denote a full-rank generator matrix of the linear code $\Cc_{1}$. The detailed procedures of distributed zero forcing technique is as follows (see Fig.~\ref{FF-Precoding}).
\begin{itemize}
\item Transmitter 1 produces the codewords $\underline{\cv}_{1} = \underline{\wv}_{1}\Gm$ and $\underline{\cv}_{2} = \underline{\wv}_{2}\Gm$, and transmits the precoded codeword $\underline{\xv}_{1} = m_{1}\underline{\cv}_{1} \oplus m_{2} \underline{\cv}_{2}$, for some coefficients $m_{1}, m_{2} \in \FF_{q}$;
\item Transmitter 2 produces the codeword $\underline{\cv}_{1} \oplus \underline{\cv}_{2} = (\underline{\wv}_{1}\oplus \underline{\wv}_{2})\Gm$ and transmits the precoded codeword $\underline{\xv}_{2} = m_{3} (\underline{\cv}_{1} \oplus \underline{\cv}_{2})$ with coefficient $m_{3} \in \FF_{q}$.
\item Receiver 1 observes:
\begin{eqnarray}
\underline{\yv}_{1} &=& q_{11}\underline{\xv}_{1} \oplus q_{12}\underline{\xv}_{2} \oplus \underline{\zetav}_{1}\\
&=& \lambda_{11}\underline{\cv}_{1} \oplus \lambda_{12}\underline{\cv}_{2} \oplus \underline{\zetav}_{1}
\end{eqnarray} where $\lambda_{11} = (q_{11}m_{1} \oplus q_{12}m_{3})$ and $\lambda_{12} = (q_{11}m_{2} \oplus q_{12}m_{3})$.
\item Receiver 2 observes:
\begin{eqnarray}
\underline{\yv}_{2} &=& q_{21}\underline{\xv}_{1} \oplus q_{22}\underline{\xv}_{2} \oplus \underline{\zetav}_{2}\\
&=& \lambda_{22}\underline{\cv}_{2} \oplus \lambda_{21}\underline{\cv}_{1} \oplus \underline{\zetav}_{2}
\end{eqnarray}where $\lambda_{22} = (q_{21}m_{2} \oplus q_{22}m_{3})$ and $\lambda_{21} = (q_{21}m_{1} \oplus q_{22}m_{3})$.
\end{itemize}
The goal is to find a precoding vector $\mv=(m_{1},m_{2},m_{3})^{\transp}$ to cancel interference at both receivers (i.e., such that
$\lambda_{12}=\lambda_{21}=0$), while preserving the desired codewords (i.e., such that $\lambda_{11},\lambda_{22} \neq 0$).
Equivalently, we want to find a non-zero vector $\mv$ to satisfy the following conditions:
\begin{itemize}
\item Condition 1 ({\em canceling the interferences})
\begin{equation}\label{eq:cond1}
\Cm\mv = \zerov
\end{equation} where
\begin{equation}
\Cm = \left[
  \begin{array}{ccc}
    0 & q_{11} & q_{12} \\
    q_{21} & 0 & q_{22} \\
  \end{array}
\right].\label{eq:C}
\end{equation}
\item Condition 2 ({\em preserving the desired signals})
\begin{equation}\label{eq:cond2}
\det(\Qm\Mm) = \lambda_{11} \lambda_{22} \neq 0
\end{equation} where  it is assumed that the condition (\ref{eq:cond1}) is satisfied (e.g., $\lambda_{12}=\lambda_{21}=0$) and
\begin{equation}
\Mm = \left[
     \begin{array}{cc}
       m_{1} & m_{2} \\
       m_{3} & m_{3}\\
     \end{array}
   \right].\label{eq:B}
\end{equation}
\end{itemize}
Since $\mbox{Rank}(\Cm) \leq 2$, there exist non-zero vectors $\mv^{*} \in \mbox{Null}(\Cm)$ that satisfies Condition 1.
Since $\Qm$ has full rank, Condition 2 is equivalent to requiring that $\Mm$ has rank 2, i.e.,
that $m_3(m_1 - m_2) \neq 0$. In short, we have to find the conditions for which a vector $\mv$ in the null-space of $\Cm$
satisfies $m_3 \neq 0$ and $m_1 \neq m_2$. Assuming $q_{11} \neq 0$ and $q_{21} \neq 0$, we have that $\Cm \mv = \zerov$ yields
\[ m_2 = - \frac{q_{12}}{q_{11}} m_3, \;\;\;\; m_1 = - \frac{q_{22}}{q_{21}} m_3. \]
Using this in the expression of $\Mm$, we find that $\det(\Mm) \neq 0$ if we choose $m_3 \neq 0$ and if
\[  - \frac{q_{12}}{q_{11}} + \frac{q_{22}}{q_{21}} = \frac{- q_{12}q_{21} + q_{11} q_{22}}{q_{11}q_{21}} = \frac{\det(\Qm)}{q_{11}q_{21}} \neq 0 \]
By assumption, the above condition is always true, therefore we conclude that a vector $\mv^*$ satisfying Conditions 1 and 2 can
always be found. In this case, the precoded channel decouples into two parallel additive noise channels
\begin{eqnarray}
\underline{\yv}_{1} &=& \lambda_{11}\underline{\cv}_{1} \oplus \underline{\zetav}_{1} \\
\underline{\yv}_{2} &=& \lambda_{22} \underline{\cv}_{2} \oplus \underline{\zetav}_{2},
\end{eqnarray}
for which rates $R_{k} \leq \log{q} - H(\zeta_{k})$ are clearly achievable by linear coding \cite{Dobrushin}.
\end{IEEEproof}

\begin{remark} The capacity region of finite-field Network-Coded CIC under the assumptions of Theorem \ref{thm:FF}
is equivalent to the capacity region of the corresponding finite-field vector broadcast channel.
In other words, partial network-coded cooperation and full cooperation yield the same performance. \hfill $\lozenge$
\end{remark}

\begin{remark} It is interesting to notice that if $q_{11} = 0$ and $\det(\Qm) \neq 0$, then $q_{21}, q_{12} \neq 0$. This
implies that $\mv$ in the null space of $\Cm$ takes on the form $(0, m_2, 0)^\transp$ for some $m_2 \neq 0$.
If  $q_{21} = 0$ and $\det(\Qm) \neq 0$, then $q_{11}, q_{22} \neq 0$. This
implies that $\mv$ in the null space of $\Cm$ takes on the form $(m_1, 0, 0)^\transp$ for some $m_1 \neq 0$.
In both cases, $\det(\Mm) = 0$ and interference cannot be removed without eliminating the useful signal at one of the
two receivers. \hfill $\lozenge$
\end{remark}

The observation in the above remark is strengthened by the following infeasibility result:

\begin{lemma}
 If the conditions of Theorem \ref{thm:FF} do not hold, the sum capacity is strictly less than the sum of the individual channel capacities of
the two additive noise channels from each transmitter to its intended receiver without interference (given by $\log q - H(\zeta_k)$, $k = 1,2$).
\end{lemma}

\begin{IEEEproof}
We will show that if the conditions of Theorem \ref{thm:FF} are not satisfied, then it is not possible to achieve the
sum rate of the two individual point to point (perfectly decoupled) channels, i.e., the sum rate is strictly lower than
$2\log{q} - (H(\zeta_{1})+H(\zeta_{2}))$.
We employ the upper bounds derived in Appendix \ref{proof:GDoF} such as
\begin{eqnarray}
\min\{R_{1},R_{2}\}&\leq&\min\{I(X_{1};Y_{1}|X_{2}), I(X_{1};Y_{2}|X_{2})\}\label{eq1}\\
\max\{R_{1},R_{2}\}&\leq& \max\{I(X_{1},X_{2};Y_{1}),I(X_{1},X_{2};Y_{2})\}\\
&=&\log{q}- \min\{H(\zeta_{1}),H(\zeta_{2})\}\label{eq2}.
\end{eqnarray} Notice that the sum rate is equal to $\min\{R_{1},R_{2}\} + \max\{R_{1},R_{2}\}$. When $q_{11} = 0$, the receiver 1 observes the $Y_{1} = q_{12}X_{2} \oplus \zeta_{1}$. Then, we have that $\min\{R_{1},R_{2}\}=0$ since $I(X_{1};Y_{1}|X_{2}) = H(Y_{1}|X_{2}) - H(Y_{1}|X_{1},X_{2}) = 0$. Using (\ref{eq1}) and (\ref{eq2}), we have that $R_{1} + R_{2} \leq \log{q} - \min\{H(\zeta_{1}),H(\zeta_{2})\}$. Similarly when $q_{21}= 0$, i.e, $Y_{2}=q_{22}X_{2} \oplus \zeta_{2}$, the  $\min\{R_{1},R_{2}\}=0$ is also zero because of $I(X_{1};Y_{2}|X_{2}) = 0$. Thus, we have that $R_{1} + R_{2} \leq \log{q} - \min\{H(\zeta_{1}),H(\zeta_{2})\}$. In both cases, the sum rates are strictly less than
$2\log{q} - (H(\zeta_{1})+H(\zeta_{2}))$.
\end{IEEEproof}

\subsection{Scaled Precoded CoF}\label{subsec:PCoF-scalar}

Motivated by the above result, we present a novel scheme named {\em Precoded Compute-and-Forward} (PCoF) for the
Gaussian Network-Coded CIC. Using CoF decoding, each receiver can reliably decode an integer linear combination of the
lattice codewords sent by transmitters.
Then, the ``interference" in the finite-field
domain can be completely eliminated by distributed zero-forcing precoding,
provided that the conditions of Theorem \ref{thm:FF} are satisfied. Using this scheme, we have:
\begin{theorem}\label{th:LC}
Scaled PCoF applied to Gaussian Network-Coded CIC with $\Hm=[h_{ij}] \in \CC^{2 \times 2}$ achieves the rate pairs $(R_{1},R_{2})$ such that
\[ R_k \leq \log^+ \left (\frac{\SNR}{\bv_{k}^{\herm}(\SNR^{-1}\Id+\tilde{\hv}_{k}\tilde{\hv}_{k}^{\herm})^{-1}\bv_{k}} \right ), \]
for any full rank integer matrix $\Bm=[\bv_{1},\bv_{2}]$ with $b_{11},b_{21} \neq 0  \mod p\ZZ[j]$
and $\beta_{k} \in \CC$ with $|\beta_{k}|\leq 1$, where $\tilde{\hv}_{k} = [\beta_{1}h_{k 1}, \beta_{2} h_{k 2}]$. \hfill \IEEEQED
\end{theorem}

In order to achieve different coding rates while preserving the lattice $\ZZ[j]$-module structure, we use a {\em family} of nested lattices $\Lambda \subseteq \Lambda_{2} \subseteq \Lambda_{1}$, where $\Lambda_k = p^{-1} g(\Cc_k) \Tm + \Lambda$ with $\Lambda = \ZZ^n[j] \Tm$
and where $\Cc_k$ denotes the linear code over $\FF_{q}$ generated by the first $r_k$ rows of a generator matrix $\Gm$,
with $r_2 \leq r_1$.  The corresponding nested lattice codes are given by $\Lc_{k} = \Lambda_{k} \cap \Vc_\Lambda$, and have rate
$R_k = \frac{r_k}{n} \log q$.
We let $\Bm =[\bv_{1}, \bv_{2}] \in \ZZ[j]^{2 \times 2}$, where $\bv_{k}$ denotes the integer coefficients vector used at receiver $k$ for the modulo-$\Lambda$ receiver mapping (see (\ref{eq:cof})), and we let $\Qm =[\Bm^{\herm}]_{q} \in \FF_{q}^{2\times 2}$.
For the time being, it is assumed that
$\det(\Qm), q_{11},q_{21} \neq 0$ over $\FF_q$.  PCoF proceeds as follows:
\begin{itemize}
\item Transmitters 1 and 2 produce the precoded messages:
\begin{eqnarray}
\underline{\uv}_{1} &=& m_{1}\underline{\wv}_{1} \oplus m_{2}\underline{\wv}_{2}\\
\underline{\uv}_{2} &=& m_{3}(\underline{\wv}_{1}\oplus \underline{\wv}_{2}),
\end{eqnarray}
respectively, where $\mv=(m_{1},m_{2},m_{3})$ is a non-zero vector $\mv \in \mbox{Null}(\Cm)$ where $\Cm$ is related to $\Qm$ as defined in (\ref{eq:C}).
\item Each transmitter $k$ produces the lattice codeword $\underline{\vv}_{k} = f(\underline{\uv}_{k}) \in \Lc_{1}$ (the densest lattice code) and transmits
the channel inputs $\underline{\xv}_{k} = [\underline{\vv}_{k} + \underline{\dv}_{k}] \mod \Lambda$, where $\dv_{k}$ are
dithering sequences.
\end{itemize}
By lattice linearity we have:
\begin{eqnarray}
\underline{\vv}_{1} &=& [g(m_{1})\underline{\tv}_{1} + g(m_{2})\underline{\tv}_{2}] \mod \Lambda\\
\underline{\vv}_{2} &=& [g(m_{3})\underline{\tv}_{1} + g(m_{3})\underline{\tv}_{2}] \mod \Lambda
\end{eqnarray} where $\underline{\tv}_{k} = f(\underline{\wv}_{k})$. As in the proof of Theorem \ref{thm:FF}, we choose the precoding vector
$\mv=(m_{1},m_{2},m_{3})$ to satisfy Condition 2, such that
\begin{equation}\label{eq:cond5}
\Qm\Mm = \diag(\lambda_{11},\lambda_{22}) \mbox{ for some } \lambda_{11},\lambda_{22} \neq 0
\end{equation}
where $\Mm$ is related to $\mv$ as defined in (\ref{eq:B}).

Each receiver $k$ applied the CoF receiver mapping (\ref{eq:cof}) with integer coefficients vector $\bv_{k}$ and (scalar) scaling factor  $\alpha_{k}$, yielding
\begin{eqnarray}
\underline{\hat{\yv}}_{k} &=& \left[\bv_{k}^{\herm}\left[
                                                     \begin{array}{c}
                                                       \underline{\vv}_{1} \\
                                                       \underline{\vv}_{2} \\
                                                     \end{array}
                                                   \right] + \underline{\zv}_{\mbox{\tiny{eff}}}(\hv_{k},\bv_{k},\alpha_{k})
\right] \mod \Lambda\\\label{eq:map}
&=& \left[\bv_{k}^{\herm}g(\Mm)
\left[
                                                     \begin{array}{c}
                                                       \underline{\tv}_{1} \\
                                                       \underline{\tv}_{2} \\
                                                     \end{array}
                                                   \right] + \underline{\zv}_{\mbox{\tiny{eff}}}(\hv_{k},\bv_{k},\alpha_{k})
\right] \mod \Lambda\\
&\stackrel{(a)}{=}&\left[([\bv_{k}^{\herm}g(\Mm)] \mod p\ZZ[j])
\left[
                                                     \begin{array}{c}
                                                       \underline{\tv}_{1} \\
                                                       \underline{\tv}_{2} \\
                                                     \end{array}
                                                   \right] + \underline{\zv}_{\mbox{\tiny{eff}}}(\hv_{k},\bv_{k},\alpha_{k})
\right] \mod \Lambda\\
&\stackrel{(b)}{=}&\left[g(\lambda_{kk})\underline{\tv}_{k} + \underline{\zv}_{\mbox{\tiny{eff}}}(\hv_{k},\bv_{k},\alpha_{k})
\right] \mod \Lambda
\end{eqnarray}
where $\hv_k = [h_{k1}, h_{k2}]$,
where (a) follows from the fact that $[p \underline{\tv}] \mod \Lambda = \underline{\zerov}$ for any codeword $\underline{\tv} \in \Lc_{k}$,
and where (b) is due to the following result:

\begin{lemma} Let $\Qm =[\Bm^{\herm}]_{q}$. If $\Qm\Mm=\diag(\lambda_{11},\lambda_{22})$ over $\FF_q$, then
\begin{eqnarray}
[\Bm^{\herm} g(\Mm)] \mod p\ZZ[j] = \diag(g(\lambda_{11}),g(\lambda_{22})).
\end{eqnarray}
\end{lemma}
\begin{proof} Using $[\Bm^{\herm}] \mod p\ZZ[j] = g(\Qm)$, we have:
\begin{eqnarray}
[\Bm^{\herm} g(\Mm)] \mod p\ZZ[j] &=& [([\Bm^{\herm}] \mod p\ZZ[j])g(\Mm)] \mod p\ZZ[j]\\
&=& [g(\Qm)g(\Mm)] \mod p\ZZ[j]\\
&=& \left[g\left(\Qm\Mm\right)\right] \mod p\ZZ[j]\\
&=& \left[g\left(\diag(\lambda_{11},\lambda_{22})\right)\right] \mod p\ZZ[j]\\
&=& \diag(g(\lambda_{11}),g(\lambda_{22}))
\end{eqnarray}
\end{proof}
From the results summarized in Section \ref{subsec:CoF}, we know that lattice decoding applied to the observation
$\hat{\underline{\yv}}_{k}$ at each receiver $k$ can reliably decode the desired message if
\begin{equation}
 R_{k} \leq  \log^+ \left ( \frac{\SNR}{\bv_{k}^{\herm}(\SNR^{-1}\Id+\hv_{k} \hv_{k}^{\herm})^{-1}\bv_{k}} \right ),\label{eq:rateconst_temp}
\end{equation}
 where the above rate-expression is obtained from (\ref{eq:cofrate}) with $\Cm=\Id$.
This rate can be improved if each transmitter $k$ scales its signal by some factor $\beta_{k} \in \Pc$,
where $\Pc = \{\beta \in \CC: |\beta| \leq 1\}$ denotes the unit disk in $\CC$, since it can create more favorable channel coefficients for the integer conversion at each receiver \cite{Song-IT}.
This choice of $\beta_{k}$ guarantees that the power constraint is satisfied at each transmitter.
The effective channel matrix induced by this scaling is given by
\begin{equation}
\tilde{\Hm}(\beta_{1},\beta_{2})= \left[
                            \begin{array}{cc}
                              \beta_{1}h_{11} & \beta_{2}h_{12} \\
                              \beta_{1}h_{21} & \beta_{2}h_{22} \\
                            \end{array}
                          \right].
\end{equation} Using (\ref{eq:rateconst_temp}) and the effective channel matrix, each receiver $k$ can reliably decode the desired message if
\begin{equation}
 R_{k} \leq  \log^+ \left ( \frac{\SNR}{\bv_{k}^{\herm}(\SNR^{-1}\Id+\tilde{\hv}_{k} \tilde{\hv}_{k}^{\herm})^{-1}\bv_{k}} \right ),
\end{equation} where $\tilde{\hv}_{k} = [\beta_{1}h_{k 1}, \beta_{2} h_{k 2}]$. This completes the proof.

\subsection{An achievable rate region for the Gaussian Network-Coded CIC}\label{sec:CIC-GA}

It was shown in Section \ref{sec:CIC-FF} that distributed zero-forcing precoding is optimal for finite-field Network-Coded CIC.
In the Gaussian case, however, the channel coefficients are not integers and hence Scaled PCoF may not be optimal due to the non-integer penalty.
Using the fact that transmitter 1 has non-causal information of message 2, we can
completely eliminate the interference of signal from transmitter 2 at receiver 1 by using DPC \cite{Costa}.  Also, we can remove the non-integer penalty
at the receiver 2 by using Scaled PCoF with a careful choice of the scaling factor of transmitter 2.
In other words, while Scaled PCoF cannot simultaneously remove the non-integer penalty
at both receivers, it can {\em completely} eliminate the non-integer penalty at receiver 2 (see Remark~\ref{remark:non-integer}), while interference at receiver 1 is handled by DPC precoding. Using this scheme, we have:

\begin{theorem}\label{thm:DPC} If $\det(\Hm) \neq 0$ and $h_{11},h_{21}\neq 0$, Scaled PCoF and DPC applied to Gaussian Network-Coded CIC  achieves the rate
pairs $(R_{1},R_{2})$ such that
\begin{eqnarray}
R_{1} &\leq& \log(1+|h_{11}|^2\SNR)\label{eq:rate1}\\
R_{2} &\leq&  \log^+\left (\frac{\SNR}{\sigma^2_{\mbox{\tiny{eff}}}(\beta)} \right ),
\end{eqnarray}
for any $\bv \in \ZZ[j]^{2}$ with $b_{1},b_{2} \neq 0 \mod  p\ZZ[j]$ and any $\beta \in \CC$ with $|\beta|=1$,
where
\begin{equation} \sigma^2_{\mbox{\tiny{eff}}}(\beta) = \left | b_1 \frac{\beta \tilde{h}_{22}}{h_{21}} - b_2 \right |^2 \SNR + \left | \frac{b_1}{h_{21}} \right |^2. \label{eq:noise-term}\end{equation}
\hfill \IEEEQED
\end{theorem}

\begin{remark}\label{remark:non-integer} Differently from using only Scaled PCoF in Section~\ref{subsec:PCoF-scalar}, the proposed scheme in this section can completely eliminate the non-integer penalty term in (\ref{eq:noise-term}), by choosing $\beta = \frac{h_{21}}{\tilde{h}_{22}\gamma} $, $b_{1} = \gamma$, and $b_{2} = 1$, where $\gamma = \left\lceil\left|\frac{h_{21}}{\tilde{h}_{22}}\right|\right\rceil$. These choices provide an almost optimal performance at high SNRs. However, they may not give an optimal performance in the moderate SNRs, since the variance of additive noise term also increases especially with a large $b_{1}$. Therefore, we find an optimal allocation parameter $\beta$ to maximize an achievable sum-rate and it will be used to plot the performances of the proposed scheme.

\hfill $\lozenge$
\end{remark}

\begin{figure*}
\centerline{\includegraphics[width=16cm]{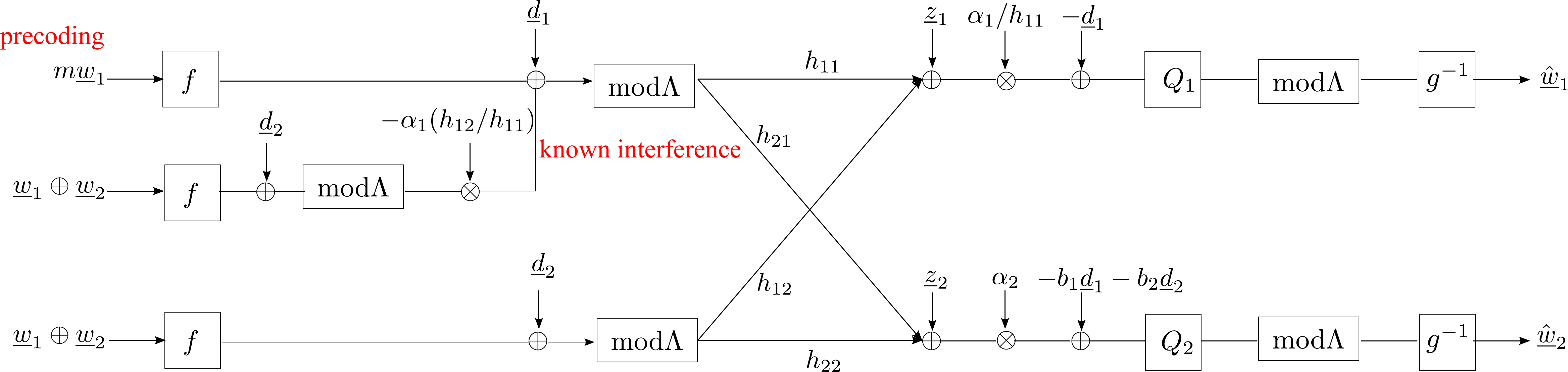}}
\caption{Encoding and decoding structures of the proposed achievability scheme. Transmitter 1 uses the DPC to cancel the interference at its intended receiver 1,
and also performs precoding over finite-field to eliminate the interference at receiver 2.}
\label{decoding}
\end{figure*}

We let $\bv=[b_{1},b_{2}] \in \ZZ[j]^{2}$ denote the integer coefficients vector used at receiver 2 for the CoF
receiver mapping (\ref{eq:cof}), and we let $q_{k} = [b_{k}]_{q}$.
Again, it is assumed that  $q_{1},q_{2} \neq 0$ over $\FF_q$.
The proposed achievability scheme proceeds as follows (see Fig.~\ref{decoding}):
\begin{itemize}
\item Transmitter 2 produces the lattice codeword $\underline{\vv}_{2}= f(\underline{\wv}_{1} \oplus \underline{\wv}_{2})$ and produces the channel input with power scaling factor $\beta \in \CC$ with $|\beta|=1$:
\begin{equation}
 \underline{\xv}_{2}= \beta \underline{\xv}_{2}'
\end{equation} where $\underline{\xv}_{2}' =  [\underline{\vv}_{2} + \underline{\dv}_{2}] \mod \Lambda$.
\item Transmitter 1 produces the precoded message $m \underline{\wv}_{1}$ where $m \in \FF_{q}$ is given by
\begin{equation}\label{eq:con}
q_{1}m \oplus q_{2} = 0 \Rightarrow m = (q_{1})^{-1}(-q_{2}),
\end{equation}
where $(q_{1})^{-1}$ denotes the multiplicative inverse of $q_{1}$ and $(-q_{2})$ denotes the additive inverse of $q_{2}$. Then,
it uses DPC for the known interference signal $h_{12}\underline{\xv}_{2}$ and forms:
\begin{equation}
\underline{\xv}_{1} = [\underline{\vv}_{1} - \alpha_{1}(h_{12}/h_{11})\underline{\xv}_{2} + \underline{\dv}_{1}] \mod \Lambda,
\end{equation} for some $\alpha_{1} \in \Cc$, where $\underline{\vv}_{1} = f(m\underline{\wv}_{1})$.
\end{itemize}
By lattice linearity we have:
\begin{eqnarray}
\underline{\vv}_{1} &=& [g(m)\underline{\tv}_{1}] \mod \Lambda\\
\underline{\vv}_{2} &=& [\underline{\tv}_{1} + \underline{\tv}_{2}] \mod \Lambda
\end{eqnarray} where $\underline{\tv}_{1} = f(\underline{\wv}_{1})$ and $\underline{\tv}_{2} = f(\underline{\wv}_{2})$.
Receivers 1 and 2 observe the $\underline{\yv}_{1}$ and $\underline{\yv}_{2}$ given in (\ref{eq:channel1}) and (\ref{eq:channel2}), respectively.
Receiver 1 performs the inflated modulo-lattice mapping as
$\hat{\underline{\yv}}_{1} = [\alpha_{1}\underline{\yv}_{1}/h_{11} - \underline{\dv}_{1}] \mod \Lambda$.
This results in the mod-$\Lambda$ additive noise channel given by:
\begin{eqnarray}
\hat{\underline{\yv}}_{1} &=& [(\alpha_{1}/h_{11}) [h_{11}\underline{\xv}_{1} + h_{12} \underline{\xv}_{2} + \underline{\zv}_{1}] -\underline{\dv}_{1}]\mod \Lambda \\
&=&[\underline{\vv}_{1} - \underline{\vv}_{1} + \alpha_{1}\underline{\xv}_{1} + \alpha_{1}(h_{12}/h_{11}) \underline{\xv}_{2} + (\alpha_{1}/h_{11})\underline{\zv}_{1} - \underline{\dv}_{1}] \mod \Lambda\\
&\stackrel{(a)}{=}& [\underline{\vv}_{1} -(1-\alpha_{1})\underline{\xv}_{1} + (\alpha_{1}/h_{11})\underline{\zv}_{1} ] \mod \Lambda,
\end{eqnarray} where (a) is due to the fact that $\underline{\xv}_{1} = [\underline{\vv}_{1} - \alpha_{1}(h_{12}/h_{11})\underline{\xv}_{2} + \underline{\dv}_{1}] \mod \Lambda$. Hence, the resulting channel from $\underline{\vv}_{1}$ to $\hat{\underline{\yv}}_{1}$ is equivalent in distribution to the point-to-point additive modulo-$\Lambda$ channel
\begin{equation*}
\hat{\underline{\yv}}_{1} =  [\underline{\vv}_{1} -(1-\alpha_{1})\underline{\uv}_{1} + (\alpha_{1}/h_{11})\underline{\zv}_{1} ] \mod \Lambda,
\end{equation*} where $\underline{\uv}_{1}$  is a random variable uniformly distributed on $\Vc_{\Lambda}$ and is statistically independent of $\underline{\zv}_{1}$ and $\underline{\vv}_{1}$ by the independence and uniformity of dithering and by the Crypto Lemma.
From standard DPC results \cite{Zamir}, choosing
\begin{equation}
\alpha_{1}=\alpha_{1,\mbox{\tiny{MMSE}}} \eqdef \frac{\SNR |h_{11}|^2}{1 + \SNR |h_{11}|^2},\label{def:alpha1}
\end{equation}
the coding rate $R_1$ is achievable if
\begin{equation}
R_{1} \leq \log(1+|h_{11}|^2\SNR).
\end{equation}
Letting $\tilde{\hv}(\beta)=[h_{21},\beta\tilde{h}_{22}]$ with $\tilde{h}_{22}=h_{22} - \alpha_{1,\mbox{\tiny{MMSE}}}h_{12}h_{21}/h_{11}$,
receiver 2 applies the CoF receiver mapping (\ref{eq:cof}) with integer coefficients $\bv$ and scaling factor $\alpha_{2} = b_{1}/h_{21}$, yielding
\begin{eqnarray*}
\hat{\underline{\yv}}_{2} &=& [\alpha_{2}\underline{\yv}_{2} - b_{1}\underline{\dv}_{1} - b_{2}\underline{\dv}_{2}] \mod \Lambda \\
&=&[b_{1}\underline{\vv}_{1} + b_{2} \underline{\vv}_{2} +\alpha_{2}(h_{21}\underline{\xv}_{1}+h_{22}\underline{\xv}_{2} + \underline{\zv}_{2})- b_{1}[\underline{\vv}_{1}+\underline{\dv}_{1}]-b_{2}[\underline{\vv}_{2}+\underline{\dv}_{2}]] \mod \Lambda\\
&=&[b_{1}\underline{\vv}_{1}+b_{2}\underline{\vv}_{2}+\alpha_{2}h_{21}[\underline{\vv}_{1}-\alpha_{1,\mbox{\tiny{MMSE}}}(h_{12}/h_{11})\underline{\xv}_{2}+\underline{\dv}_{1}+\underline{\lambdav}] + \alpha_{2}h_{22}\underline{\xv}_{2}+\alpha_{2}\underline{\zv}_{2} \\
&&- b_{1}[\underline{\vv}_{1}+\underline{\dv}_{1}]-b_{2}\underline{\xv}_{2}'] \mod  \Lambda \\
&=&[b_{1}\underline{\vv}_{1} + b_{2} \underline{\vv}_{2} + (\alpha_{2}h_{21}-b_{1})[\underline{\vv}_{1}+\underline{\dv}_{1}] + (\alpha_{2}\beta\tilde{h}_{22} - b_{2})\underline{\xv}_{2}' +\alpha_{2}h_{21}\underline{\lambdav} + \alpha_{2}\underline{\zv}_{2}] \mod \Lambda\\
&\stackrel{(a)}{=}& \left[\bv^{\transp}\left[
                                                 \begin{array}{c}
                                                   \underline{\vv}_{1} \\
                                                   \underline{\vv}_{2} \\
                                                 \end{array}
                                               \right]+(b_{1}\beta\tilde{h}_{22}/h_{21}-b_{2})\underline{\xv}_{2}'+(b_{1}/h_{21})\underline{\zv}_{2}\right] \mod \Lambda,
 \end{eqnarray*} where $\underline{\lambdav}=Q_{\Lambda}(\underline{\vv}_{1}-\alpha_{1,\mbox{\tiny{MMSE}}}\beta(h_{12}/h_{11})\underline{\xv}_{2}+\underline{\dv}_{1})$ and $(a)$ is due to the fact that $\alpha_{2}h_{21}\underline{\lambdav}= b_{1}\underline{\lambdav} \in \Lambda$. As explained above, the resulting channel is equivalent in distribution to the following modulo-$\Lambda$ channel
\begin{eqnarray*}
\hat{\underline{\yv}}_{2}
&=& \left[\bv^{\transp}\left[
                                                 \begin{array}{c}
                                                   \underline{\vv}_{1} \\
                                                   \underline{\vv}_{2} \\
                                                 \end{array}
                                               \right]+(b_{1}\beta\tilde{h}_{22}/h_{21}-b_{2})\underline{\uv}_{2}+(b_{1}/h_{21})\underline{\zv}_{2}\right] \mod \Lambda\\
&=&\Big[\left(\bv^{\transp}\left[
                        \begin{array}{cc}
                          g(m) & 0 \\
                          1 & 1 \\
                        \end{array}
                      \right] \mod p\ZZ[j]\right)
\left[
                                                 \begin{array}{c}
                                                   \underline{\tv}_{1} \\
                                                   \underline{\tv}_{2} \\
                                                 \end{array}
                                               \right]+ \underline{\zv}_{\mbox{\tiny{eff}}}(\tilde{\hv}(\beta),\bv)
\Big]  \mod \Lambda\\
&\stackrel{(a)}{=}& [([b_{2}] \mod p\ZZ[j])\underline{\tv}_{2} +\underline{\zv}_{\mbox{\tiny{eff}}}(\tilde{\hv}(\beta),\bv)] \mod \Lambda
\end{eqnarray*} where $\underline{\uv}_{2}$ is uniformly distributed on $\Vc_{\Lambda}$ and is independent of $\underline{\vv}_{1}$, $\underline{\vv}_{2}$, and $\underline{\zv}_{2}$ by the independence and uniformity of dithering and by the Crypto Lemma, and $(a)$ follows from the fact that
$m$ is chosen to satisfy (\ref{eq:con}), i.e., $b_{1}g(m) + b_{2} \mod p\ZZ[j] = 0$. Furthermore, we define
\begin{equation}
 \underline{\zv}_{\mbox{\tiny{eff}}}(\tilde{\hv}(\beta),\bv) = (b_{1}\beta\tilde{h}_{22}/h_{21} - b_{2})\underline{\uv}_{2} + (b_{1}/h_{21})\underline{\zv}_{2}.\label{eq:enoise}
\end{equation}
Receiver 2 decodes $\underline{\tv}_2$ by applying lattice decoding to $\hat{\underline{\yv}}_{2}$ if
\begin{equation}
R_{2} \leq  \log^+\left (\frac{\SNR}{\sigma^2_{\mbox{\tiny{eff}}}(\beta)} \right ),\label{eq:R2_proof}
\end{equation} where
\begin{equation} \sigma^2_{\mbox{\tiny{eff}}}(\beta) = \left | b_1 \frac{\beta \tilde{h}_{22}}{h_{21}} - b_2 \right |^2 \SNR + \left | \frac{b_1}{h_{21}} \right |^2.
\end{equation} This completes the proof of Theorem~\ref{thm:DPC}.

\begin{example} We evaluate the performance of proposed schemes with respect to their average achievable sum rate, where
averaging is with respect to the channel realizations with i.i.d. coefficients $h_{ij} \sim \Cc\Nc(0,1)$.
Also, we considered the performance of full-cooperation (i.e., vector broadcast channel with sum-power constraint (see for example \cite{Yu} for an efficient algorithm to compute the vector broadcast channel sum-capacity).
In Fig.~\ref{simulation1}, Scaled PCoF shows the satisfactory performance in the moderate SNRs (i.e., $\SNR < 20$ dB).
Yet, this scheme suffers from the non-integer penalty at high SNRs.
Remarkably, Scaled PCoF with DPC (and optimization with respect to the scaling factor $\beta$ in Theorem \ref{thm:DPC}) performs within a constant gap
with respect to full-cooperation at any SNR.  \hfill $\lozenge$
\end{example}

\begin{figure}
\centerline{\includegraphics[width=14cm]{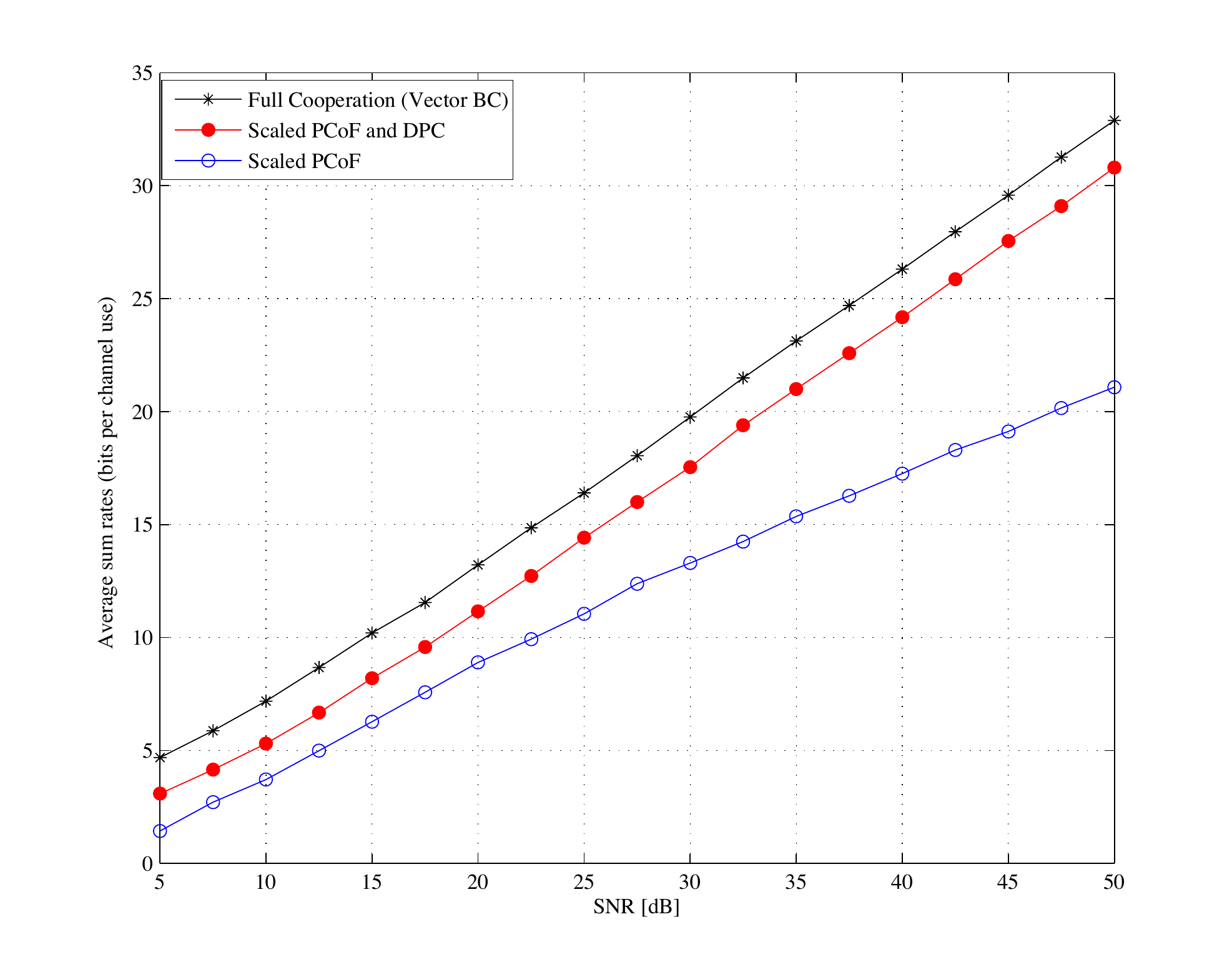}}
\caption{Average sum rate for Gaussian Network-Coded CIC with i.i.d. channel coefficients $\sim \Cc\Nc(0,1)$.}
\label{simulation1}
\end{figure}

\subsection{Generalized Degrees of Freedom}\label{subsec:GDoF}

In the high SNR regime, a useful proxy for the performance of wireless networks
is provided by the Generalized Degrees-of-Freedom (GDoFs), which characterize the capacity pre-log factor
in different relative scaling regimes of the channel coefficients, as SNR grows to infinity \cite{Etkin}.
In this section we study the symmetric GDoFs.  In particular, we consider the following channel model:
\begin{eqnarray}
\underline{\yv}_{1} &=& h_{11}\sqrt{\SNR} \underline{\xv}_{1} + h_{12}\sqrt{\INR} \underline{\xv}_{2} + \underline{\zv}_{1}\\
\underline{\yv}_{2} &=& h_{21}\sqrt{\INR} \underline{\xv}_{1} + h_{22}\sqrt{\SNR} \underline{\xv}_{2} + \underline{\zv}_{2}
\end{eqnarray}where $h_{ij} \in \CC$ are bounded non-zero constants independent of $\SNR, \INR$,
$\underline{\zv}_{k}$ is the i.i.d. Gaussian noise $\sim \Cc\Nc(0,1)$, and $\frac{1}{n}\EE[\|\underline{\xv}_{k}\|^2] \leq 1$ for $k=1,2$.
The channel is parameterized by $\SNR$ and $\INR$, both growing to infinity such that
$\INR = \SNR^{\rho}$ as $\SNR \rightarrow \infty$, where $\rho \geq 0$ defines the relative strength
of the direct and interference paths.


Letting $\Cc(\SNR,\rho)$ denote the capacity region of network-coded CIC for given $\SNR$ and $\rho$, the symmetric GDoF region (denoted by $\Dc(\rho)$ is defined by
\begin{equation}
\Dc(\rho) = \left\{(d_{1}(\rho),d_{2}(\rho): d_{i}(\rho) = \lim_{\SNR \rightarrow \infty}\frac{R_{i}}{\log{\SNR}} \mbox{ such that } (R_{1},R_{2}) \in \Cc(\SNR,\rho) \right\}.
\end{equation}The main result of this section is given by:
\begin{theorem}\label{thm:GDoFregion}
For the Gaussian Network-Coded CIC, the symmetric GDoF region ($\Dc(\rho)$) is the set of the DoF tuples $(d_{1}(\rho),d_{2}(\rho))$ satisfying the following constraints:
\begin{eqnarray}
d_{1}(\rho) &\leq& \max\{1,\rho\}\\
d_{2}(\rho) &\leq& \max\{1,\rho\}\\
d_{\mbox{\tiny{sum}}}(\rho)=d_{1}(\rho)+d_{2}(\rho) &\leq& 1+\rho.\label{sum-GDoF}
\end{eqnarray}
\end{theorem}
\begin{IEEEproof}
See Appendix~\ref{proof:GDoF}.
\end{IEEEproof}

%

In order to demonstrate the benefit gain of the mixed message at the non-cognitive transmitter,
we compare the sum GDoF (defined by $d_{\mbox{\tiny{sum}}}(\rho)$) of Gaussian IC and Gaussian CIC.
The sum GDoF of Gaussian IC is computed in \cite{Etkin}, and it is given by
\begin{equation}
d_{\mbox{\tiny{sum}}}(\rho) = \left\{
                                \begin{array}{ll}
                                  2(1-\rho), & 0\leq \rho < \frac{1}{2} \\
                                  2\rho, & \frac{1}{2} \leq \rho < \frac{2}{3} \\
                                  2-\rho, & \frac{2}{3} \leq \rho < 1 \\
                                  \rho, & 1 \leq \rho < 2 \\
                                  2, & \rho \geq 2.
                                \end{array}
                              \right.
\end{equation}
Also, from the constant gap result in \cite{Rini}, we can immediately compute the sum GDoF
of Gaussian CIC as
\begin{eqnarray}
d_{\mbox{\tiny{sum}}}(\rho)
&=& \left\{
      \begin{array}{ll}
        2-\rho, & \rho \leq 1 \\
        \rho, & \rho > 1.
      \end{array}
    \right.
\end{eqnarray}
The sum symmetric GDoF of these three channel models are shown in Fig.~\ref{GDoF}.
It is also immediate to observe that the sum GDoF of full-cooperation is given by
\begin{equation} \label{sum-GDoF-fullcoop}
d_{\mbox{\tiny{sum}}}(\rho) = 2\times \max\{1,\rho\}.
\end{equation}
In this case, the upper bound can be obtained from the $2 \times 2$ MIMO capacity with full CSI,
and an easily analyzable achievable scheme consists of employing simple linear precoding given by $\beta(\Hm')^{-1}\Bm$, where $\beta$ denotes a scaling value to normalize the precoding matrix and
\begin{eqnarray}
&&\Hm' =\left[
        \begin{array}{cc}
          h_{11} & \sqrt{\frac{\INR}{\SNR}}h_{12} \\
          \sqrt{\frac{\INR}{\SNR}}h_{21} & h_{22}\\
        \end{array}
      \right] \mbox { and } \Bm = \left[
        \begin{array}{cc}
          1 & 0 \\
          0 & 1 \\
        \end{array}
      \right] \mbox{ for }\rho \leq 1\\
&& \Hm' =\left[
        \begin{array}{cc}
          \sqrt{\frac{\SNR}{\INR}}h_{11} & h_{12} \\
          h_{21} & \sqrt{\frac{\SNR}{\INR}}h_{22}\\
        \end{array}
      \right] \mbox { and } \Bm = \left[
        \begin{array}{cc}
          0 & 1 \\
          1 & 0 \\
        \end{array}
      \right] \mbox{ for } \rho > 1,
\end{eqnarray} where notice that $\Hm'$ is a constant-valued matrix when $\SNR \rightarrow \infty$ for any $\rho$. Then, receiver $k$ can observe an interference-free signal as
\begin{eqnarray}
\underline{\yv}_{k}=\left\{
                      \begin{array}{ll}
                        \beta \sqrt{\SNR}\underline{\xv}_{k} + \underline{\zv}_{k}, & \hbox{for $\rho \leq 1$} \\
                        \beta \sqrt{\INR}\underline{\xv}_{k} + \underline{\zv}_{k}, & \hbox{for $\rho > 1$.}
                      \end{array}
                    \right.
\end{eqnarray} Since $\beta$ is a constant, the sum-DoF in (\ref{sum-GDoF-fullcoop}) is achieved. Observing that (\ref{sum-GDoF}) and (\ref{sum-GDoF-fullcoop}) coincide for $\rho=1$,
we conclude that the sum DoF of the Network-Coded CIC coincides with the sum DoF of full-cooperation, while the sum GDoF
is strictly worse than full cooperation when $\rho \neq 1$. Furthermore, the network-coded cognition yields higher sum
GDoFs than the conventional cognition when $\rho \geq 1/2$ and higher sum GDoFs than the standard IC
when $\rho \geq 1/3$.

\begin{remark}
Apparently, having a rank-1 linear combination of both messages at transmitter 2 instead of just message 2
hurts for small $\rho$ (weak interference) and it is helpful in the intermediate to strong interference regime. Obviously,
in a system where the backhaul network is rate-constrained but can be optimized with respect to the employed
network code used, one would dispatch to the non-cognitive  transmitter its own message only if the wireless segment operates in the regime
of weak interference, and a linear combination of the two messages if it operates in the medium or strong interference regimes, thus obtaining
the upper envelope of the conventional and Network-Coded CIC sum GDoF. \hfill $\lozenge$
\end{remark}


\section{Two-User MIMO IC: Coordination, Cognition, Two-Hop}\label{sec:twoMIMO}

In this section, we study three communication channels (see Fig.~\ref{system-model}) with the two-user MIMO IC as a
building block, namely, Network-Coded ICC (representative of a cellular system downlink with interference coordination),
Network-Coded CIC (the MIMO generalization of the model of Section \ref{sec:CIC-FF}), and $2\times 2\times 2$ IC (a {\em canonical}
two-flows two-hop network that has attracted considerable attention in recent literature \cite{Gou,Shomorony}).
In all these models, we assume that all nodes have $M$ transmit/receive antennas.
Let $\underline{\wv}_{k,\ell} \in \FF_{q}^{r}$, $\ell=1,\ldots,M$, denote the independent messages intended for destination $k$, for $k=1,2$.
For simplicity of exposition, we define the message matrix $\underline{\Wm}_{k}$ with rows $\underline{\wv}_{k,1},\ldots,\underline{\wv}_{k,M}$,
where $\underline{\wv}_{k,\ell}$ can be all-zero vectors for $\ell > S_{k}$ if user $k$ has $S_{k}$ independent
information messages.

In the Network-Coded ICC, the  source has no knowledge of the CSI and can deliver fixed (i.e., not dependent on the wireless channel matrices)
linear combinations of the information messages to each transmitter, such that
each transmitter $k$ knows $M$ linear combinations as $\Sm_{k1} \underline{\Wm}_{1}\oplus\Sm_{k2}\underline{\Wm}_{2}$, for suitable integer matrices
$\Sm_{ki}$.  In the wireless channel, a block of $n$ channel uses of the discrete-time complex baseband MIMO IC is described by
\begin{equation}
\left[
  \begin{array}{c}
    \underline{\Ym}_{1} \\
    \underline{\Ym}_{2} \\
  \end{array}
\right]=\left[
                                             \begin{array}{cc}
                                               \Fm_{11} & \Fm_{12} \\
                                               \Fm_{21} & \Fm_{22} \\
                                             \end{array}
                                           \right]\left[
  \begin{array}{c}
    \underline{\Xm}_{1} \\
    \underline{\Xm}_{2} \\
  \end{array}
\right]+\left[
  \begin{array}{c}
    \underline{\Zm}_{1} \\
    \underline{\Zm}_{2} \\
  \end{array}
\right]\label{model:brn}
\end{equation}
where the matrices $\underline{\Xm}_{k}$ and $\underline{\Ym}_{k}$ contain, arranged by rows, the channel input sequences
$\underline{\xv}_{k,\ell} \in \CC^{1 \times n}$, the channel output sequences $\underline{\yv}_{k,\ell} \in \CC^{1 \times n}$,
and where $\Fm_{j k} \in \CC^{M \times M}$ denotes the channel matrix between transmitter $k$ and receiver $j$.
The Network-Coded CIC has the wireless channel component given in (\ref{model:brn}), but in this case
the two transmitters have {\em different} knowledge on the messages. In particular,
transmitter 1 (the cognitive transmitter) knows both messages $\underline{\Wm}_{1},\underline{\Wm}_{2}$
and transmitter 2 (the non-cognitive transmitter) only knows linear combinations $\Sm_{21}\underline{\Wm}_{1}\oplus \Sm_{22}\underline{\Wm}_{2}$, where the rank of the linear combinations is not sufficient to recover the individual messages. Finally, we consider the $2\times 2 \times 2$ IC, as shown in Fig.~\ref{system-model} (c),
where each transmitter $k$ (referred to as ``source'' in this relay setting) has a message for its intended destination $k$, for $k=1,2$.
In this model, the first hop is also described by (\ref{model:brn}) and in the second hop  a block of $n$ channel uses of the discrete-time complex MIMO IC
is described by
\begin{equation}
\left[
  \begin{array}{c}
    \underline{\Ym}_{3} \\
    \underline{\Ym}_{4} \\
  \end{array}
\right]=\left[
                                             \begin{array}{cc}
                                               \Fm_{33} & \Fm_{34} \\
                                               \Fm_{43} & \Fm_{44} \\
                                             \end{array}
                                           \right]\left[
  \begin{array}{c}
    \underline{\Xm}_{3} \\
    \underline{\Xm}_{4} \\
  \end{array}
\right]+\left[
  \begin{array}{c}
    \underline{\Zm}_{3} \\
    \underline{\Zm}_{4} \\
  \end{array}
\right].\label{model:2hop}
\end{equation}
where we denote the two transmitter-receiver pairs in the second hop by $k = 3,4$, and
where $\underline{\Zm}_{k}$ contains i.i.d. Gaussian noise samples $\sim \Cc\Nc(0,1)$.
We assume that the elements of $\Fm_{j k}$ are drawn i.i.d. according to a continuous distribution (i.e., Gaussian distribution).
The channel matrices are assumed to be constant over the whole block of length $n$ and known to all nodes, and
we consider a total power constraint equal to $P_{\rm{sum}}$ at each transmitter (both sources and relays). Also, it is assumed that relays operate in a full-duplex mode.

Before stating the main results of this section, it is useful to introduce the following notation.
With reference to Section \ref{subsec:CoF}, for a set of modulo-$\Lambda$ additive noise channel of the type
(\ref{eq:cof}), induced by nested lattice coding, by the channel matrix $\Hm_k \Cm_k$ and by the integer combining matrix $\Bm_k$ with columns $\bv_{k,\ell}$,
for $k = 1,2$ and $\ell = 1, \ldots, L_k$, for some integer $L_k$, we define
\begin{equation}
R_{\rm comp}(\Hm_k\Cm_k, \Bm_k, \SNR) = \min_{\ell = 1, \ldots, L_k} \left \{ \log^+ \left (\frac{\SNR}{\sigma^2_{\mbox{\tiny{eff}},k,\ell}} \right ) \right \}, \label{eq:Rcomp}
\end{equation}
where
\begin{eqnarray}
\sigma^{2}_{\mbox{\tiny{eff}},k,\ell}
&=& \bv_{k,\ell}^{\herm}\Cm_k(\SNR^{-1}\Id+\Cm_k^{\herm}\Hm_k^{\herm}\Hm_k\Cm_k)^{-1}\Cm_k^{\herm} \bv_{k,\ell}.
\end{eqnarray} Also, we define the {\em constant} matrices:
\begin{eqnarray}
\Cm_{12}&\eqdef&\left[
                                                                \begin{array}{c}
                                                                  \mbox{0}_{1 \times (M-1)} \\
                                                                  \Id_{M-1} \\
                                                                \end{array}
                                                              \right]\label{eq:C12}\\
\Cm_{22}&\eqdef&\left[
                                                                \begin{array}{c}
                                                                  \Id_{M-1} \\
                                                                  \mbox{0}_{1 \times (M-1)} \\
                                                                \end{array}
                                                              \right].\label{eq:C22}
\end{eqnarray}

With this notation, we have:

\begin{theorem}\label{thm:122}For the Network-Coded ICC and Network-Coded CIC,
PCoF with CIA  can achieve the {\em symmetric} sum rate of $(2M-1)R$ with all messages of the  same rate given by
\begin{equation}
R=\min_{k=1,2}\left\{R_{\rm comp}(\Hm_{k}\Cm_{k},\Bm_{k},\SNR)\right\}
\end{equation}
for any full-rank integer matrices $\Am_{1} \in \ZZ[j]^{M \times M}, \Am_{2} \in \ZZ[j]^{(M-1) \times (M-1)}$ and $\Bm_{1}, \Bm_{2} \in \ZZ[j]^{M \times M}$,
and any alignment precoding matrices $\Vm_{k}$ satisfying the alignment conditions in (\ref{cond:ALI}), where
\begin{eqnarray}
\Hm_{k} &=& \Fm_{k1}\Vm_{1}, \;\;\;\; \Cm_{k}=\left[
\begin{array}{ccc}
\Am_{1}  & \Cm_{k2}\Am_{2}
\end{array}
\right]\\
\SNR&=&\min_{k=1,2}\left\{\frac{P_{\rm{sum}}}{\trace{\left(\Vm_{k}\Am_{k}\Am_{k}^{\herm}\Vm_{k}^{\herm}\right)}}\right\}.
\end{eqnarray}
\hfill \IEEEQED
\end{theorem}

\begin{theorem}\label{thm:222}
For the $2 \times 2 \times 2$ IC,  PCoF with CIA can achieve the {\em symmetric} sum rate of $(2M-1)R$ with all messages of the same rate given by
\begin{equation}
R=\left\{ \min_{k=1,2}\{R_{\rm comp}(\Hm_{k}\Cm_{k},\Bm_{k},\SNR)\}, \min_{k=3,4}\{R_{\rm comp}(\Hm_{k}\Cm_{k},\Bm_{k}, \SNR')\}\right\},
\end{equation}
for any full-rank integer matrices $\Am_{1}, \Am_{3} \in \ZZ[j]^{M \times M}, \Am_{2}, \Am_{4} \in \ZZ[j]^{(M-1) \times (M-1)}$ and $\Bm_{k} \in \ZZ[j]^{M \times M}, k=1,\ldots, 4$,
and any alignment precoding matrices
$\Vm_{k}$ satisfying the alignment conditions in (\ref{cond:ALI}), where
\begin{eqnarray}
\Hm_{k} &=& \Fm_{k1}\Vm_{1}, \;\;\;\; \Cm_{k} =
\left[
\begin{array}{ccc}
\Am_{1}  & \Cm_{k2}\Am_{2}
\end{array}
\right], \;\;\;\; k=1,2\\
\Hm_{k}&=&\Fm_{k3}\Vm_{3}, \;\;\;\;
\Cm_{k} = \left[
\begin{array}{ccc}
\Am_{3}  & \Cm_{(k-2)2}\Am_{4}
\end{array}
\right], \;\;\;\; k=3,4\\
\SNR&=&\min_{k=1,2}\left\{\frac{P_{\rm{sum}}}{\trace{\left(\Vm_{k}\Am_{k}\Am_{k}^{\herm}\Vm_{k}^{\herm}\right)}}\right\}\\
\SNR'&=& \min_{k=3,4}\left\{\frac{P_{\rm{sum}}}{\trace{\left(\Vm_{k}\Am_{k}\Am_{k}^{\herm}\Vm_{k}^{\herm}\right)}}\right\}.
\end{eqnarray}
\hfill \IEEEQED
\end{theorem}

The next result shows that the per-message rate $R$ grows as $\log{\SNR}$ when $P_{\rm{sum}} \rightarrow \infty$, thus obtaining an achievable sum DoF result for all the
above channel models:
\begin{corollary}\label{cor:DoF}
PCoF with CIA achieves sum DoF equal to $(2M-1)$ for the Network-Coded ICC, Network-Coded CIC,
and $2\times2\times 2$ IC, when all nodes have $M$ multiple antennas.
\end{corollary}
\begin{IEEEproof}
See Appendix~\ref{proof:DoF}.
\end{IEEEproof}

For the Network-Coded CIC, we can improve the DoF by appropriately combining the DPC and PCoF as done in Section~\ref{sec:CIC-GA} for single antenna case.
Exploiting this idea, we obtain:

\begin{theorem}\label{thm:RateCIC} For the Network-Coded CIC, PCoF and DPC can achieve the following sum-rate:
\begin{equation}
R_{{\rm sum}} = \sum_{\ell=1}^{M} \log\left(1+\frac{\SNR}{\|\Fm_{11}^{-1}(\ell)\|^2}\right) + \sum_{\ell=1}^{M}\log\left(1+\frac{\SNR}{\|\Fm_{21}^{-1}(\ell)\|^2\|\Vm(\ell)\|^2}\right)
\end{equation} where $\Fm_{11}^{-1}(\ell)$, $\Fm_{21}^{-1}(\ell)$, and $\Vm(\ell)$ denotes the $\ell$-th column of $\Fm_{11}^{-1}$, $\Fm_{21}^{-1}$, and $\Vm$, respectively, and where
$\Vm=\left(\Fm_{21}^{-1}\Fm_{22}-\Fm_{11}^{-1}\Fm_{12}\right)^{-1}$.
\end{theorem}
\begin{IEEEproof}
See Appendix~\ref{proof:MIMOCIC}.
\end{IEEEproof}

\begin{theorem}\label{thm:DoFCIC} For the Network-Coded CIC, the sum DoF is equal to $2M$ when all nodes have $M$ multiple antennas.
\end{theorem}
\begin{IEEEproof}
The proof is immediately done from Theorem~\ref{thm:RateCIC}.
\end{IEEEproof}

The proofs of Theorems~\ref{thm:122} and~\ref{thm:222} are provided in Sections~\ref{subsec:CoF-SA} and~\ref{subsec:LP}.
Our achievable scheme is based on the extension of the PCoF approach to the MIMO case.
This scheme consists of two phases: 1) Using the CoF framework, we transform the two-user MIMO IC into a {\em deterministic}
finite-field IC. 2) A linear precoding scheme is used over finite-field to eliminate the interferences (see Figs.~\ref{precoding-222}).
The main performance bottleneck of CoF consists of the non-integer penalty, which ultimately limits the performance of
CoF at high SNR \cite{Niesen1}. To overcome this bottleneck, we employ CIA in order to create an ``aligned" channel
matrix for which exact integer forcing is possible, similarly to what was done in Section~\ref{subsec:CoF}.

Specifically,  we use alignment precoding matrices $\Vm_{1}$ and $\Vm_{2}$ at the two transmitters such that
\begin{equation}
\left[
\begin{array}{ccc}
 \Fm_{k1}\Vm_{1} & \Fm_{k2}\Vm_{2} \\
\end{array}
\right] = \Hm_{k}\Cm_{k},
\end{equation}
where $\Hm_{k} \in \CC^{M\times M}$ and $\Cm_{k} \in \ZZ[j]^{M \times 2M-1}$.
Linear precoding over the complex field may produce a power-penalty due to the non-unitary nature
of the alignment matrices, and this can degrade the performance at finite SNR.
In order to counter this effect, we use {\em Integer Forcing Beamforming} (IFB) \cite{Song-ISIT}.
The main idea is that $\Vm_{k}$ can be pre-multiplied (from the right) by some appropriately
chosen full-rank integer matrix $\Am_{k}$ since its effect can be undone by precoding
over $\FF_{q}$, using $[\Am_{k}]_{q}^{-1}$. Then, we can optimize the integer matrix in order to minimize
the power penalty of alignment. The optimization of alignment and IFB in order to obtain good finite SNR performance is
postponed to Section \ref{sec:finite}. The details of the coding scheme are given in the following sections.

\subsection{CoF Framework based on  Channel Integer Alignment} \label{subsec:CoF-SA}

In this section we show how to turn any two-user MIMO IC into a deterministic finite-field IC using the CoF framework.
Consider the MIMO IC in (\ref{model:brn}).
For $k = 1,2$, let $\underline{\Wm}_{\mbox{\tiny{T}}_{k}} = \Sm_{k1}\underline{\Wm}_1 \oplus \Sm_{k2} \underline{\Wm}_2$ denote
the network coded messages at transmitter $k$. We let
$\underline{\wv}_{\mbox{\tiny{T}}_{k,\ell}} \in \FF_{q}^{r}$ denote the $\ell$-th row of $\underline{\Wm}_{\mbox{\tiny{T}}_{k}}$, and we let
$\underline{\Wm}_{\mbox{\tiny{T}}_{1}}, \underline{\Wm}_1$ have dimension $M \times r$ and
$\underline{\Wm}_{\mbox{\tiny{T}}_{2}}, \underline{\Wm}_2$ have dimension $(M - 1) \times r$.
The precoding matrices $\Sm_{k1}, \Sm_{k2}$ over $\FF_q$ will be determined in Section~\ref{subsec:LP}.
We let $\Vm_{1}=[\vv_{1,1}, \ldots, \vv_{1,M}] \in \CC^{M \times M}$ and $\Vm_{2} = [\vv_{2,1}, \ldots, \vv_{2,M-1}]  \in \CC^{M \times (M-1)}$
denote the precoding matrices used at transmitters 1 and 2, respectively, chosen to satisfy the {\em alignment conditions}
\begin{eqnarray}
\Fm_{11}\vv_{1,\ell+1} &=& \Fm_{12}\vv_{2,\ell}\nonumber\\
\Fm_{21}\vv_{1,\ell} &=& \Fm_{22}\vv_{2,\ell}\label{cond:ALI}
\end{eqnarray}
for $\ell=1,\ldots,M-1$. The feasibility of conditions (\ref{cond:ALI}) is shown in \cite{Gou} for any integer $M\geq 2$, almost surely with respect to the
continuously distributed channel matrices $\{\Fm_{jk}\}$.

Let $\Am_{1} \in \ZZ[j]^{M \times M}$ and $\Am_{2} \in \ZZ[j]^{(M-1) \times (M-1)}$ denote full rank integer matrices
(the optimization of which in order to minimize the transmit power penalty is discussed in Section \ref{sec:finite}).
The transmitters make use of the same lattice code $\Lc$ of rate $R$, where $\Lambda$ is chosen such that $\sigma_{\Lambda}^2 =\SNR$.
Then, CoF based on CIA proceeds as follows.

\noindent
{\bf Encoding:}
\begin{itemize}
\item Each transmitter $k$ precodes its messages over $\FF_{q}$ as
\begin{equation}
\underline{\Wm}'_{\mbox{\tiny{T}}_{k}} = [\Am_{k}]_{q}^{-1}\underline{\Wm}_{\mbox{\tiny{T}}_{k}}, \;\;\;\; k=1,2. \label{eq:precoding}
\end{equation}
Then, the precoded messages (rows of $\underline{\Wm}'_{\mbox{\tiny{T}}_{k}}$) are encoded using the nested lattice codes as
$\underline{\tv}'_{k,\ell} = f(\underline{\wv}'_{\mbox{{\tiny T}}_{k,\ell}})$.
Finally, the channel input sequences are given by the rows of
\begin{equation} \label{suca1}
\underline{\Xm}''_{k} =  \Vm_{k}\Am_{k}\underline{\Xm}'_{k},
\end{equation}
where $\underline{\Xm}'_k$ has rows $\underline{\xv}'_{k,\ell} = [\underline{\tv}'_{k,\ell}+\underline{\dv}_{k,\ell}] \mod \Lambda$.
\end{itemize}
Due to the sum-power constraint equal to $P_{\rm{sum}}$ at each transmitter, the second moment of coarse lattice (i.e., $\SNR$) must satisfy
\begin{equation}
\SNR\cdot\trace(\Vm_{k}\Am_{k}\Am_{k}^{\herm}\Vm_{k}^{\herm}) \leq P_{\rm{sum}} \mbox{ for }k=1,2.
\end{equation} Thus, we can choose:
\begin{equation}
\SNR= \min \left \{   \frac{P_{\rm{sum}}}{\trace(\Vm_{k}\Am_{k}\Am_{k}^{\herm}\Vm_{k}^{\herm})} \;\; : \;\; k=1,2 \right \}. \label{eq:power-const}
\end{equation}

\noindent
{\bf Decoding:}
\begin{itemize}
\item  Receiver 1 observes:
\begin{eqnarray}
\underline{\Ym}_{1} &=& \Fm_{11}\underline{\Xm}''_{1}+\Fm_{12}\underline{\Xm}''_{2} + \underline{\Zm}_{1}\\
&\stackrel{(a)}{=}& \underbrace{\Fm_{11}\Vm_{1}}_{\triangleq \Hm_{1}} [
             \begin{array}{cc}
               \Id_{M } & \Cm_{12} \\
             \end{array}]\left[
                                \begin{array}{c}
                                  \Am_{1}\underline{\Xm}'_{1} \\
                                  \Am_{2}\underline{\Xm}'_{2} \\
                                \end{array}
                              \right]
                                  + \underline{\Zm}_{1}\\
&=&\Hm_{1}\Cm_{1}\left[
                                \begin{array}{c}
                                  \underline{\Xm}'_{1} \\
                                  \underline{\Xm}'_{2} \\
                                \end{array}
                              \right]
                                  + \underline{\Zm}_{1} \label{eq:alignedsig1}
\end{eqnarray} where $(a)$ follows from the fact that the precoding vectors satisfy the alignment conditions in (\ref{cond:ALI}) and $\Cm_{1}=[
             \begin{array}{cc}
               \Am_{1} & \Cm_{12}\Am_{2} \\
             \end{array}]$.
\item Similarly, receiver 2 observes the aligned signals:
\begin{eqnarray}
\underline{\Ym}_{2}
&=& \Fm_{21}\underline{\Xm}_{1} + \Fm_{22}\underline{\Xm}_{2} + \underline{\Zm}_{2}\\
&=& \underbrace{\Fm_{21}\Vm_{1}}_{\triangleq \Hm_{2}}[
             \begin{array}{cc}
               \Id_{M} & \Cm_{22} \\
             \end{array}]\left[
                                \begin{array}{c}
                                  \Am_{1}\underline{\Xm}'_{1} \\
                                  \Am_{2}\underline{\Xm}'_{2} \\
                                \end{array}
                              \right]
                                  + \underline{\Zm}_{2}\\
&=& \Hm_{2}\Cm_{2}\left[
                                \begin{array}{c}
                                  \underline{\Xm}'_{1} \\
                                  \underline{\Xm}'_{2} \\
                                \end{array}
                              \right]+\underline{\Zm}_{2}\label{eq:alignedsig2}
\end{eqnarray}where  $\Cm_{2}=[
             \begin{array}{cc}
               \Am_{1} & \Cm_{22}\Am_{2} \\
             \end{array}]$.
\end{itemize}
Notice that the channel matrices in (\ref{eq:alignedsig1}) and (\ref{eq:alignedsig2}) follow the particular form in (\ref{2user-GMAC}). Following the CoF framework in (\ref{cofrate}) and (\ref{cofeq}), if $R\leq R_{\rm{comp}}(\Hm_{k}\Cm_{k},\Bm_{k},\SNR)$, receiver $k$ can decode
the $M$ linear combinations with full-rank integer coefficients matrix $\Bm_{k}$:
\begin{eqnarray}
\underline{\Um}_{k} &=& [\Bm_{k}^{\herm}]_{q}[\Cm_{k}]_{q}\left[
                                          \begin{array}{c}
                                            \underline{\Wm}'_{\mbox{\tiny{T}}_{1}} \\
                                            \underline{\Wm}'_{\mbox{\tiny{T}}_{2}}\\
                                          \end{array}
                                        \right]\\
&=& [\Bm_{k}^{\herm}]_{q}\left[
     \begin{array}{cc}
       [\Am_{1}]_{q} & [\Cm_{k2}]_{q}[\Am_{2}]_{q} \\
     \end{array}
   \right]\left[
                                          \begin{array}{c}
                                            \underline{\Wm}'_{\mbox{\tiny{T}}_{1}} \\
                                            \underline{\Wm}'_{\mbox{\tiny{T}}_{2}} \\
                                          \end{array}
                                        \right]\\
&\stackrel{(a)}{=}&
   [\Bm_{k}^{\herm}]_{q}\left[
     \begin{array}{cc}
       \Id_{M} & [\Cm_{k2}]_{q} \\
     \end{array}
   \right]\left[
                                          \begin{array}{c}
                                            \underline{\Wm}_{\mbox{\tiny{T}}_{1}} \\
                                            \underline{\Wm}_{\mbox{\tiny{T}}_{2}} \\
                                          \end{array}
                                        \right]\label{eq:estMes}
\end{eqnarray} where $(a)$ is due to the precoding over $\FF_{q}$ in (\ref{eq:precoding}).
Let $\hat{\underline{\Wm}}_{\mbox{\tiny{T}}_{1}} = [\Bm_{1}^{\herm}]_{q}^{-1}\underline{\Um}_{1}$ and $\hat{\underline{\Wm}}_{\mbox{\tiny{T}}_{2}}$ denote the first $M-1$ rows of $[\Bm_{2}^{\herm}]_{q}^{-1}\underline{\Um}_{2}$. The mapping between
$\{\underline{\Wm}_{\mbox{\tiny{T}}_{1}}, \underline{\Wm}_{\mbox{\tiny{T}}_{2}}\}$
and $\{\hat{\underline{\Wm}}_{\mbox{\tiny{T}}_{1}} , \hat{\underline{\Wm}}_{\mbox{\tiny{T}}_{2}} \}$  defines a {\em deterministic} finite-field IC given by:
\begin{equation}
\left[
  \begin{array}{c}
    \hat{\underline{\Wm}}_{\mbox{\tiny{T}}_{1}} \\
   \hat{\underline{\Wm}}_{\mbox{\tiny{T}}_{2}} \\
  \end{array}
\right]=\Qm_{{\rm sys}}\left[
  \begin{array}{c}
    \underline{\Wm}_{\mbox{\tiny{T}}_{1}} \\
    \underline{\Wm}_{\mbox{\tiny{T}}_{2}} \\
  \end{array}
\right]\label{eq:fmodel}
\end{equation}
where the {\em system matrix} is defined by
\begin{equation}
\Qm_{{\rm sys}}\eqdef\left[
      \begin{array}{cc}
        \Id_{M} & \Qm_{12} \\
        \Qm_{21} & \Id_{M-1} \\
      \end{array}
    \right]\label{def:Q}
\end{equation}
and where
\begin{equation}
\Qm_{12}  = \left[
                                                                \begin{array}{c}
                                                                  \mbox{0}_{1 \times (M-1)} \\
                                                                  \Id_{M-1} \\
                                                                \end{array}
                                                              \right], \;\;\;\;\;  \Qm_{21} = [
                                 \begin{array}{cc}
                                   \Id_{M-1} & \mbox{0}_{(M-1) \times 1} \\
                                 \end{array}].   \label{Q12-Q21}
\end{equation}
Notice that the system matrix is fixed and independent of the channel matrices, since it is determined only by the alignment conditions.

\subsection{Linear Precoding over deterministic networks}\label{subsec:LP}

In this section we determine linear precoding schemes to eliminate the interferences in the finite-field domain.
Recall that transmitter 2 sends only $M-1$ messages in order to use CIA. Accordingly,
$\Qm_{\rm{sys}}$ in (\ref{eq:fmodel}) has dimension $(2M-1) \times (2M-1)$.

\subsubsection{Network-Coded ICC}

In this model, the source can deliver linear combinations of information messages with coefficients
$\Qm_{{\rm sys}}^{-1}$:
\begin{equation}
\left[
  \begin{array}{c}
    \underline{\Wm}_{\mbox{\tiny{T}}_{1}} \\
     \underline{\Wm}_{\mbox{\tiny{T}}_{2}}\\
  \end{array}
\right]=\Qm_{{\rm sys}}^{-1}\left[
  \begin{array}{c}
    \underline{\Wm}_{1} \\
     \underline{\Wm}_{2}\\
  \end{array}
\right].
\end{equation} We have:

\begin{lemma}\label{lem:Qfull}
The system matrix $\Qm_{{\rm sys}}$ defined in (\ref{def:Q}) is full-rank over $\FF_{q}$.
\end{lemma}
\begin{IEEEproof}
The determinant of $\Qm_{{\rm sys}}$ is  given by
\begin{eqnarray}
\det(\Qm_{{\rm sys}}) &=& \det(\Id_{M})\det(\Id_{M-1}\oplus(-\Qm_{21}\Qm_{12}))\\
&=&\det(\Id_{M-1}\oplus(-\Qm_{21}\Qm_{12})) = 1,
\end{eqnarray}
since $\Id_{M-1}\oplus(-\Qm_{21}\Qm_{12})$ is a lower triangular matrix with unit diagonal elements.
\end{IEEEproof}
Such precoding yields immediately $\hat{\underline{\Wm}}_{\mbox{\tiny{T}}_{k}} = \underline{\Wm}_k$ for $k = 1,2$.
This proves Theorem~\ref{thm:122} for the Network-Coded ICC.

\subsubsection{$2\times 2\times 2$ IC}

We use the CoF framework based on CIA illustrated in Section \ref{subsec:CoF-SA} in order to turn
each hop (i.e., a two-user MIMO IC) into a deterministic finite-field IC defined
by $\Qm_{{\rm sys}}$ in (\ref{def:Q}).
At the two sources, no precoding is used such that $\underline{\Wm}_{\mbox{\tiny{T}}_{k}} = \underline{\Wm}_{k}$, for $k=1,2$.
Hence, the deterministic finite-field IC corresponding to the first-hop of the $2\times2\times2$ IC network
has outputs $\hat{\underline{\Wm}}_{\mbox{\tiny{T}}_{1}}, \hat{\underline{\Wm}}_{\mbox{\tiny{T}}_{2}}$ related to
$\underline{\Wm}_{1}$ and $\underline{\Wm}_{2}$ by (\ref{eq:fmodel}).

Relays 1 and 2 perform precoding of the decoded linear combination messages such as
$\underline{\Wm}_{\mbox{\tiny{T}}_{3}} = \Mm_{1} \hat{\underline{\Wm}}_{\mbox{\tiny{T}}_{1}}$
and $\underline{\Wm}_{\mbox{\tiny{T}}_{4}} = \Mm_{2}\hat{\underline{\Wm}}_{\mbox{\tiny{T}}_{2}}$,
where the precoding matrices $\Mm_{1}$ and $\Mm_{2}$ are defined in Lemma~\ref{lem:pre}. Operating in a similar way as for the first hop,
the second hop deterministic finite-field IC is given by
\begin{eqnarray}
\left[
                                             \begin{array}{c}
                                              \hat{\underline{\Wm}}_{\mbox{\tiny{T}}_{3}} \\
                                               \hat{\underline{\Wm}}_{\mbox{\tiny{T}}_{4}} \\
                                             \end{array}
                                           \right]&=&\Qm_{{\rm sys}}\left[
                                             \begin{array}{c}
                                               \underline{\Wm}_{\mbox{\tiny{T}}_{3}} \\
                                               \underline{\Wm}_{\mbox{\tiny{T}}_{4}} \\
                                             \end{array}
                                           \right]\label{eq:2hop}.
\end{eqnarray}
Concatenating the two hops, the end-to-end finite-field deterministic network is described by
\begin{equation}
\left[
                                             \begin{array}{c}
                                              \hat{\underline{\Wm}}_{\mbox{\tiny{T}}_{3}} \\
                                               \hat{\underline{\Wm}}_{\mbox{\tiny{T}}_{4}} \\
                                             \end{array}
                                           \right]=\Qm_{{\rm sys}}\left[
                                      \begin{array}{cc}
                                        \Mm_{1} & 0 \\
                                        0 & \Mm_{2} \\
                                      \end{array}
                                    \right]\Qm_{{\rm sys}}\left[
                                             \begin{array}{c}
                                               \underline{\Wm}_{1} \\
                                               \underline{\Wm}_{2} \\
                                             \end{array}
                                           \right].
\end{equation} Lemma \ref{lem:pre} shows that the decoded linear combinations are equal to its desired messages at destination 1 and
are equal to the messages with a change of sign (multiplication by $-1$ in the finite-field) at destination 2 (see Fig.~\ref{precoding-222}). This proved Theorem~\ref{thm:222}.

\begin{figure}
\centerline{\includegraphics[width=14cm]{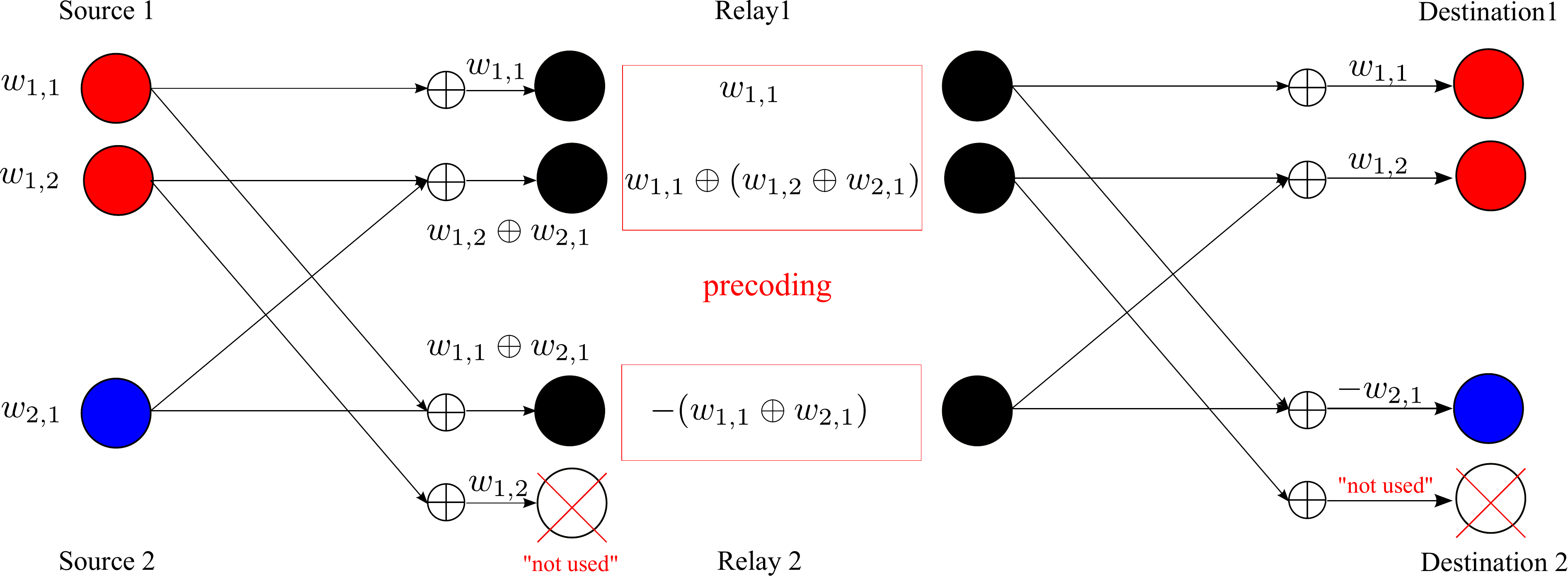}}
\caption{A {\em deterministic} noiseless $2\times 2 \times 2$ finite-field IC.}
\label{precoding-222}
\end{figure}

\begin{lemma}\label{lem:pre}
Choosing precoding matrices $\Mm_{1}$ and $\Mm_{2}$ as
\begin{eqnarray}
\Mm_{1}&=&(\Id_{M}\oplus(-\Qm_{12}\Qm_{21}))^{-1}\label{eq:M1}\\
\Mm_{2}&=& -(\Id_{M-1}\oplus(-\Qm_{21}\Qm_{12}))^{-1}\label{eq:M2}
\end{eqnarray}
the end-to-end system matrix becomes a diagonal matrix:
\begin{eqnarray}
\Qm_{{\rm sys}}\left[
                                      \begin{array}{cc}
                                        \Mm_{1} & 0 \\
                                        0 & \Mm_{2} \\
                                      \end{array}
                                    \right]\Qm_{{\rm sys}}&=& \left[
                                      \begin{array}{cc}
                                         \Mm_{1}\oplus\Qm_{12}\Mm_{2}\Qm_{21}&  \Mm_{1}\Qm_{12}\oplus\Qm_{12}\Mm_{2} \\
                                        \Qm_{21}\Mm_{1}\oplus \Mm_{2}\Qm_{21}& \Qm_{21}\Mm_{1}\Qm_{12}\oplus\Mm_{2}\\
                                      \end{array}
                                    \right]\\
                                     &=& \left[
                                      \begin{array}{cc}
                                        \Id_{M} & 0 \\
                                        0 & -\Id_{M-1} \\
                                      \end{array}
                                    \right].
\end{eqnarray}
\end{lemma}
\begin{IEEEproof}
From the Matrix Inversion Lemma \cite[Thm 18.2.8]{Harville}, we can rewrite $\Mm_{1}$ and $\Mm_{2}$ as
\begin{eqnarray}
\Mm_{1} &=&  \Id_{M} \oplus \Qm_{12}(\Id_{M-1}\oplus(-\Qm_{21}\Qm_{12}))^{-1}\Qm_{21}\label{eq:M1'}\\
\Mm_{2} &=& -(\Id_{M-1}\oplus\Qm_{21}(\Id_{M }\oplus(-\Qm_{12}\Qm_{21}))^{-1}\Qm_{12})\label{eq:M2'}.
\end{eqnarray}
\begin{itemize}
\item Canceling the interferences:
\begin{eqnarray*}
 \Mm_{1}\Qm_{12}\oplus\Qm_{12}\Mm_{2}
&=& \Qm_{12} \oplus \Qm_{12}(\Id_{M-1}\oplus( - \Qm_{21}\Qm_{12}))^{-1}\Qm_{21}\Qm_{12}\\
&&\oplus \Qm_{12}(-(\Id_{M-1} \oplus(- \Qm_{21}\Qm_{12}))^{-1})\\
&=& \Qm_{12} \oplus( - \Qm_{12})(\Id_{M-1}\oplus(-\Qm_{21}\Qm_{12}))^{-1}((-\Qm_{21}\Qm_{12}) \oplus \Id_{M-1})\\
&=& \bold{0}_{M\times (M-1)}\\
\Qm_{21}\Mm_{1}\oplus \Mm_{2}\Qm_{21}  &=& \Qm_{21} \oplus \Qm_{21}\Qm_{12}(\Id_{M-1}\oplus(-\Qm_{21}\Qm_{12}))^{-1}\Qm_{21} \\
&&\oplus(- (\Id_{(M-1) \times (M-1)}\oplus(-\Qm_{21}\Qm_{12}))^{-1})\Qm_{21}\\
&=& \Qm_{21}\oplus(\Id_{M-1}\oplus(-\Qm_{21}\Qm_{12}))(\Id_{M-1}\oplus(-\Qm_{21}\Qm_{12}))^{-1}(-\Qm_{21})\\
&=& \bold{0}_{(M-1)\times M},
\end{eqnarray*} where we used $\Mm_{1}$ in (\ref{eq:M1'}) and $\Mm_{2}$ in (\ref{eq:M2}).
\item Preserving the desired signals:
\begin{eqnarray*}
\Mm_{1}\oplus\Qm_{12}\Mm_{2}\Qm_{21} &=&  \Id_{M } \oplus \Qm_{12}(\Id_{M-1}\oplus(-\Qm_{21}\Qm_{12}))^{-1}\Qm_{21}\\
&& \oplus(- \Qm_{12}(\Id_{M-1}\oplus( - \Qm_{21}\Qm_{12}))^{-1}\Qm_{21})\\
&=& \Id_{M},
\end{eqnarray*}where we used $\Mm_{1}$ in (\ref{eq:M1'}) and $\Mm_{2}$ in (\ref{eq:M2}).
\begin{eqnarray*}
 \Qm_{21}\Mm_{1}\Qm_{12}\oplus\Mm_{2} &=& \Qm_{21}(\Id_{M}\oplus(-\Qm_{12}\Qm_{21}))^{-1}\Qm_{12}\oplus( - \Id_{M-1})\\
&& \oplus(- \Qm_{21}(\Id_{M}\oplus(-\Qm_{12}\Qm_{21}))^{-1}\Qm_{12})\\
&=& -\Id_{M-1},
\end{eqnarray*}where we used $\Mm_{1}$ in (\ref{eq:M1}) and $\Mm_{2}$ in (\ref{eq:M2'}).
This completes the proof.
\end{itemize}
\end{IEEEproof}

\subsubsection{Network-Coded CIC}

In this case we assume that transmitter 1 knows  both  messages $\underline{\Wm}_{1}$ and $\underline{\Wm}_{2}$, and transmitter 2 only
knows $M-1$ linear combinations $\Sm_{1}\underline{\Wm}'_{1} \oplus \Sm_{2}\underline{\Wm}_{2}$, where
$\Sm_{1}, \Sm_{2} \in \FF_{q}^{(M-1) \times (M-1)}$ are full-rank matrices and where
$\underline{\Wm}'_{1} = \Qm_{21} \underline{\Wm}_1$ contains the first $M-1$ rows of $\underline{\Wm}_1$ (see the definition of $\Qm_{21}$ in (\ref{Q12-Q21})).
In fact, we may assume that transmitter 2 also knows the interference-free message $\underline{\wv}_{1,M}$ (the last row of $\underline{\Wm}_1$)
but this is not used in our scheme. Transmitters 1 and 2 perform the precoding (over $\FF_{q}$) in the following way:
\begin{eqnarray}
\underline{\Wm}_{\mbox{\tiny{T}}_{1}} &=& \Mm_{1}(\underline{\Wm}_{1}\oplus\Qm_{12}\Sm_{1}^{-1}\Sm_{2}\underline{\Wm}_{2})\\
\underline{\Wm}_{\mbox{\tiny{T}}_{2}} &=& \Mm_{2}\Sm_{1}^{-1}(\Sm_{1}\underline{\Wm}'_{1} \oplus \Sm_{2}\underline{\Wm}_{2})
\end{eqnarray}where $\Mm_{1}$ and $\Mm_{2}$ are defined in Lemma~\ref{lem:pre}.
From (\ref{eq:fmodel}), we have:
\begin{eqnarray}
\left[
  \begin{array}{c}
    \hat{\underline{\Wm}}_{\mbox{\tiny{T}}_{1}} \\
    \hat{\underline{\Wm}}_{\mbox{\tiny{T}}_{2}} \\
  \end{array}
\right]&=&\Qm_{{\rm sys}}\left[
  \begin{array}{c}
    \underline{\Wm}_{\mbox{\tiny{T}}_{1}} \\
    \underline{\Wm}_{\mbox{\tiny{T}}_{2}} \\
  \end{array}
\right]\\
& = &\Qm_{{\rm sys}}\left[
             \begin{array}{cc}
               \Mm_{1} & 0 \\
               0 & \Mm_{2} \\
             \end{array}
           \right]\left[
                    \begin{array}{cc}
                      \Id_{M} & \Qm_{12}\Sm_{1}^{-1}\Sm_{2} \\
                      \Qm_{21} & \Sm_{1}^{-1}\Sm_{2} \\
                    \end{array}
                  \right]
           \left[
  \begin{array}{c}
    \underline{\Wm}_{1} \\
    \underline{\Wm}_{2} \\
  \end{array}
\right]\\
&=&\Qm_{{\rm sys}}\left[
             \begin{array}{cc}
               \Mm_{1} & 0 \\
               0 & \Mm_{2} \\
             \end{array}
           \right]\Qm_{{\rm sys}}\left[
                       \begin{array}{cc}
                         \Id_{M} & 0 \\
                         0 & \Sm_{1}^{-1}\Sm_{2} \\
                       \end{array}
                     \right]\left[
  \begin{array}{c}
    \underline{\Wm}_{1} \\
    \underline{\Wm}_{2} \\
  \end{array}
\right]\\
&\stackrel{(a)}{=}& \left[
                       \begin{array}{cc}
                         \Id_{M} & 0 \\
                         0 & -\Sm_{1}^{-1}\Sm_{2} \\
                       \end{array}
                     \right]\left[
  \begin{array}{c}
    \underline{\Wm}_{1} \\
    \underline{\Wm}_{2} \\
  \end{array}
\right]
\end{eqnarray}
where $(a)$ follows from Lemma~\ref{lem:pre}.
This shows that the decoded linear combinations are equal to its desired messages at receiver 1 and are equal to the
messages with multiplication by full-rank matrix $(-\Sm_{1}^{-1}\Sm_{2})^{-1}$ (in the finite-field domain) at receiver 2.
Based on this, Theorem~\ref{thm:122} is proved for the Network-Coded CIC.

\section{Improving the sum rates using successive cancellation}\label{sec:SIC}

In this section we improve the sum rate in Theorem~\ref{thm:122} by using CoF with successive cancellation.
We focus on Network-Coded ICC to explain the proposed scheme.
As shown before, precoding of information messages (over $\FF_{q}$) can eliminate interference
from the other transmitter so that each receiver observes full-rank integer linear combinations of its own
intended lattice codewords in the corresponding MIMO modulo $\Lambda$ channel.
Once the network is reduced to two decoupled MIMO modulo $\Lambda$ channels,
each receiver can perform {\em successive cancellation} following the idea first proposed in \cite{Ordentlich}.
In this way, each message can be recovered reliably at rate equal to the computation rate of the corresponding
equation, without being constrained by the equal rate requirement (minimum of the computation rates of all equations).

In order to achieve the different coding rates while preserving the lattice $\ZZ[j]$-module structure, we use a {\em family} of nested lattice codes $\Lambda \subseteq \Lambda_{2M-1} \subseteq \cdots \subseteq \Lambda_{1}$, obtained by a nested construction A as described in \cite[Sect. IV.B]{Nazer}. In particular, we let $\Lambda_{\ell} = p^{-1}g(\Cc_{\ell})\Tm + \Lambda$ with $\Lambda = \ZZ^{n}[j]\Tm$ and with $\Cc_{\ell}$ denoting the linear code over $\FF_{q}$ generated by the first $r_{\ell}$ rows of a common generator matrix $\Gm$, with $r_{2M-1} \leq \cdots \leq r_{1}$. The corresponding nested lattice codes are given by $\Lc_{\ell} = \Lambda_{\ell} \cap \Vc_{\Lambda}$ for $\ell=1,\ldots,2M-1$. Let $\underline{\wv}_{k,\ell} \in \FF_{q}^{r_{1}}$ be the zero-padded message to the common length $r_{1}$.

Encoding follows the same procedure outlined in Section~\ref{sec:twoMIMO}. Namely, we let
\begin{equation}
\left[
\begin{array}{ccc}
\underline{\Wm}_{\mbox{\tiny{T}}_{1}}   \\
\underline{\Wm}_{\mbox{\tiny{T}}_{2}}
\end{array}
\right]=\Qm_{{\rm sys}}^{-1}\left[
\begin{array}{ccc}
\underline{\Wm}_{1}   \\
\underline{\Wm}_{2}
\end{array}
\right] \label{eq:prec1-sec5}
\end{equation}
in order to eliminate interference. Recall that each transmitter $k$ precodes its messages over $\FF_{q}$ as in (\ref{eq:precoding})
where the integer matrix $\Am_k$ is used for IFB with the purpose of minimizing the power penalty (see later).
Then, the precoded messages are encoded using the densest lattice code
$\Lc_{1}$ as $\underline{\tv}'_{\ell,k} = f(\underline{\wv}'_{\mbox{\tiny{T}}_{k,\ell}})$.
Finally, the channel input sequences are given by the rows of  $\underline{\Xm}''_{k}$ defined in (\ref{suca1}).
%
Let $\underline{\tv}_{k,\ell} = f(\underline{\wv}_{k,\ell})$ denote the lattice codeword corresponding to information message $\underline{\wv}_{k,\ell}$.
Using lattice linearity, we can express the precoding in the complex (lattice) domain as:
\begin{equation}
\left[
  \begin{array}{c}
   \underline{\Tm}'_{1} \\
    \underline{\Tm}'_{2} \\
  \end{array}
\right]=\left[
                     \begin{array}{cc}
                       \Am_{1}^{-1} & 0 \\
                       0 & \Am_{2}^{-1}  \\
                     \end{array}
                   \right]g(\Qm_{{\rm sys}}^{-1})
\left[
  \begin{array}{c}
    \underline{\Tm}_{1} \\
    \underline{\Tm}_{2} \\
  \end{array}
\right] \mod \Lambda. \label{eq:relation}
\end{equation}
From  (\ref{eq:alignedsig1}) and (\ref{eq:alignedsig2}), each receiver $k$ observes the integer aligned signals:
\begin{eqnarray}
\underline{\Ym}_{k} &=&\Hm_{k}\Cm_{k}\left[
                                            \begin{array}{c}
                                              \underline{\Xm}'_{1} \\
                                              \underline{\Xm}'_{2} \\
                                            \end{array}
                                          \right] + \underline{\Zm}_{k}\\
&=&\Hm_{k}\left[
                               \begin{array}{cc}
                                 \Id_{M } & \Cm_{k2} \\
                               \end{array}
                             \right]\left[
                                     \begin{array}{cc}
                                       \Am_{1} & 0 \\
                                       0 & \Am_{2}\\
                                     \end{array}
                                   \right]\left[
                                            \begin{array}{c}
                                              \underline{\Xm}'_{1} \\
                                              \underline{\Xm}'_{2} \\
                                            \end{array}
                                          \right] + \underline{\Zm}_{k}.
\end{eqnarray}
The modulo $\Lambda$ vector channel after applying the CoF receiver mapping
(\ref{eq:cof}) with integer coefficients matrix $\Bm_{k}$ seen at each receiver $k = 1,2$ is given as follows:
\begin{itemize}
\item  At receiver 1 we have:
\begin{eqnarray}
\hat{\underline{\Ym}}_{1}&=&\left[\Bm_{1}^{\herm}\left[
                               \begin{array}{cc}
                                 \Id_{M} & \Cm_{12} \\
                               \end{array}
                             \right]
\left[
                                     \begin{array}{cc}
                                       \Am_{1} & 0 \\
                                       0 & \Am_{2}\\
                                     \end{array}
                                   \right]\left[
  \begin{array}{c}
   \underline{\Tm}'_{1} \\
    \underline{\Tm}'_{2} \\
  \end{array}
\right]
+\underline{\Zm}_{\mbox{\tiny{eff}}}(\Hm_{1}\Cm_{1},\Bm_{1})\right] \mod \Lambda \nonumber \\
&\stackrel{(a)}{=}& \left[\Bm_{1}^{\herm}\left[
                               \begin{array}{cc}
                                 \Id_{M} & \Cm_{12} \\
                               \end{array}
                             \right]g(\Qm_{{\rm sys}}^{-1})\left[
  \begin{array}{c}
    \underline{\Tm}_{1} \\
    \underline{\Tm}_{2} \\
  \end{array}
\right]+\underline{\Zm}_{\mbox{\tiny{eff}}}(\Hm_{1}\Cm_{1},\Bm_{1})\right] \mod \Lambda\\
&\stackrel{(b)}{=}&\left[\Bm_{1}^{\herm}\underline{\Tm}_{1}+\underline{\Zm}_{\mbox{\tiny{eff}}}(\Hm_{1}\Cm_{1},\Bm_{1})\right] \mod \Lambda \label{eq:MIMO1}
\end{eqnarray}
where $(a)$ follows the (\ref{eq:relation}), $(b)$ is due to the fact that
\begin{equation}
\left[\left[
                               \begin{array}{cc}
                                 \Id_{M} & \Cm_{12} \\
                               \end{array}
                             \right]\right]_{q}\Qm_{{\rm sys}}^{-1} = \left[
                         \begin{array}{cc}
                           \Id_{M} & \textbf{0}_{M \times M-1} \\
                         \end{array}
                       \right].
\end{equation}
and where $\underline{\Zm}_{\mbox{\tiny{eff}}}(\Hm_{1}\Cm_{1},\Bm_{1})$ denotes the $M \times n$ matrix of effective noises
with rows $\underline{\zv}_{\mbox{\tiny{eff}}}(\Hm_1\Cm_1,\bv_{1,\ell},\alphav_{1,\ell})$, and the projection
vector $\alphav_{1,\ell}$ is determined as a function of $\Hm_1\Cm_1,\bv_{1,\ell}$ as said in Section \ref{subsec:CoF}.
\item Similarly, at receiver 2 we have:
\begin{eqnarray}
\hat{\underline{\Ym}}_{2}&=&\left[\Bm_{2}^{\herm}\left[
                               \begin{array}{cc}
                                 \Id_{M} & \Cm_{22} \\
                               \end{array}
                             \right]\left[
                                     \begin{array}{cc}
                                       \Am_{1} & 0 \\
                                       0 & \Am_{2}\\
                                     \end{array}
                                   \right]\left[
  \begin{array}{c}
    \underline{\Tm}'_{1} \\
    \underline{\Tm}'_{2} \\
  \end{array}
\right]
+\underline{\Zm}_{\mbox{\tiny{eff}}}(\Hm_{2}\Cm_{2},\Bm_{2})\right] \mod \Lambda \nonumber \\
&=& \left[\Bm_{2}^{\herm}\left[
                               \begin{array}{cc}
                                 \Id_{M} & \Cm_{22} \\
                               \end{array}
                             \right]g(\Qm_{{\rm sys}}^{-1})\left[
  \begin{array}{c}
    \underline{\Tm}_{1} \\
    \underline{\Tm}_{2} \\
  \end{array}
\right]+\underline{\Zm}_{\mbox{\tiny{eff}}}(\Hm_{2}\Cm_{2},\Bm_{2})\right] \mod \Lambda\\
&\stackrel{(a)}{=}&\left[\Bm_{2}^{\herm}\left[
                  \begin{array}{c}
                    \underline{\Tm}_{2} \\
                    \underline{\tv} \\
                  \end{array}
                \right]+\underline{\Zm}_{\mbox{\tiny{eff}}}(\Hm_{2}\Cm_{2},\Bm_{2})\right] \mod \Lambda,   \label{wakka}
\end{eqnarray}
where $\underline{\tv}$ denotes some linear combination of lattice codewords,
irrelevant for receiver 2, $(a)$ follows from the fact that
\begin{equation}
\left[\left[
                               \begin{array}{cc}
                                 \Id_{M} & \Cm_{22} \\
                               \end{array}
                             \right]\right]_{q}\Qm_{{\rm sys}}^{-1} = \left[
                          \begin{array}{cc}
                            \textbf{0}_{M-1\times M-1} & \Id_{M-1} \\
                             \star
                             \\
                          \end{array}
                        \right]
\end{equation} and where $\star$ denotes some non-zero vector in $\FF_{q}^{1 \times (2M-1)}$. In (\ref{wakka}),
$\underline{\Zm}_{\mbox{\tiny{eff}}}(\Hm_{2}\Cm_{2},\Bm_{2})$ is defined similarly to
$\underline{\Zm}_{\mbox{\tiny{eff}}}(\Hm_{1}\Cm_{1},\Bm_{1})$.
Receiver 2 can recover its $M-1$ messages as long as it has $M-1$ full-rank linear combinations of its own messages.
In order to remove the unintended messages collected in $\underline{\tv}$, we choose $\Bm_{2}$ in the form:
\begin{equation}
\Bm_{2}^{\herm} = \left[
            \begin{array}{cc}
              \tilde{\Bm}_{2}^{\herm} & \zerov
            \end{array}
          \right]
\end{equation}
where $\tilde{\Bm}_{2} \in \ZZ[j]^{(M-1) \times (M-1)}$ is full-rank.
Then, the first $M-1$ observations of receiver 2 is given by
\begin{equation}
\hat{\underline{\Ym}}'_{2} = \left[\tilde{\Bm}_{2}^{\herm} \underline{\Tm}_{2}+\underline{\Zm}_{\mbox{\tiny{eff}}}(\Hm_{2}\Cm_{2},\Bm_{2})\right] \mod \Lambda. \label{eq:MIMO2}
\end{equation}
\end{itemize}
From (\ref{eq:MIMO1}) and (\ref{eq:MIMO2}), we have that each receiver obtains a full-rank interference-free MIMO integer valued modulo $\Lambda$ channel with effective additive noise. At this point,  each receiver can perform {\em successive cancellation} \cite{Ordentlich}, thus relaxing the minimum common
computation rate constraint. Focusing on receiver 1, we illustrate the successive cancellation procedure with given integer matrix
$\Bm_{1}$ and computation rates $\{\log^+(\SNR/\sigma^{2}_{\mbox{\tiny{eff}},1,\ell}) \; : \; \ell=1,\ldots,M\}$.
The same procedure can be straightforwardly applied to receiver 2, given the formal equivalence of (\ref{eq:MIMO1}) and (\ref{eq:MIMO2}).
Without loss of generality, assume that
\begin{equation}
\sigma^{2}_{\mbox{\tiny{eff}},1,1} \leq \cdots \leq \sigma^{2}_{\mbox{\tiny{eff}},1,M}.
\end{equation}
Letting $R_{1,\ell}$ denote the rate of $\ell$-th message of user 1, we have $R_{1,\ell} = r_{j_{1,\ell}}$ for some $j_{1,\ell} \in \{1,\ldots,2M-1\}$, i.e., the $\ell$-th message of user 1 is encoded using nested lattice codes $\Lc_{j_{1,\ell}}$. For the time being, we assume that
$R_{1,1} \geq R_{1,2} \geq \cdots \geq R_{1,M}$ (the ordering will be determined later on, according to column permutation of $\Bm_{1}$ that is required for successive cancellation). Receiver 1 can reliably decode $\underline{\sv}_{1} = [\bv_{1,1}^{\herm}\underline{\Tm}_{1}] \mod \Lambda$ as long as
\begin{equation}
R_{1,1} \leq \log^+\left (\frac{\SNR}{\sigma^{2}_{\mbox{\tiny{eff}},1,1}} \right ).
\end{equation}
Then, it proceeds to decode $\underline{\sv}_{2} = [\bv_{1,2}^{\herm}\underline{\Tm}_{1}] \mod \Lambda$. Using the previously decoded $\underline{\sv}_{1}$, it can perform the cancellation:
\begin{eqnarray}
[\hat{\underline{\yv}}_{2} + e_{21}\underline{\sv}_{1}] \mod \Lambda &=& [\underline{\sv}_{2} + e_{21}\underline{\sv}_{1} + \underline{\zv}_{\mbox{\tiny{eff}}}(\Hm_{1}\Cm_{1},\bv_{1,2}, \alpha_{1,2})] \mod \Lambda\\
&=&[\tilde{\underline{\sv}}_{2} + \underline{\zv}_{\mbox{\tiny{eff}}}(\Hm_{1}\Cm_{1},\bv_{1,2}, \alpha_{1,2})] \mod \Lambda.
\end{eqnarray} Here, $e_{21} \in \ZZ[j]$ is chosen so that $[(\bv_{1,2}(1)+e_{21}\bv_{1,1}(1))] \mod p\ZZ[j] = 0$ where $\bv(j)$ denotes the $j$-th element of vector $\bv$. Then $\tilde{\underline{\sv}}_{2}$ does not include $\underline{\tv}_{1}$ and hence  receiver 1 can reliably decode $\tilde{\underline{\sv}}_{2}$ as long as
\begin{equation}
R_{1,2} \leq \log^+ \left (\frac{\SNR}{\sigma^{2}_{\mbox{\tiny{eff}},1,2}} \right ).
\end{equation}
Now, receiver 1 can obtain $\underline{\sv}_{2}$ such as $\underline{\sv}_{2} = [\tilde{\underline{\sv}}_{2} - e_{21}\underline{\sv}_{1}] \mod \Lambda$.
Receiver 1 can decode the remaining linear combinations $\underline{\sv}_{\ell}$ for $\ell \geq 3$ in a similar manner. Namely, before decoding $\underline{\sv}_{\ell}$, receiver 1 adds $\left[\sum_{j=1}^{\ell-1} e_{\ell j}\underline{\sv}_{j}\right]\mod\Lambda$ (i.e., an integer valued linear combinations of previously decoded $\underline{\sv}_{j}$'s). Here the coefficients $e_{\ell j}$ are chosen so that the impact of $\underline{\tv}_{1},\ldots,\underline{\tv}_{\ell-1}$ is canceled out from $\underline{\sv}_{\ell}$. Assuming that such coefficients exist, receiver 1 can decode $\tilde{\underline{\sv}}_{\ell} = \left[\underline{\sv}_{\ell} + \sum_{j=1}^{\ell-1} e_{\ell j}\underline{\sv}_{j} \right]\mod\Lambda$
as long as $R_{1,\ell}$ is less than the corresponding computation rate of the $\ell$-th equation.

From \cite[Lemma 2]{Ordentlich}, such cancellation coefficients exist for at least one column permutation
vector $\pi_{1}$ of $\Bm_{1}$. Accordingly, all $M$ linear combinations can be decoded as long as
\begin{equation}
R_{1,\pi_{1}(\ell)} \leq \log^+\left (\frac{\SNR}{\sigma^{2}_{\mbox{\tiny{eff}},1,\ell}} \right )\;\;\;\ \mbox{for } \ell=1,\ldots,M.
\end{equation}
Therefore, the sum rate $\sum_{\ell=1}^{M} \log^+\left (\frac{\SNR}{\sigma^{2}_{\mbox{\tiny{eff}},1,\ell}} \right )$
is achievable.  Similarly, there exists at least one column permutation vector $\pi_{2}$ of $\Bm_{2}$ for which all
$M-1$ linear combinations at receiver 2 can be decoded as long as
\begin{equation}
R_{2,\pi_{2}(\ell)} \leq \log^+ \left (\frac{\SNR}{\sigma^{2}_{\mbox{\tiny{eff}},2,\ell}} \right )\;\;\;\ \mbox{for } \ell=1,\ldots,M-1,
\end{equation}
where we let $R_{2,\ell} = r_{j_{2,\ell}}$, for some index mapping $j_{2,\ell} \in \{1, \ldots, 2M-1\}$,
denote the rate of $\ell$-th message of user 2.  The exactly same procedure can be applied to the Network-Coded CIC.
The successive cancellation replaces the sum-rate formula in Theorem~\ref{thm:122} as follows:
\begin{corollary}\label{thm:122sum}
For the Network-Coded ICC and Network-Coded CIC, PCoF with CIA can achieve sum rate as
\begin{equation}
R_{\rm sum} =
\sum_{\ell=1}^M \log^+ \left ( \frac{\SNR}{\sigma^{2}_{\mbox{\tiny{eff}},1,\ell}} \right ) +
\sum_{\ell=1}^{M-1} \log^+ \left (\frac{\SNR}{\sigma^{2}_{\mbox{\tiny{eff}},2,\ell}} \right )
\end{equation}
for any full rank integer matrices $\Am_{1} \in \ZZ[j]^{M \times M}, \Am_{2} \in \ZZ[j]^{(M-1) \times (M-1)}$, $\Bm_{1} \in \ZZ[j]^{M \times M},
\tilde{\Bm}_{2} \in \ZZ[j]^{(M-1) \times (M-1)}$, and any alignment precoding matrices $\Vm_{k}$ to satisfy
the {\em alignment conditions} in (\ref{cond:ALI}), where
\begin{eqnarray}
\Bm_{2}^{\herm} &=& \left[
            \begin{array}{cc}
              \tilde{\Bm}_{2}^{\herm} & \zerov
            \end{array}
          \right]\\
\Hm_{k} &=& \Fm_{k1}\Vm_{1}, \Cm_{k}=\left[
\begin{array}{ccc}
\Am_{1}  & \Cm_{k2}\Am_{2}
\end{array}
\right], \;\;\; k=1,2\\
\SNR&=&\min_{k=1,2}\left\{\frac{P_{\rm{sum}}}{\trace{\left(\Vm_{k}\Am_{k}\Am_{k}^{\herm}\Vm_{k}^{\herm}\right)}}\right\},
\end{eqnarray} and where
\begin{eqnarray*}
\sigma^{2}_{\mbox{\tiny{eff}},k,\ell}
&=& \bv_{k,\ell}^{\herm}\Cm_k(\SNR^{-1}\Id+\Cm_k^{\herm}\Hm_k^{\herm}\Hm_k\Cm_k)^{-1}\Cm_k^{\herm}\bv_{k,\ell}, \;\;\; k = 1,2.
\end{eqnarray*}

\hfill \IEEEQED
\end{corollary}

\section{Optimization of achievable rates}\label{sec:finite}

In this section we optimize the integer matrices $\Am_{k}$ and $\Bm_{k}$ in Theorems \ref{thm:122}-\ref{thm:222} and Corollary~\ref{thm:122sum} by assuming that the precoding matrices $\Vm_{k}$ are given. The dimensions of $\Am_{k}$ and $\Bm_{k}$ can be either $M\times M$ or $(M-1)\times (M-1)$,
depending on $k$. Since this does not change the optimization problem, we will drop the index $k$ and just consider dimension $M$.
The power-penalty optimization with respect to $\Am$ takes on the form:
\begin{eqnarray}
\argmin && \trace\left(\Vm\Am\Am^{\herm}\Vm^{\herm}\right)=\sum_{\ell=1}^{M}\|\Vm\av_{\ell}\|^2\nonumber\\
\mbox{subject to}&& \mbox{$\Am$ is full rank over $\ZZ[j]$}\label{eq:opt1}
\end{eqnarray} where $\av_{\ell}$ denotes the $\ell$-th column of $\Am$.
Also, the minimization of the effective noise variance with respect to $\Bm$ takes on the form:
\begin{eqnarray}
\argmin && \max _{\ell}\left \{ \| \Lm \bv_{\ell} \|^2 \right \}\nonumber\\
\mbox{subject to}&& \mbox{$\Bm$ is full rank over $\ZZ[j]$} \label{eq:opt2}
\end{eqnarray}
where $\Lm$ denotes a square-root factor of $(\SNR^{-1}\Id+\Hm^{\herm}\Hm)^{-1}$,
$\Hm$ denotes an aligned channel matrix and $\bv_{\ell}$ denotes the $\ell$-th column of $\Bm$.

We notice that problem (\ref{eq:opt1}) (resp., (\ref{eq:opt2})) is equivalent to finding a reduced basis for the lattice generated by
$\Vm$ (resp., $\Lm$). In particular, the reduced basis takes on the form $\Vm\Um$ where $\Um$ is a unimodular matrix
over $\ZZ[j]$. Hence, choosing $\Am = \Um$ yields the minimum power-penalty subject to the full rank condition in (\ref{eq:opt1}).
In practice we used the (complex) LLL algorithm \cite{Napias}~\footnote{We can also use the HKZ and Minkowski lattice basis reduction algorithm (see \cite{Viterbo} for details).}, with refinement of the LLL reduced basis approximation by Phost or Schnorr-Euchner lattice search \cite{Damen}. We let $\uv_{\ell}$ denote the $\ell$-th column of $\Um$. Phost or Schnorr-Euchner enumeration generates all non-zero lattice
points in a sphere centered at the origin, with radius equal to $d=\max_{\ell}\|\Vm\uv_{\ell}\|^{2}$. This radius guarantees the existence
of solutions because of having the trivial solution $\{\uv_{1},\ldots,\uv_{M}\}$.
Define the set of integer vectors such that the corresponding lattice
points are in a sphere with radius $d$ by
\begin{equation}
(\Ac,\Wc)=\{(\av_{\ell},w_{\ell}): w_{\ell}=\|\Vm\av_{\ell}\|^2 \leq d\}.\label{def:set}
\end{equation}
The following lemma shows that the greedy algorithm (Algorithm 1) finds a solution (i.e., $M$ linearly independent integer vectors)
to the problem (\ref{eq:opt1}) (or (\ref{eq:opt2})).

\begin{lemma} For given $(\Ac,\Wc)$ defined in (\ref{def:set}), Algorithm 1 finds a solution to the following two problems:
\begin{eqnarray}
\min_{\Sc \subset \Ac} && \sum_{\ell \in \Sc} w_{\ell} \\ 
\mbox{subject to} && \{\av_{\ell}: \ell \in \Sc\} \mbox{ are linearly independent}\\
&& |\Sc| = M.
\end{eqnarray}
and
\begin{eqnarray}
\min_{\Sc \subset \Ac} && \max\{w_{\ell}:\ell \in \Sc\} \\
\mbox{subject to} && \{\av_{\ell}: \ell \in \Sc\} \mbox{ are linearly independent}\\
&& |\Sc| = M.
\end{eqnarray}
\end{lemma}
\begin{IEEEproof}
The first problem consists of the minimization of linear function subject to a matroid constraint, where the matroid $\Mc(\Omega,\Ic)$ is defined by the ground set $\Omega=[1:|\Ac|]$ and by the collection of independent sets $\Ic=\{\Sc \subset \Omega: \{\av_{\ell}: \ell \in \Sc\} \mbox{ are linearly independent}\}$. Rado and Edmonds \cite{Rado, Edmonds} proved that a greedy algorithm (Algorithm 1) finds an optimal solution. In case of the second problem, we provide a simple proof as follows. Suppose that the indices of elements in $\Ac$ are rearranged according to the increasing ordering of the weights $w_{\ell}$. The problem is then reduced to finding the minimum index $\ell^{\dag}$ such that $\{\av_{1},\ldots,\av_{\ell^{\dag}}\}$ includes the $M$ linearly independent vectors. This is precisely what Algorithm 1 does.
\end{IEEEproof}

\begin{algorithm}
\caption{The Greedy Algorithm}
\textbf{Input}: $(\Ac,\Phi)=\{(\av_{\ell},w_{\ell}): \av_{\ell} \in \ZZ[j]^{M\times 1}, w_{\ell} \in \ZZ_{+}\}$\newline
\textbf{Output}: $\Sc \subset \Ac$ with $|\Sc|=M$
\begin{enumerate}
\item Rearrange the indices of vectors in $\Ac$ such that $w_{1} \leq w_{2} \leq \cdots \leq w_{|\Ac|}$
\item Initially, $\ell=1$ and $\Sc=\phi$
\item If $\mbox{Rank}(\Sc\cup\{\ell\}) > \mbox{Rank}(\Sc)$ then $\Sc \leftarrow \Sc \cup \{\ell\}$
\item Set $\ell=\ell+1$
\item Repeat 3)-4) until $|\Sc|=M$
\end{enumerate}
\end{algorithm}

\subsection{Finite SNR Results}

\begin{figure}
\centerline{\includegraphics[width=14cm]{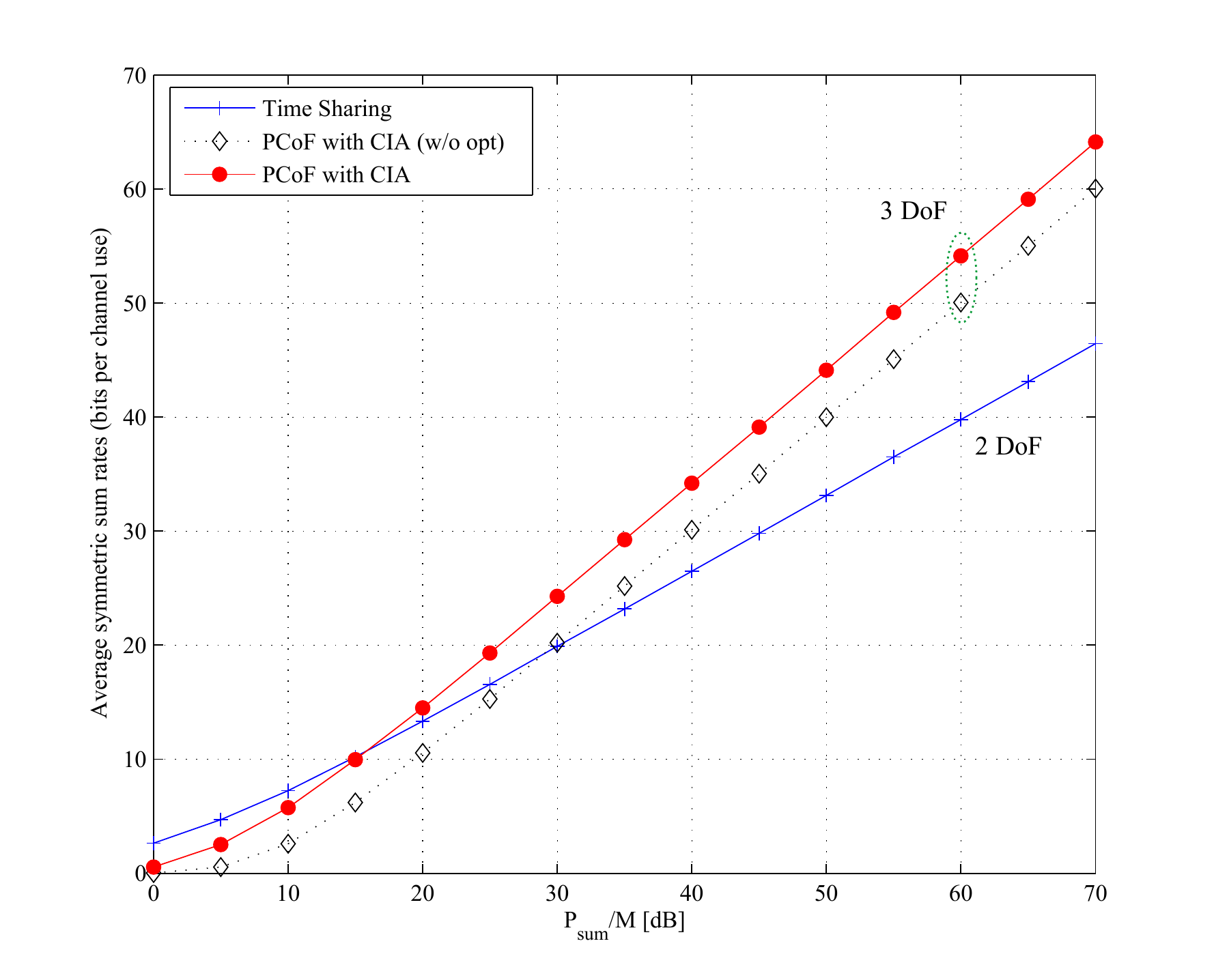}}
\caption{Performance comparison of  PCoF with CIA  and time-sharing with respect to ergodic symmetric sum rates for $2\times 2\times 2$ MIMO interference channel with $M=2$.}
\label{simulation}
\end{figure}
\begin{figure}
\centerline{\includegraphics[width=14cm]{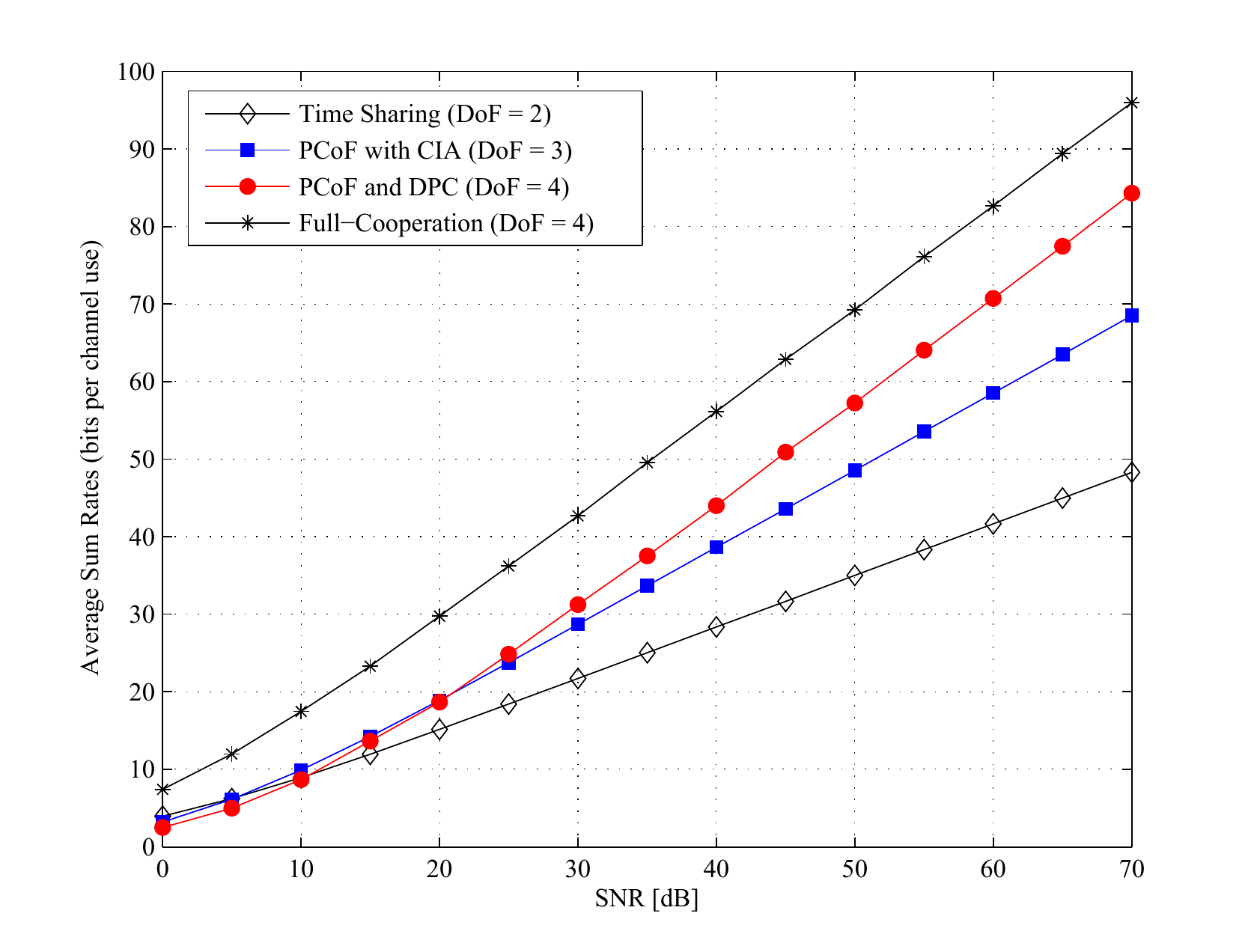}}
\caption{Performance comparison of PCoF with CIA, PCoF and DPC, and time-sharing with respect to ergodic symmetric sum rates for MIMO interference coordination channel with $M=2$.}
\label{simulation_sum}
\end{figure}
We evaluate the performance of PCoF with CIA in terms of its average achievable sum rates.
We computed the ergodic sum rates by Monte Carlo averaging with respect to the channel realizations  with i.i.d. elements $\sim \Cc\Nc(0,1)$. Recall that we consider a total power constraint equal to $P_{\rm{sum}}$ at each transmitter (both sources and relays). We first consider the symmetric sum rates in Theorem~\ref{thm:222} for $2\times 2\times 2$ Gaussian IC. For comparison, we considered the performance of {\em time-sharing} where IFR is used for each $M \times M$ MIMO IC. We used the IFR since it is known to almost achieve the performance of joint maximum likelihood receiver \cite{Zhan} and has a similar complexity with PCoF. In this case, an achievable symmetric sum rate of time-sharing is obtained as

\begin{equation}
R=\min\{R_{\rm{comp}}(\Fm_{11},\Bm_{1},\SNR), R_{\rm{comp}}(\Fm_{33},\Bm_{2},\SNR)\}
\end{equation}for any full-rank matrices $\Bm_{1}, \Bm_{2} \in \ZZ[j]^{M \times M}$, where $\SNR = P_{\rm{sum}}$.
Here, we used  $2P_{\rm{sum}}$ for power constraint since with time-sharing each transmitter is active on only half of the time slots.
For PCoF with CIA, we need to find precoding matrices for satisfying the alignment condition in (\ref{cond:ALI}).
For $M=2$, the conditions are given by
\begin{eqnarray*}
\Fm_{11}\vv_{1,2} = \Fm_{12}\vv_{2,1} \mbox{ and } \Fm_{21}\vv_{1,1} = \Fm_{22}\vv_{2,1}.
\end{eqnarray*}
For the simulation, we used the following precoding matrices to satisfy the above conditions:
\begin{eqnarray}
\Vm_{1} &=& \left[
              \begin{array}{cc}
                \Fm_{21}^{-1}\Fm_{22}\onev & \Fm_{11}^{-1}\Fm_{12}\onev \\
              \end{array}
            \right]
\mbox{ and } \vv_{2,1} =\onev.\label{eq:constM2}
\end{eqnarray}
Also, the same construction is used for the second hop.
We used the complex LLL algorithm to optimize integer matrices, yielding lower bound on achievable rate in Theorem~\ref{thm:222}.
Since source 1 (or relay 1) transmits one more stream than source 2 (or relay 2), the former always requires higher transmission power.
In order to efficiently satisfy the average power-constraint, the role of sources 1 and 2 (equivalently, relays 1 and 2) is alternatively reversed in successive time slots.
In Fig.~\ref{simulation}, we observe that PCoF with CIA can have the SNR gain about $5$ dB by optimizing the integer matrices for IFR and IFB,
comparing with simply using identity matrices. Also, PCoF with CIA provides a higher sum rate
than time-sharing if $\SNR \geq 15$ dB, and its gain over time-sharing increases with $\SNR$, showing that in this case the DoF
result matters also at finite SNR.

In addition, we evaluate the performance of the Network-Coded ICC with respect to sum rates.
The achievable sum rate is given in Corollary\ref{thm:122sum}. For comparison, we considered the performance
of {\em time-sharing} where the achievable sum rate is equal to the capacity of the individual (interference free)
MIMO channel with full CSI at both transmitter and receiver, given by
\begin{equation}
R_{{\rm sum}}=\sum_{\ell=1}^{M}\log(1+P_{\ell}\lambda_{\ell}(\Fm_{kk}^{\herm}\Fm_{kk}))
\end{equation}
where $P_{\ell}$ is obtained via water-filling over the eigenvalues of $\Fm_{kk}^{\herm}\Fm_{kk}$,
denoted by $\lambda_{\ell}(\Fm_{kk}^{\herm}\Fm_{kk})$, such as
\begin{equation}
P_{\ell} = \left[\mu-\frac{1}{\lambda_{\ell}(\Fm_{kk}^{\herm}\Fm_{kk})}\right]^{+}
\end{equation}
with $\mu$ is chosen to satisfy the total power constraint $\sum_{\ell=1}^{M}P_{\ell} = 2P_{\rm{sum}}$. Again, with time-sharing
the per-slot power constraint is $2P_{\rm{sum}}$.
For PCoF with CIA, we used the same construction method in (\ref{eq:constM2}) for $M=2$. Also, in case of PCoF and DPC, we used the achievable sum-rates in Theorem~\ref{thm:RateCIC}. Fig.~\ref{simulation_sum} shows that PCoF and DPC provides a higher sum rate than PCoF with CIA and time-sharing, having a larger gap as SNR increases. Also, PCoF with CIA shows the satisfactory performance in the moderate SNRs (i.e., $\SNR \leq 30$ dB). Yet, this scheme suffers from the non-integer penalty at high SNRs. Remarkably, PCoF and DPC performs within a constant gap with respect to full-cooperation at any SNR.


\section{Concluding Remarks}

In this work we have studied two-user Gaussian networks with cognition, coordination, and two hops. We first investigated a cognitive interference channel (CIC) where one of the transmitters (non-cognitive) has knowledge of a rank-1 linear combination of the two information messages, while the other transmitter (cognitive) has access to a rank-2 linear combination of the same messages. This is referred to as the Network-Coded CIC, since such linear combination may be the result of some random linear network coding scheme implemented in the backbone wired network. For such channel, we developed an achievable region based on a few novel concepts: Precoded Compute and Forward (PCoF) with Channel Integer Alignment (CIA), combined with standard Dirty-Paper Coding. We also developed a capacity region outer bound and found the sum symmetric Generalized Degrees of Freedom (GDoF) of the Network-Coded CIC. Through the GDoF characterization, we showed that knowing ``mixed data" (linear combination of the information messages) provides an unbounded spectral efficiency gain over the classical CIC counterpart. Then, we considered a Gaussian relay network having two-user MIMO IC as the fundamental building block. We used PCoF with CIA to convert the MIMO IC into a deterministic finite-field IC. Then, we applied a linear precoding scheme over the finite-field to eliminate the interferences in the finite-field domain. Using this unified approach, we provided finite-SNR results (not just degrees of freedom) which show that the proposed coding schemes are competitive against the state-of-the-art interference avoidance based on orthogonal access, for standard randomly generated Rayleigh fading channels.

In conclusion, we showed that exploiting an algebraic structure is beneficial to manage an interference for some two-user Gaussian networks. It would be an interesting future work to extend this approach (e.g., using an algebraic structure) into general discrete memoryless channels and multihop multiflow networks.


\appendices

\section{Proof of Theorem \ref{thm:GDoFregion}}\label{proof:GDoF}

\subsection{Converse}
For given rates $R_{1}$ and $R_{2}$, we define $R_{{\rm min}} = \min\{R_{1},R_{2}\}$ and $R_{\triangle} = \max\{R_{1},R_{2}\} - R_{{\rm min}}$. If $R_{1} > R_{2}$ then $W_{1}=[W_{11},W_{12}=W_{\triangle}]$ and $W_{2}=[W_{21},W_{22}=\zerov]$. In the reverse case, we have that $W_{1}=[W_{11},W_{12}=\zerov]$
and $W_{2}=[W_{21},W_{22}=W_{\triangle}]$. Notice that $H(W_{11})=H(W_{21}) = R_{{\rm min}}$ and $H(W_{\triangle}) = R_{\triangle}$.
In both cases, the non-cognitive transmitter knows the linear combination, $W_{1}\oplus W_{2} = \left[
                                                                        \begin{array}{cc}
                                                                          W_{11}\oplus W_{21} & W_{\triangle} \\
                                                                        \end{array}
                                                                      \right]$.
From the well-known Crypto Lemma, the $W_{11} \oplus W_{21}$ is mutually statistically independent of $W_{11}$, as well as $W_{11} \oplus W_{21}$ is mutually statistically independent of $W_{21}$.  
First, we derive the upper bound on the minimum rate $R_{{\rm min}}$ equal to $H(W_{11})$ and $H(W_{21})$:
\begin{eqnarray}
n R_{{\rm min}} &=& H(W_{11}) = H(W_{11} | W_{11} \oplus W_{21}, W_{\triangle})\\
&=& H(W_{11}|W_{11} \oplus W_{21}, W_{\triangle})- H(W_{11}|Y_{1}^{n}, W_{11}\oplus W_{21}, W_{\triangle})\\
&&+ H(W_{11}| Y_{1}^{n}, W_{11} \oplus W_{21},W_{\triangle})\\
&\stackrel{(a)}{\leq}& I(W_{11};Y_{1}^{n}|W_{11} \oplus W_{21},W_{\triangle})+n\epsilon_{n}\\
&=& h(Y_{1}^{n}|W_{11}\oplus W_{21},W_{\Delta}) - h(Y_{1}^{n}|W_{11}\oplus W_{21}, W_{\Delta}, W_{11})+n\epsilon_{n}\\
&\stackrel{(b)}{=}& h(Y_{1}^{n}|X_{2}^{n},W_{11}\oplus W_{21},M_{\Delta}) - h(Y_{1}^{n}|X_{1}^{n},X_{2}^{n},W_{11}\oplus W_{21},W_{\Delta}, W_{11})+n\epsilon_{n}\\
&\leq& h(Y_{1}^{n}|X_{2}^{n}) - h(Y_{1}^{n}|X_{1}^{n},X_{2}^{n},W_{11}\oplus W_{21},W_{\Delta}, W_{11})+n\epsilon_{n}\\
&\stackrel{(c)}{=} & h(Y_{1}^{n}|X_{2}^{n}) - h(Y_{1}^{n}|X_{1}^{n},X_{2}^{n})+n\epsilon_{n}\\
&=& I(X_{1}^{n};Y_{1}^{n}|X_{2}^{n}) + n\epsilon_{n}\\
&\leq& n\log(1+|h_{11}|^2\SNR)+n\epsilon_{n}  \label{song1}
\end{eqnarray}where (a) follows from the Fano's inequality and data processing inequality as
\begin{equation}
H(W_{11}|Y_{1}^{n},W_{11} \oplus W_{21},W_{\triangle}) \leq H(W_{11}|Y_{1}^{n})\leq H(W_{11}|\hat{W}_{11})\leq n\epsilon_{n},
\end{equation} (b) follows from the fact that encoder 2 has $(W_{11}\oplus W_{21},W_{\Delta})$, therefore for any coding scheme $X_{2}^{n}$ is a function of $(W_{11}\oplus W_{21},W_{\Delta})$, and (c) follows from the fact that there is a Markov chain
\[ (W_{11}, W_{21}, W_\triangle) \rightarrow (X_1^n,X_2^n) \rightarrow Y_1^n. \]
In the same manner, we get:
\begin{eqnarray}
n R_{{\rm min}} &=& H(W_{21}) = H(W_{21}|W_{11}\oplus W_{21},W_{\triangle})\\
&=& H(W_{21}|W_{11} \oplus W_{21},W_{\triangle}) - H(W_{21}|Y_{2}^{n},W_{11} \oplus W_{21},W_{\triangle}) \nonumber \\
& & + H(W_{21}|Y_{2}^{n}, W_{11}\oplus W_{21},W_{\triangle})\\
&\leq& I(W_{21};Y_{2}^{n}|W_{11} \oplus W_{21},W_{\triangle}) + n\epsilon_{n}\\
& \leq & h(Y_{2}^{n}|X_{2}^{n}) - h(Y_{2}^{n}|X_{1}^{n},X_{2}^{n}) + n \epsilon_{n}\\
&=& I(X_{1}^{n};Y_{2}^{n}|X_{2}^{n}) + n\epsilon_{n}\\
&\leq& n\log(1+|h_{21}|^2\INR) + n\epsilon_{n}. \label{song2}
\end{eqnarray}
From (\ref{song1}) and (\ref{song2}), we have
\begin{equation}\label{eq:upper1}
R_{{\rm min}} \leq \min\{\log(1+|h_{11}|^2\SNR), \log(1+|h_{21}|^2\INR)\}.
\end{equation}
An obvious upper bound on $R_{1}$ and $R_{2}$ are given by
\begin{eqnarray}
n R_{1} &\leq&
I(X_{1}^{n},X_{2}^{n};Y_{k}^{n}) + n\epsilon_{n}\\
&\leq&  n\log(1+|h_{1 1}|^2\SNR + |h_{1 2}|^2\INR) +n\epsilon_{n}\label{eq:bound1}\\
n R_{2} &\leq&
I(X_{1}^{n},X_{2}^{n};Y_{2}^{n}) + n\epsilon_{n}\\
&\leq&  n\log(1+|h_{2 1}|^2\INR + |h_{2 2}|^2\SNR) +n\epsilon_{n}\label{eq:bound2}.
\end{eqnarray} Using (\ref{eq:bound1}), (\ref{eq:bound2}), and $\INR=\SNR^{\rho}$, we have:
\begin{eqnarray}
d_{1}(\rho) &=& \lim_{\SNR \rightarrow \infty} \frac{R_{1}}{\log\SNR} \leq \max\{1,\rho\}\\
d_{2}(\rho) &=& \lim_{\SNR \rightarrow \infty} \frac{R_{2}}{\log\SNR} \leq \max\{1,\rho\}.
\end{eqnarray}  Also, from (\ref{eq:bound1}) and (\ref{eq:bound2}), we have:
\begin{eqnarray}\label{eq:upper2}
R_{{\rm max}} &\leq& \max\{\log(1+|h_{11}|^2\SNR+|h_{12}|^2\INR),\nonumber\\
&&\log(1+|h_{21}|^2\INR+|h_{22}|^2\SNR)\}.
\end{eqnarray}
Using (\ref{eq:upper1}), (\ref{eq:upper2}), and $\INR = \SNR^{\rho}$, we have the upper bounds in the asymptotic case:
\begin{eqnarray}
\lim_{\SNR \rightarrow \infty}\Big(\frac{R_{{\rm min}}}{\log\SNR} + \frac{R_{{\rm max}}}{\log\SNR}\Big)&\leq& \min\{1,\rho\}+\max\{1,\rho\},
\end{eqnarray}
yielding the upper bound on the sum symmetric GDoF as
\begin{equation}
d_{\mbox{\tiny{sum}}}(\rho)=\lim_{\SNR \rightarrow \infty} \frac{R_{{\rm min}}+R_{{\rm max}}}{\log\SNR} \leq 1 + \rho.
\end{equation}

\subsection{Achievable scheme}

We will present coding schemes to achieve two corner points $(d_{1}(\rho),d_{2}(\rho)) = (1,\rho)$ and $(d_{1}(\rho),d_{2}(\rho)) = (\rho,1)$. One coding scheme achieves the first corner point (see Section~\ref{subsec:proof1}) and the other two coding schemes achieve the second corner point $(d_{1}(\rho),d_{2}(\rho)) = (\rho,1)$ depending on the interference level $\rho$ (see Section~\ref{subsec:proof2} for $\rho < 1$ and see Section~\ref{subsec:proof3} for $\rho \geq 1$).

\subsubsection{$(d_{1}(\rho),d_{2}(\rho)) = (1,\rho)$}\label{subsec:proof1}

We use the achievable rates given in Theorem \ref{thm:DPC}. It is immediately shown that the achievable GDoF of message 1 (cognitive user), obtained by
\begin{equation}\label{eq:GDoF1}
d_{1}(\rho) = \lim_{\SNR \rightarrow \infty} \frac{\log(1+|h_{11}|^2\SNR)}{\log\SNR} = 1.
\end{equation} In this proof, we show that message 2 (non-cognitive user) achieves GDoF equal to $\rho$,
by carefully choosing the power scaling factor $\beta \in \Pc$. The effective channel for Scaled PCoF is given by
$\tilde{\hv}(\beta) = [h_{21}\sqrt{\INR}, \beta(h_{22}\sqrt{\SNR}-\alpha_{1,\mbox{\tiny{MMSE}}}(h_{12}h_{21}/h_{11})\SNR^{\rho-\frac{1}{2}})]$ and can be rewritten as
\begin{equation}
\tilde{\hv}(\beta) = \SNR^{\rho/2}[h_{21},\beta\tilde{h}_{22}]
\end{equation} where $\tilde{h}_{22} = h_{22}\SNR^{(1-\rho)/2} - h \SNR^{(\rho - 1)/2}$ and $h=\alpha_{1,\mbox{\tiny{MMSE}}}(h_{12}h_{21}/h_{11})$.
Here, we choose
$\beta=\beta^{\star} \triangleq h_{21}/(\tilde{h}_{22}\gamma)$, where $\gamma \geq 1$ is an integer with
$\gamma=\lceil|h_{21}/\tilde{h}_{22}|\rceil \in \ZZ_{+}$.
This produces a kind of ``aligned" channel:
\begin{equation}
\tilde{\hv}= \SNR^{\rho/2}[h_{21},h_{21}/\gamma].
\end{equation}
Letting $b_{1}=\gamma$, and $b_{2}=1$, the effective noise in (\ref{eq:enoise}) is obtained by
\begin{eqnarray}
\zv_{\mbox{\tiny{eff}}}(\tilde{\hv},\bv) = \frac{\gamma}{h_{21}\SNR^{\rho/2}}\underline{\zv}_{2},
\end{eqnarray}{\BLUE and accordingly, its variance is given by
\begin{equation}
\sigma_{{\rm eff}}^2(\beta^{\star}) = \frac{\gamma^2}{|h_{21}|^2\SNR^{\rho}}.\label{eq:effnoise_proof}
\end{equation}}
This shows that non-integer penalty is completely eliminated. Also, we can use the zero forcing precoding over $\FF_{q}$ since the chosen integer coefficients $b_{1}=\gamma$ and $b_{2}=1$ are non-zero. {\BLUE Plugging (\ref{eq:effnoise_proof}) and transmit power 1 into (\ref{eq:R2_proof}), we have the lower bound on the achievable rate of Scaled PCoF:
\begin{equation}
\max_{\beta}R_{2}(\beta) \geq R_{2}(\beta^{*}) =\log\left(\frac{1}{\sigma_{{\rm eff}}^2(\beta^{\star})}\right)= \rho\log(|h_{21}|^2\SNR) - 2\log(\gamma).
\end{equation}} The lower and upper bounds on $\gamma$ is given by
\begin{eqnarray}
1\leq \gamma \leq 1+\left|\frac{h_{21}}{h_{22}\SNR^{(1-\rho)/2}-h\SNR^{(\rho-1)/2}}\right|
\end{eqnarray} where $\gamma$ converges to a constant as $\SNR \rightarrow \infty$. Finally, the achievable GDoF of the non-cognitive transmitter is derived as
\begin{equation}\label{eq:GDoF2}
d_{2}(\rho)  \geq \lim_{\SNR,\INR \rightarrow \infty}\frac{R_{2}(\beta^{*})}{\log\SNR}= \rho.
\end{equation}From (\ref{eq:GDoF1}) and (\ref{eq:GDoF2}), the proposed scheme can achieve the corner point $(d_{1}(\rho),d_{2}(\rho))=(1,\rho)$.

\subsubsection{$(d_{1}(\rho),d_{2}(\rho) = (\rho,1)$, $\rho < 1$}\label{subsec:proof2}

Since $R_{1} < R_{2}$, the user messages have the following form:
\begin{eqnarray}
\underline{\wv}_{1} = [\underline{\wv}_{11}, \underline{\wv}_{12}=0]\mbox{ and }\underline{\wv}_{2} = [\underline{\wv}_{21}, \underline{\wv}_{22}],
\end{eqnarray} where notice that $R_{1} = R_{11}$, $R_{2}=R_{21}+R_{22}$, and $R_{11}=R_{21}$. Accordingly, transmitter 1 knows $\underline{\wv}_{11}$, $\underline{\wv}_{21}$, and $\underline{\wv}_{22}$, and transmitter 2 knows $\underline{\wv}_{11} \oplus \underline{\wv}_{21}$ and $\underline{\wv}_{22}$. We let $\bv=[b_{1},b_{2}] \in \ZZ[j]^{2}$ denote the integer coefficients vector used at receiver 2 for the CoF receiver mapping (\ref{eq:cof}), and we let $q_{k} = [b_{k}]_{q}$.
Again, it is assumed that  $q_{1},q_{2} \neq 0$ over $\FF_q$. The proposed achievable scheme proceeds as follows:
\begin{itemize}
\item Transmitter 2 produces the lattice codewords $\underline{\vv}_{21} = f(\underline{\wv}_{11} \oplus \underline{\wv}_{21})$ and $\underline{\vv}_{22} = f(\underline{\wv}_{22})$. Then, it transmits the channel input:
\begin{equation*}
\underline{\xv}_{2} = \sqrt{ \SNR^{\rho -1}} \underline{\xv}_{21}+ \sqrt{1- \SNR^{\rho -1}}\underline{\xv}_{22},
\end{equation*} where $\underline{\xv}_{21} =\beta\underline{\xv}_{21}'$ with power scaling factor $\beta \in \CC$ with $|\beta|=1$, $\underline{\xv}'_{21} =[\underline{\vv}_{21} + \underline{\dv}_{21}] \mod \Lambda$, and $\underline{\xv}_{22} = [\underline{\vv}_{22} + \underline{\dv}_{22}] \mod \Lambda$.
\item Transmitter 1 performs the DPC encoding as in Section~\ref{subsec:proof1}:
\begin{equation*}
\underline{\xv}_{1} = [\underline{\vv}_{1} - (h_{12}/h_{11})\sqrt{\SNR^{\rho-1}}\underline{\xv}_{2} + \underline{\dv}_{1}] \mod \Lambda,
\end{equation*}
where $\underline{\vv}_{1} = f(m\underline{\wv}_{11})$ with $m=(q_{1})^{-1}(-q_{2})$.
\end{itemize}

From the standard DPC result (see Section~\ref{sec:CIC-GA}), the coding rate $R_{11}$ is achievable if
\begin{equation}
R_{11} \leq \log(|h_{11}|^2\SNR).
\end{equation} Also, receiver 2 observes:
\begin{eqnarray}
\underline{\yv}_{2} &=&h_{21} \sqrt{\INR}\underline{\xv}_{1} + h_{22}\sqrt{\SNR}\underline{\xv}_{2}+\underline{\zv}_{2}\\
&=& h_{21}\sqrt{\INR}\underline{\xv}_{1} + h_{22}\sqrt{\INR}\underline{\xv}_{21}+ h_{22}\sqrt{\SNR - \SNR^{\rho}}\underline{\xv}_{22} + \underline{\zv}_{2}.
\end{eqnarray} Treating the undesired signals $\underline{\xv}_{1}$ and $\underline{\xv}_{21}$ as noise, receiver 2 is able to decode the message $\underline{\wv}_{22}$ if
\begin{equation}\label{eq:R22}
R_{22} \leq \log\left(1+\frac{|h_{22}|^2(\SNR-\SNR^{\rho})}{1+(|h_{21}|^2+|h_{22}|^2)\SNR^{\rho}}\right).
\end{equation} Subtracting the decoded signal $\underline{\xv}_{22}$ from $\underline{\yv}_{2}$, receiver 2 has:
\begin{equation}
\underline{\yv}'_{2} = h_{21}\sqrt{\INR}(\underline{\vv}_{1} + \underline{\dv}_{1} + \lambdav) + \tilde{h}_{22}\underline{\xv}_{21} + \underline{\zv}_{2},
\end{equation} where $\tilde{h}_{22}= \sqrt{\SNR^{\rho}}h_{22} - \sqrt{\SNR^{3\rho-2}}h_{12}h_{21}/h_{11}$ and $\lambdav = Q_{\Lambda}(\underline{\vv}_{1} - (h_{12}/h_{11})\sqrt{\SNR^{\rho-1}}\underline{\xv}_{2} + \underline{\dv}_{1})$. Receiver 2 applies the CoF receiver mapping in (\ref{eq:cof}) with integer coefficients vector $\bv$ and scaling factor $\alpha_{2} =\frac{b_{1}}{h_{21}\sqrt{\INR}}$, yielding
\begin{eqnarray*}
\hat{\underline{\yv}'}_{2} &=& [\alpha_{2}\underline{\yv}_{2}' - b_{1}\underline{\dv}_{1} - b_{2}\underline{\dv}_{21}] \mod \Lambda\\
&=& \left[b_{1}\underline{\vv}_{1}+b_{2}\underline{\vv}_{21}+(\alpha_{2}\beta\tilde{h}_{22}-b_{2})\underline{\xv}_{21}'+\alpha_{2}h_{21}\sqrt{\INR}\lambdav+\alpha_{2}\underline{\zv}_{2}\right] \mod \Lambda\\
&\stackrel{(a)}{=} & \left[b_{1}\underline{\vv}_{1}+b_{2}\underline{\vv}_{21}+(b_{1}\beta\tilde{h}_{22}/(h_{21}\sqrt{\INR})-b_{2})\underline{\xv}'_{21}+(b_{1}/(h_{21}\sqrt{\INR}))\underline{\zv}_{2}\right] \mod \Lambda\\
&\stackrel{(b)}{=} & \left[([b_{2}] \mod p\ZZ[j])f(\underline{\wv}_{21})+(b_{1}\beta\tilde{h}_{22}/(h_{21}\sqrt{\INR})-b_{2})\underline{\xv}'_{21}+(b_{1}/(h_{21}\sqrt{\INR}))\underline{\zv}_{2}\right] \mod \Lambda\\
&\stackrel{(c)}{=} & \left[([b_{2}] \mod p\ZZ[j])f(\underline{\wv}_{21})+(b_{1}/(h_{21}\sqrt{\INR}))\underline{\zv}_{2}\right] \mod \Lambda,
\end{eqnarray*} where $\lambdav = Q_{\Lambda}(\underline{\vv}_{1} - (h_{12}/h_{11})\sqrt{\SNR^{\rho-1}}\underline{\xv}_{2} + \underline{\dv}_{1})$, (a) is due to the fact that $\alpha_{2}h_{21}\sqrt{\INR}\lambdav = b_{1}\lambdav \in \Lambda$, (b) follows from the fact that $m$ is chosen such that $[b_{1}g(m)+b_{2}] \mod p\ZZ[j] = 0$, and (c) is due to the fact that $\beta = h_{21}\sqrt{\INR}/(\tilde{h}_{22}b_{1})$, $b_{1} = \left\lceil|h_{21}\sqrt{\INR}/\tilde{h}_{22}\right\rceil$, and $b_{2}=1$. Then, receiver 2 can reliably decode the message $\underline{\wv}_{21}$ if
\begin{equation}
R_{21} \leq \log^{+}\left(\frac{|h_{21}|^2\INR}{b_{1}^2}\right)
\end{equation} where notice that $b_{1}$ is a constant when $\SNR \rightarrow \infty$. Since $R_{11}=R_{21}$, we have:
\begin{eqnarray}
R_{11}=R_{21} &\leq& \min\left\{\log(|h_{11}|^2\SNR),\log^{+}\left(\frac{|h_{21}|^2\INR}{b_{1}^2}\right)\right\}\nonumber\\
&=& \log^{+}\left(\frac{|h_{21}|^2\INR}{b_{1}^2}\right)\label{eq:R21}.
\end{eqnarray} Using (\ref{eq:R21}) and (\ref{eq:R22}), we have:
\begin{eqnarray}
d_{1}(\rho) &=& \lim_{\SNR \rightarrow \infty} \frac{R_{11}}{\log{\SNR}} = \rho\\
d_{2}(\rho) &=& \lim_{\SNR \rightarrow \infty} \frac{R_{21}+R_{22}}{\log{\SNR}} = 1.
\end{eqnarray}

\subsubsection{$(d_{1}(\rho),d_{2}(\rho) = (\rho,1)$, $\rho \geq 1$}\label{subsec:proof3}

Since $R_{1} \geq R_{2}$, the user messages have the following form:
\begin{eqnarray}
\underline{\wv}_{1} = [\underline{\wv}_{11}, \underline{\wv}_{12}]\mbox{ and }\underline{\wv}_{2} = [\underline{\wv}_{21}, \underline{\wv}_{22}=0],
\end{eqnarray} where notice that $R_{1} = R_{11}+R_{12}$, $R_{2}=R_{21}$, and $R_{11}=R_{21}$. Accordingly, transmitter 1 knows $\underline{\wv}_{11}$, $\underline{\wv}_{12}$, and $\underline{\wv}_{21}$, and transmitter 2 knows $\underline{\wv}_{11} \oplus \underline{\wv}_{21}$ and $\underline{\wv}_{12}$. We let $\bv=[b_{1},b_{2}] \in \ZZ[j]^{2}$ denote the integer coefficients vector used at receiver 1 for the CoF receiver mapping (\ref{eq:cof}), and we let $q_{k} = [b_{k}]_{q}$.
Again, it is assumed that  $q_{1},q_{2} \neq 0$ over $\FF_q$. The proposed achievable scheme proceeds as follows:
\begin{itemize}
\item Transmitter 2 produces the lattice codewords $\underline{\vv}_{21} = f(\underline{\wv}_{11} \oplus \underline{\wv}_{21})$ and $\underline{\vv}_{22} = f(\underline{\wv}_{12})$. Then, it transmits the channel input:
\begin{equation}
\underline{\xv}_{2} = \sqrt{ \SNR^{1-\rho}} \underline{\xv}_{21}+ \sqrt{1- \SNR^{1-\rho}}\underline{\xv}_{22}
\end{equation} where $\underline{\xv}_{21} = \beta \underline{\xv}_{21}'$ with power scaling factor $\beta \in \CC$ with $|\beta|=1$, $\underline{\xv}_{21}' = [\underline{\vv}_{21} + \underline{\dv}_{21}] \mod \Lambda$, and $\underline{\xv}_{22} = [\underline{\vv}_{22} + \underline{\dv}_{22}] \mod \Lambda$.
\item Transmitter 1 performs the DPC encoding as in Section~\ref{subsec:proof1} with the {\em primary} user message $\underline{\wv}_{21}$:
\begin{equation}
\underline{\xv}_{1} = [\underline{\vv}_{1} - (h_{22}/h_{21})\underline{\xv}_{2} + \underline{\dv}_{1}] \mod \Lambda,
\end{equation}
where $\underline{\vv}_{1} = f(m\underline{\wv}_{21})$ with $m=q_{1}^{-1}(-q_{2})$. Differently from the previous coding schemes, DPC is performed to cancel the known interference at receiver 2.
\end{itemize}

Thanks to DPC, receiver 2 can reliably decode the message $\underline{\wv}_{21}$ if
\begin{equation}\label{eq:R21'}
R_{21} \leq \log(|h_{12}|^2\INR).
\end{equation} Also, receiver 1 observes:
\begin{eqnarray}
\underline{\yv}_{1} &=&h_{11} \sqrt{\SNR}\underline{\xv}_{1}+h_{12} \sqrt{\INR}\underline{\xv}_{2}+\underline{\zv}_{1}\\
&=& h_{11}\sqrt{\SNR}\underline{\xv}_{1} + h_{12}\sqrt{\SNR}\underline{\xv}_{21}+h_{12} \sqrt{\SNR^{\rho}-\SNR}\underline{\xv}_{22}+\underline{\zv}_{1}.
\end{eqnarray}Treating the undesired signals $\underline{\xv}_{1}$ and $\underline{\xv}_{21}$ as noise, receiver 1 can reliably decode the message $\underline{\wv}_{12}$ if
\begin{equation}\label{eq:R12'}
R_{12} \leq \log\left(1+\frac{|h_{21}|^2(\SNR^{\rho}-\SNR)}{1+(|h_{11}|^2+|h_{12}|^2)\SNR}\right).
\end{equation} Subtracting the known signal $\underline{\xv}_{22}$, receiver 1 has:
\begin{equation}
\underline{\yv}_{1}' = h_{11}\sqrt{\SNR}(\underline{\vv}_{1}+\underline{\dv}_{1} + \lambdav) + \tilde{h}_{12}\underline{\xv}_{21} + \underline{\zv}_{1},
\end{equation}where $\tilde{h}_{12} = h_{12}\sqrt{\SNR} - \sqrt{\SNR^{2-\rho}}h_{11}h_{22}/h_{21}$ and $\lambdav = Q_{\Lambda}(\underline{\vv}_{1} - (h_{22}/h_{21})\underline{\xv}_{2} + \underline{\dv}_{1})$. Receiver 1 applies the CoF receiver mapping in (\ref{eq:cof}) with integer coefficients vector $\bv$ and scaling factor $\alpha_{1} =\frac{b_{1}}{h_{11}\sqrt{\SNR}}$:
\begin{eqnarray*}
\hat{\underline{\yv}'}_{1} &=&[\alpha_{1}\underline{\yv}'_{1} - b_{1}\underline{\dv}_{1}-b_{2}\underline{\dv}_{21}] \mod \Lambda\\
&=& \left[b_{1}\underline{\vv}_{1}+b_{2}\underline{\vv}_{21} + (b_{1}\beta\tilde{h}_{12}/(h_{11}\sqrt{\SNR})-b_{2})\underline{\xv}'_{21}+ (b_{1}/(h_{11}\sqrt{\SNR}))\underline{\zv}_{1}\right] \mod \Lambda\\
&\stackrel{(a)}{=}& \left[([b_{2}] \mod p\ZZ[j])f(\underline{\wv}_{11}) + (b_{1}\beta\tilde{h}_{12}/(h_{11}\sqrt{\SNR})-b_{2})\underline{\xv}'_{21}+ (b_{1}/(h_{11}\sqrt{\SNR}))\underline{\zv}_{1}\right] \mod \Lambda\\
&\stackrel{(b)}{=}& \left[([b_{2}] \mod p\ZZ[j])f(\underline{\wv}_{11}) + (b_{1}/(h_{11}\sqrt{\SNR}))\underline{\zv}_{1}\right] \mod \Lambda,
\end{eqnarray*} where (a) follows from the fact that $m$ is chosen such that $[b_{1}g(m)+b_{2}] \mod p\ZZ[j] = 0$ and (b) follows from the fact that $\beta = h_{11}\sqrt{\SNR}/(\tilde{h}_{12}b_{1})$, $b_{1}=\left\lceil|h_{11}\sqrt{\SNR}/\tilde{h}_{12}|\right\rceil$, and $b_{2}=1$. Then, receiver 1 can reliably decode the message $\underline{\wv}_{11}$ if
\begin{equation}\label{eq:R11'}
R_{11} \leq \log^{+}\left(\frac{|h_{11}|^2\SNR}{b_{1}^2}\right),
\end{equation}where notice that $b_{1}$ is a constant when $\SNR \rightarrow \infty$. Since $R_{11}=R_{21}$, we have that
\begin{equation}\label{eq:rate}
R_{11}=R_{21} = \min\left\{\log^{+}\left(\frac{|h_{11}|^2\SNR}{b_{1}^2}\right), \log(1+|h_{12}|^2\INR)\right\} = \log^{+}\left(\frac{|h_{11}|^2\SNR}{b_{1}^2}\right).
\end{equation} From (\ref{eq:R12'}) and (\ref{eq:rate}), we can get:
\begin{eqnarray}
d_{1}(\rho) &=& \lim_{\SNR\rightarrow\infty}\frac{R_{11}+R_{12}}{\log\SNR} = \rho\\
d_{2}(\rho) &=& \lim_{\SNR\rightarrow\infty}\frac{R_{21}}{\log\SNR} = 1.
\end{eqnarray}This completes the proof.

\section{Proof of Corollary~\ref{cor:DoF}}\label{proof:DoF}

For the DoF proof, we assume that $P_{\rm{sum}}$ goes to infinity and equivalently, $\SNR$ goes to infinity.  In this proof, we will show that the
individual messages rate $R$ (equal for all messages) grows as $\log{\SNR}$, i.e.,
\begin{equation}
\lim_{\SNR\rightarrow \infty}\frac{R}{\log{\SNR}} = 1.
\end{equation}  Assuming that we use exact IFR (see Section~\ref{subsec:CoF}, eq. (\ref{exact-IFR})), we have:
\begin{equation}
R_{\rm{comp}}(\Hm_{k}\Cm_{k},\Bm_{k},\SNR) \geq \log(\SNR)  - \max_{\ell}\log\left(\|(\Hm_{k}^{-1})^{\herm}\bv_{k,\ell}\|^2\right).
\end{equation} This definitely shows that  $\lim_{\SNR\rightarrow \infty}\frac{R_{\rm{comp}}(\Hm_{k}\Cm_{k},\Bm_{k},\SNR)}{\log{\SNR}} = 1$. Accordingly, $R$ grows as $\log{\SNR}$. However, $\Hm_{k}$ must be full-rank in order to allow exact IFR.
Since $\Hm_{k} = \Fm_{k1}\Vm_{1}$ for $k=1,2$ and $\Hm_{k} = \Fm_{k3}\Vm_{3}$ for $k=3,4$, we need to show that
$\Vm_{1}$ and $\Vm_{3}$ are full rank.  For the alignment, we use the following construction method proposed in \cite{Gou}:
\begin{eqnarray}
\vv_{1,\ell+1} &=& (\Fm_{11}^{-1}\Fm_{12}\Fm_{22}^{-1}\Fm_{21})^{\ell}\vv_{1,1}\\
\vv_{2,\ell} &=& (\Fm_{22}^{-1}\Fm_{21}\Fm_{11}^{-1}\Fm_{12})^{\ell-1}\Fm_{22}^{-1}\Fm_{21}\vv_{1,1}\label{ali:const}
\end{eqnarray}for $\ell=1,\ldots,M-1$. Once $\vv_{1,1}$ is determined, other vectors are completely determined by the above equations.
As argued in \cite{Gou}, since the channel matrices are drawn form a continuous distribution then
$\Fm\triangleq(\Fm_{11}^{-1}\Fm_{12}\Fm_{22}^{-1}\Fm_{21})$ has all distinct eigenvalues almost surely. From \cite[Thm 1.3.9]{Horn}, $\Fm$ is diagonalizable such as
\begin{equation*}
\Fm = \Em\left[
           \begin{array}{ccc}
             \lambda_{1} &  &  \\
              & \ddots &  \\
              &  & \lambda_{M} \\
           \end{array}
         \right]\Em^{-1}
\end{equation*} where the $i$-th column of $\Em$ is an eigenvector of $\Fm$ associated with $\lambda_{i}$.
Choosing $\vv_{1,1}= \Em \onev$, where $\onev$ denotes the all 1's vector, the alignment precoding matrix $\Vm_{1}$ can be rewritten as
\begin{equation}
\Vm_{1}=\Em
\underbrace{\left[
         \begin{array}{cccc}
           1 & \lambda_{1} & \cdots & \lambda_{1}^{M-1} \\
           \vdots & \vdots & \ddots & \vdots \\
           1 & \lambda_{M} & \cdots & \lambda_{M}^{M-1} \\
         \end{array}
       \right]}_{\eqdef \Jm}
\end{equation} where $\Jm$ denotes the Vandermonde matrix. Therefore, the determinant of the $\Vm_{1}$ is computed by
\begin{eqnarray}
\det(\Vm_{1}) &=& \det(\Em)\det(\Jm)\\
&=&\det(\Em) \prod_{1\leq i < j \leq M}(\lambda_{j} - \lambda_{i}) \neq 0.
\end{eqnarray}
This shows that $\Vm_{1}$ is full rank . With the same procedure, we can show that $\Vm_{3}$ is full rank.

\section{Proof of Theorem~\ref{thm:DoFCIC}}\label{proof:MIMOCIC}

We prove that PCoF with CIA and DPC can achieve the optimal $2M$ DoF.
This scheme can be regarded as MIMO extension of Scaled PCoF and DPC  proposed
 in Section~\ref{sec:CIC-GA}. Consider the MIMO IC in (\ref{model:brn}).
 Let $\{\underline{\wv}_{k,\ell}: \ell=1,\ldots,M\}$ denote the independent messages to be intended for receiver $k$, for $k=1,2$. Without loss of generality, it is assumed that  transmitter 2 knows $\underline{\Wm}_{1} \oplus \underline{\Wm}_{2}$.
Our achievable scheme proceeds as follows.

\textbf{Encoding:}
\begin{itemize}
\item Transmitter 2 independently produces the $M$ lattice codewords $\underline{\tv}_{2,\ell} = f(\underline{\wv}_{1,\ell} \oplus \underline{\wv}_{2,\ell})$ for $\ell=1,\ldots,M$ and transmits the channel input:
\begin{equation}
\underline{\Xm}_{2} = \Vm\underline{\Xm}_{2}',
\end{equation} for some $\Vm \in \CC^{M \times M}$, where $\underline{\Xm}_{2}' = [\underline{\Tm}_{2} + \underline{\Dm}_{2}] \mod \Lambda$.
\item Transmitter 1 performs the DPC using the known interference signal $\Fm_{12}\underline{\Xm}_{2}$  to get:
\begin{equation}
\underline{\Xm}_{1}=\left[\underline{\Tm}_{1} - \Fm_{11}^{-1}\Fm_{12}\underline{\Xm}_{2} + \underline{\Dm}_{1}\right] \mod \Lambda
\end{equation} where  $\underline{\Tm}_{1}=f(-\underline{\Wm}_{1})$ denotes the lattice codewords corresponding to precoded messages $-\underline{\Wm}_{1}$.
\end{itemize}

\textbf{Decoding:}
\begin{itemize}
\item Receiver 1 performs the modulo-lattice mapping as $\hat{\underline{\Ym}}_{1} = [\Fm_{11}^{-1}\underline{\Ym}_{1} - \underline{\Dm}_{1}] \mod \Lambda$ that yields:
\begin{eqnarray*}
\hat{\underline{\Ym}}_{1} &=& [\underline{\Tm}_{1} - \underline{\Tm}_{1} +\underline{\Xm}_{1} + \Fm_{11}^{-1}\Fm_{12}\underline{\Xm}_{2} - \underline{\Dm}_{1} + \Fm_{11}^{-1}\underline{\Zm}_{1}] \mod \Lambda\\
&=&[\underline{\Tm}_{1} + \underline{\Xm}_{1} - \left([\underline{\Tm}_{1}-\Fm_{11}^{-1}\Fm_{12}\underline{\Xm}_{2}+\underline{\Dm}_{1} ] \mod \Lambda\right) +
\Fm_{11}^{-1}\underline{\Zm}_{1}] \mod \Lambda\\
&=& [\underline{\Tm}_{1} + \Fm_{11}^{-1}\underline{\Zm}_{1}] \mod \Lambda,
\end{eqnarray*} where the last equality is due to the fact that $\underline{\Xm}_{1}=\left[\underline{\Tm}_{1} - \Fm_{11}^{-1}\Fm_{12}\underline{\Xm}_{2} + \underline{\Dm}_{1}\right] \mod \Lambda$.  This shows that receiver 1 has interference-free channel, thus achieving the $M$ DoF. Since we have $M$ parallel point-to-point channels, the following sum-rate is achievable:
\begin{equation}
R_{{\rm sum},1}=\sum_{\ell=1}^{M} \log\left(1+\frac{\SNR}{\|\Fm_{11}^{-1}(\ell)\|^2}\right).
\end{equation}
\item Receiver 2 applies the CoF receiver mapping (\ref{eq:cof}) with integer coefficients $\Bm^{\herm}=\left[
                                \begin{array}{cc}
                                  \Id & \Id \\
                                \end{array}
                              \right]$ and scaling factor $\Fm_{21}^{-1}$, yielding
\begin{eqnarray*}
\hat{\underline{\Ym}}_{2} &=&\left[\Fm_{21}^{-1}\underline{\Ym}_{2} - \underline{\Dm}_{1} - \underline{\Dm}_{2}\right]\\
&=& \left[\underline{\Tm}_{1}+\underline{\Tm}_{2}
 + ((\Fm_{21}^{-1}\Fm_{22} - \Fm_{11}^{-1}\Fm_{12})\Vm-\Id)\underline{\Xm}_{2} '+ \Fm_{21}^{-1}\underline{\Zm}_{2}\right] \mod \Lambda\\
&\stackrel{(a)}{=}& [\underline{\Tm}_{1}+\underline{\Tm}_{2}+\Fm_{21}^{-1}\underline{\Zm}_{2}] \mod \Lambda\\
&\stackrel{(b)}{=}&[f(\underline{\Wm}_{2}) + \Fm_{21}^{-1}\underline{\Zm}_{2}] \mod \Lambda,
\end{eqnarray*}where  $(a)$ is due to the fact that the precoding matrix is chosen as $\Vm=(\Fm_{21}^{-1}\Fm_{22} - \Fm_{11}^{-1}\Fm_{12})^{-1}$ and $(b)$ follows the precoding over $\FF_{q}$ using the coefficient $(-1)$ at the cognitive transmitter .  Then, receiver 2 can achieve the $M$ DoF. Also, from $M$ parallel point-to-point channels, we can achieve the sum-rate of
\begin{equation}
R_{{\rm sum},2}= \sum_{\ell=1}^{M}\log\left(1+\frac{\SNR}{\|\Fm_{21}^{-1}(\ell)\|^2\|\Vm(\ell)\|^2}\right).
\end{equation}
\end{itemize} The proof is done from $R_{{\rm sum}} = R_{{\rm sum},1} + R_{{\rm sum},2}$.

\section*{Acknowledgment}

This work was partially supported by NSF Grant CCF 1161801 and by a collaborative project with ETRI.


\end{document}